\documentclass[reqno,fleqn]{amsart}
\usepackage[T1]{fontenc}
\usepackage[latin9]{inputenc}
\usepackage{enumitem}
\usepackage{amstext}
\usepackage{amsthm}
\usepackage{amssymb}
\usepackage{marvosym}

\makeatletter
\numberwithin{equation}{section}
\numberwithin{figure}{section}
\theoremstyle{plain}
\newtheorem{thm}{\protect\theoremname}[section]
  \theoremstyle{plain}
  \newtheorem{fact}[thm]{\protect\factname}
  \theoremstyle{plain}
  \newtheorem{prop}[thm]{\protect\propositionname}
  \theoremstyle{plain}
  \newtheorem{lem}[thm]{\protect\lemmaname}
  \theoremstyle{plain}
  \newtheorem{cor}[thm]{\protect\corollaryname}
  \theoremstyle{plain}
  \newtheorem*{thm*}{\protect\theoremname}
  \theoremstyle{remark}
  \newtheorem{claim}[thm]{\protect\claimname}

\usepackage{amsmath}
\usepackage{sherstov}
\usepackage{small-caption}
\usepackage{xcolor}
\usepackage{setspace}
\usepackage{centernot}
\usepackage{bbm}

\newcommand{\moo}{\{-1,+1\}}

\usepackage{hyperref}
\hypersetup{
pdfauthor = {A. A. Sherstov},
plainpages = false,
bookmarks = true,
bookmarksopen = true,
pdftex,
colorlinks = true,
citecolor = teal,
linkcolor = teal,
urlcolor  = teal,
}

\newtheoremstyle{myplain}      {10pt}{10pt}{\itshape}{}{\scshape}{.}{.5em}{}
\newtheoremstyle{mydefinition} {10pt}{10pt}{}{}{\scshape}{.}{.5em}{}
\newtheoremstyle{myremark} {10pt}{10pt}{}{}{\itshape}{.}{.5em}{}

\@namedef{thm}{\@thm{\let \thm@swap \@gobble \th@myplain }{thm}{Theorem}}
\@namedef{lem}{\@thm{\let \thm@swap \@gobble \th@myplain }{thm}{Lemma}}
\@namedef{cor}{\@thm{\let \thm@swap \@gobble \th@myplain }{thm}{Corollary}}
\@namedef{prop}{\@thm{\let \thm@swap \@gobble \th@myplain }{thm}{Proposition}}
\@namedef{claim}{\@thm{\let \thm@swap \@gobble \th@myplain }{thm}{Claim}}
\@namedef{fact}{\@thm{\let \thm@swap \@gobble \th@myplain }{thm}{Fact}}
\@namedef{defn}{\@thm{\let \thm@swap \@gobble \th@mydefinition }{thm}{Definition}}
\@namedef{rem}{\@thm{\let \thm@swap \@gobble \th@mydefinition }{thm}{Remark}}
\@namedef{example}{\@thm{\let \thm@swap \@gobble \th@mydefinition }{thm}{Example}}

\@namedef{thm*}{\@thm {\th@myplain }{}{Theorem}}
\@namedef{lem*}{\@thm {\th@myplain }{}{Lemma}}
\@namedef{cor*}{\@thm {\th@myplain }{}{Corollary}}
\@namedef{prop*}{\@thm {\th@myplain }{}{Proposition}}
\@namedef{claim*}{\@thm {\th@myplain }{}{Claim}}
\@namedef{fact*}{\@thm {\th@myplain }{}{Fact}}
\@namedef{defn*}{\@thm {\th@mydefinition }{}{Definition}}
\@namedef{rem*}{\@thm {\th@mydefinition }{}{Remark}}
\@namedef{example*}{\@thm {\th@mydefinition }{}{Example}}

\renewcommand{\mathcal}[1]{\mathscr{#1}}

\makeatletter
\def\@seccntformat#1{%
  \protect\textup{%
    \protect\@secnumfont
    \expandafter\protect\csname format#1\endcsname 
    \csname the#1\endcsname
    \protect\@secnumpunct
  }%
}

\makeatletter
\newcommand \SparseDotfill {\leavevmode \leaders \hb@xt@ .7em{\hss .\hss }\hfill \kern \z@}
\makeatother	

\makeatletter
\def\@tocline#1#2#3#4#5#6#7{\relax
  \ifnum #1>\c@tocdepth 
  \else
    \par \addpenalty\@secpenalty\addvspace{\ifnum #1=1 2mm \else #2\fi}%
    \begingroup \hyphenpenalty\@M
    \@ifempty{#4}{%
      \@tempdima\csname r@tocindent\number#1\endcsname\relax
    }{%
      \@tempdima#4\relax
    }%
    \parindent\z@ \leftskip#3\relax \advance\leftskip\@tempdima\relax
    \rightskip\@pnumwidth plus4em \parfillskip-\@pnumwidth
          \ifnum #1=1 \bfseries #5\else #5\fi 
   \leavevmode\hskip-\@tempdima
      \ifcase #1
       \or\or \hskip 1em \or \hskip 2em \else \hskip 3em \fi%
#6     \nobreak\relax
{\ifnum #1=1\hfill \else \SparseDotfill\fi}
 \hbox to\@pnumwidth{\@tocpagenum{
    \ifnum #1=1 \bfseries \fi #7}}\par
    \nobreak
    \endgroup
  \fi}
\makeatother

\DeclareMathOperator{\ED}{ED}

\newcommand{\norm}[1]{|\!|\!| #1 | \! | \! |}

\setenumerate{label=\rm (\roman*),labelsep=5mm}

\newcommand{\iu}{\mathbf{i}}

\usepackage{graphicx}
\DeclareMathOperator{\orcomplexity}{\Pi} 

\newcommand{\SURJ}{\operatorname{SURJ}}
\newcommand{\THR}{\operatorname{THR}}

\providecommand{\noopsort}[1]{}

\AtBeginDocument{
  
}

\makeatother

  \providecommand{\claimname}{Claim}
  \providecommand{\corollaryname}{Corollary}
  \providecommand{\factname}{Fact}
  \providecommand{\lemmaname}{Lemma}
  \providecommand{\propositionname}{Proposition}
  \providecommand{\theoremname}{Theorem}
\providecommand{\theoremname}{Theorem}

\begin{document}

\title{Algorithmic Polynomials}

\author{Alexander A. Sherstov}

\thanks{$^{*}$ Computer Science Department, UCLA, Los Angeles, CA~90095.
{\large{}\Letter ~}\texttt{sherstov@cs.ucla.edu }Supported by NSF
CAREER award CCF-1149018 and an Alfred P. Sloan Foundation Research
Fellowship.}
\begin{abstract}
The \emph{approximate degree} of a Boolean function $f(x_{1},x_{2},\ldots,x_{n})$
is the minimum degree of a real polynomial that approximates $f$
pointwise within $1/3$. Upper bounds on approximate degree have a
variety of applications in learning theory, differential privacy,
and algorithm design in general. Nearly all known upper bounds on
approximate degree arise in an existential manner from bounds on quantum
query complexity. 

We develop a first-principles, classical approach to the polynomial
approximation of Boolean functions. We use it to give the first constructive
upper bounds on the approximate degree of several fundamental problems:

\bigskip{}
\begin{enumerate}[itemsep=1mm,labelsep=2mm,label=\textbullet]
\item $O\bigl(n^{\frac{3}{4}-\frac{1}{4(2^{k}-1)}}\bigr)$ for the $k$-element
distinctness problem;{\small \par}
\item $O(n^{1-\frac{1}{k+1}})$ for the $k$-subset sum problem;{\small \par}
\item $O(n^{1-\frac{1}{k+1}})$ for any $k$-DNF or $k$-CNF formula;{\small \par}
\item $O(n^{3/4})$ for the surjectivity problem.\bigskip{}
{\small \par}
\end{enumerate}
In all cases, we obtain explicit, closed-form approximating polynomials
that are unrelated to the quantum arguments from previous work. Our
first three results match the bounds from quantum query complexity.
Our fourth result improves polynomially on the $\Theta(n)$ quantum
query complexity of the problem and refutes the conjecture by several
experts that surjectivity has approximate degree $\Omega(n)$. In
particular, we exhibit the first \emph{natural} problem with a polynomial
gap between approximate degree and quantum query complexity.
\end{abstract}

\maketitle
\belowdisplayskip=12pt plus 3pt minus 3pt 
\abovedisplayskip=12pt plus 3pt minus 3pt 
\thispagestyle{empty}

\newpage{}\thispagestyle{empty}
\hypersetup{linkcolor=black} \tableofcontents{}\newpage{}

\hypersetup{linkcolor=teal} 
\thispagestyle{empty}

\section{Introduction}

Let $f\colon X\to\zoo$ be a given Boolean function, defined on a
subset $X\subseteq\zoo^{n}.$ The \emph{$\epsilon$-approximate degree
of $f$}, denoted $\deg_{\epsilon}(f),$ is the minimum degree of
a multivariate real polynomial $p$ such that $|f(x)-p(x)|\le\epsilon$
for all $x\in X.$ The standard setting of the error parameter for
most applications is $\epsilon=1/3$, an aesthetically motivated constant
that can be replaced by any other in $(0,1/2)$ at the expense of
a constant-factor increase in approximate degree. The notion of approximate
degree originated 25~years ago in the pioneering work of Nisan and
Szegedy~\cite{nisan-szegedy94degree} and has since proved to be
a powerful and versatile tool in theoretical computer science. Lower
bounds on approximate degree have complexity-theoretic applications,
whereas upper bounds are a tool in algorithm design. In the former
category, the notion of approximate degree has enabled spectacular
progress in circuit complexity~\cite{paturi-saks94rational,siu-roy-kailath94rational,beigel91rational,aspnes91voting,krause94depth2mod,KP98threshold,sherstov07ac-majmaj,beame-huyn-ngoc09multiparty-focs},
quantum query complexity~\cite{beals-et-al01quantum-by-polynomials,buhrman-et-al99small-error,aaronson-shi04distinctness,aaronson04sdpt-for-search,ambainis05collision,klauck07product-thms,BKT17poly-strikes-back},
and communication complexity~\cite{buhrman-dewolf01polynomials,razborov02quantum,buhrman07pp-upp,sherstov07ac-majmaj,sherstov07quantum,RS07dc-dnf,lee-shraibman08disjointness,chatt-ada08disjointness,dual-survey,beame-huyn-ngoc09multiparty-focs,sherstov12mdisj,sherstov13directional}.
On the algorithmic side, approximate degree underlies many of the
strongest results obtained to date in computational learning~\cite{tt99DNF-incl-excl,KS01dnf,KOS:02,KKMS,odonnell03degree,ACRSZ07nand},
differentially private data release~\cite{tuv12releasing-marginals,ctuw14release-of-marginals},
and algorithm design in general~\cite{linial-nisan90incl-excl,kahn96incl-excl,sherstov07inclexcl-ccc}.

Despite these applications, progress in understanding approximate
degree as a complexity measure has been slow and difficult. With very
few exceptions~\cite{nisan-szegedy94degree,kahn96incl-excl,sherstov07inclexcl-ccc,sherstov12noisy},
all known upper bounds on approximate degree arise from \emph{quantum
query algorithms}. The connection between approximate degree and quantum
query complexity was discovered by Beals et al.~\cite{beals-et-al01quantum-by-polynomials},
who proved that the acceptance probability of an algorithm that makes
$T$ queries is representable by a real polynomial of degree $2T$.
Put another way, every quantum algorithm implies an approximating
polynomial of comparable complexity for the problem in question. Since
the seminal work of Beals et al., essentially all upper bounds on
approximate degree have come from quantum query algorithms, e.g.,~\cite{buhrman-et-al99small-error,de-wolf08approx-degree,ambainis04k-distinctness,FGG08nand,ACRSZ07nand,drucker-de-wolf11uniform-approx-by-quantum-algorithms,drucker-dewolf11quantum-method,belovs12distinctness,mahadev-de-wolf15rational-approx-postselection}.
An illustrative example is the problem of determining the approximate
degree of Boolean formulas of size $n,$ posed in 2003 by O'Donnell
and Servedio~\cite{odonnell03degree}. Progress on this question
was stalled for a long time until it was finally resolved by Ambainis
et al.~\cite{ACRSZ07nand}, who built on the work of Farhi et al.~\cite{FGG08nand}
to give a near-optimal quantum query algorithm for any Boolean formula.

While quantum query complexity has been a fruitful source of approximate
degree upper bounds, the exclusive reliance on quantum techniques
for the polynomial approximation of Boolean functions is problematic.
For one thing, a quantum query algorithm generally does not give any
information about the approximating polynomial apart from its existence.
For example, converting the quantum algorithms of~\cite{ambainis04k-distinctness,ACRSZ07nand,belovs12distinctness}
to polynomials results in expressions so large and complicated that
they are no longer meaningful. More importantly, quantum query algorithms
are more constrained objects than real polynomials, and an optimal
query algorithm for a given problem may be far less efficient than
a polynomial constructed from scratch. Given the many unresolved questions
on approximate degree, there is a compelling need for polynomial approximation
techniques that go beyond quantum query complexity.

In this paper, we take a fresh look at several breakthrough upper
bounds for approximate degree, obtained over the years by sophisticated
quantum query algorithms. In each case, we are able to construct an
approximating polynomial from first principles that matches or improves
on the complexity of the best quantum algorithm. All of our constructions
produce explicit, closed-form polynomials that are unrelated to the
corresponding quantum algorithms and are in the author's opinion substantially
simpler. In one notable instance, our construction achieves a polynomial
improvement on the complexity of the best possible quantum algorithm,
refuting a conjecture~\cite{bun-thaler17adeg-ac0} on the approximate
degree of that problem and exhibiting the first \emph{natural} example
of a polynomial gap between approximate degree and quantum query complexity.
Our proofs, discussed shortly, contribute novel techniques to the
area.

\subsection{\emph{k}-Element distinctness}

The starting point in our work is the \emph{element distinctness problem}~\cite{bdhhmsw01element-distinctness,aaronson-shi04distinctness,ambainis04k-distinctness,ambainis05collision,kutin05collision,belovs12distinctness},
which is one of the most studied questions in quantum query complexity
and a major success story of the field. The input to the problem is
a list of $n$ elements from a given range of size $r,$ and the objective
is to determine if the elements are pairwise distinct. A well-studied
generalization of this problem is \emph{$k$-element distinctness},
where $k$ is an arbitrary constant and the objective is to determine
if some $k$-tuple of the elements are identical. Formally, the input
to element distinctness and $k$-element distinctness is represented
by a Boolean matrix $x\in\zoo^{n\times r}$ in which every row $i$
has precisely one ``$1$'' entry, corresponding to the value of
the $i$th element.\footnote{Alternately, the input can be represented by a string of $n\lceil\log r\rceil$
bits. Switching to this more compact representation changes the complexity
of the problem by a factor of at most $\lceil\log r\rceil,$ which
is negligible in all settings of interest.} Aaronson and Shi~\cite{aaronson-shi04distinctness}, Ambainis~\cite{ambainis05collision},
and Kutin~\cite{kutin05collision} showed that element distinctness
has quantum query complexity $\Omega(n^{2/3}).$ In follow-up work,
Ambainis~\cite{ambainis04k-distinctness} gave a quantum algorithm
for element distinctness with $O(n^{2/3})$ queries, matching the
lower bound in~\cite{aaronson-shi04distinctness,ambainis05collision,kutin05collision}.
For the more general problem of $k$-element distinctness, Ambainis's
algorithm~\cite{ambainis04k-distinctness} requires $O(n^{k/(k+1)})$
queries. Using a different approach, Belovs~\cite{belovs12distinctness}
gave a polynomially faster algorithm for $k$-element distinctness,
with query complexity $O(n^{\frac{3}{4}-\frac{1}{4(2^{k}-1)}})$.
Belovs's algorithm is currently the fastest known.

The algorithms of Ambainis~\cite{ambainis04k-distinctness} and Belovs~\cite{belovs12distinctness}
are highly nontrivial. The former is based on a quantum walk on the
Johnson graph, whereas the latter uses the framework of learning graphs.
We give an elementary, closed-form construction of an approximating
polynomial for $k$-element distinctness that bypasses the quantum
work. Formally, let $\ED_{n,r,k}\colon\zoo_{\leq n}^{n\times r}\to\zoo$
be given by
\[
\ED_{n,r,k}(x)=\begin{cases}
1 & \text{if }x_{1,j}+x_{2,j}+\cdots+x_{n,j}<k\text{ for each \ensuremath{j,}}\\
0 & \text{otherwise.}
\end{cases}
\]
The notation $\zoo_{\leq n}^{n\times r}$ for the domain of this function
indicates that we allow \emph{arbitrary} input matrices $x\in\zoo^{n\times r}$
of Hamming weight at most $n$, with no restriction on the placement
of the ``1'' bits. This is of course a problem more general than
$k$-element distinctness. We prove:
\begin{thm}[$k$-element distinctness]
\label{thm:MAIN-ed}Let $k\geq1$ be a fixed integer. Then for all
$n,r\geq1,$ 
\begin{align*}
\deg_{1/3}(\ED_{n,r,k}) & =O\left(\sqrt{n}\min\{n,r\}^{\frac{1}{2}-\frac{1}{4(1-2^{-k})}}\right).
\end{align*}
Moreover, the approximating polynomial is given explicitly in each
case.
\end{thm}

\noindent Theorem~\ref{thm:MAIN-ed} matches the quantum query bound
of $O(n^{\frac{3}{4}-\frac{1}{4(2^{k}-1)}})\equiv O(n^{1-\frac{1}{4(1-2^{-k})}})$
due to Belovs~\cite{belovs12distinctness} and further generalizes
it to every $r\geq1.$

\subsection{\emph{k}-Subset sum, \emph{k}-DNF and \emph{k}-CNF formulas}

Another well-studied problem in quantum query complexity is \emph{$k$-subset
sum}~\cite{childs-eisenberg05subset-sum,belovs-spalek13k-sum}. The
input to this problem is a list of $n$ elements from a given finite
Abelian group $G,$ and the objective is to determine whether there
is a $k$-tuple of elements that sum to $0.$ Formally, the input
is represented by a matrix $x\in\zoo^{n\times|G|}$ with precisely
one ``1'' entry in every row. Childs and Eisenberg~\cite{childs-eisenberg05subset-sum}
contributed an alternate analysis of Ambainis's algorithm for $k$-element
distinctness~\cite{ambainis04k-distinctness} and showed how to adapt
it to compute $k$-subset sum or any other function property with
$1$-certificate complexity at most $k.$ In particular, any such
problem has an approximating polynomial of degree $O(n^{k/(k+1)}).$
We give a first-principles construction of an approximating polynomial
for any problem in this class, using techniques that are elementary
and unrelated to the quantum work of Ambainis~\cite{ambainis04k-distinctness}
and Childs and Eisenberg~\cite{childs-eisenberg05subset-sum}. Our
result is more general:
\begin{thm}[$k$-DNF and $k$-CNF formulas]
\label{thm:MAIN-k-dnf}Let $k\geq0$ be a fixed integer. Let $f\colon\zoo_{\leq n}^{N}\to\zoo$
be representable on its domain by a $k$-DNF or $k$-CNF formula.
Then
\[
\deg_{1/3}(f)=O(n^{\frac{k}{k+1}}).
\]
Moreover, the approximating polynomial is given explicitly in each
case.
\end{thm}

\noindent Recall that a $k$-DNF formula in Boolean variables $x_{1},x_{2},\ldots,x_{N}$
is the disjunction of an arbitrary number of terms, where each term
is the conjunction of at most $k$ literals from among $x_{1},\overline{x_{1}},x_{2},\overline{x_{2}},\ldots,x_{N},\overline{x_{N}}.$
An essential aspect of Theorem~\ref{thm:MAIN-k-dnf} is that the
approximate degree upper bound depends only on the Hamming weight
$x_{1}+x_{2}+\cdots+x_{N}$ of the input and does not depend at all
on the number of variables $N$, which can be arbitrarily large. Several
special cases of Theorem~\ref{thm:MAIN-k-dnf} are worth noting.
The theorem clearly applies to $k$-subset sum, which is by definition
representable on its domain by a $k$-DNF formula. Moreover, in the
terminology of Childs and Eisenberg~\cite{childs-eisenberg05subset-sum},
Theorem~\ref{thm:MAIN-k-dnf} applies to any function property with
$1$-certificate complexity at most~$k$. Finally, taking $N=n$
shows that Theorem~\ref{thm:MAIN-k-dnf} applies to any function
$f\colon\zoon\to\zoo$ representable by a $k$-DNF or $k$-CNF~formula.

\subsection{Surjectivity}

While our proofs of Theorems~\ref{thm:MAIN-ed} and~\ref{thm:MAIN-k-dnf}
are significantly simpler than their quantum query counterparts, they
do not give a quantitative improvement on previous work. This brings
us to our next result. In the \emph{surjectivity problem~}\cite{beame-machmouchi12quantum-query-ac0},
the input is a list of $n$ elements from a given range of size $r,$
where $r\leq n.$ The objective is to determine whether the input
features all $r$ elements of the range. In function terminology,
the input represents a mapping $\{1,2,\ldots,n\}\to\{1,2,\ldots,r\},$
and the objective is to determine whether the mapping is surjective.
As usual in the quantum query literature, the input is represented
by a Boolean matrix $x\in\zoo^{n\times r}$ in which every row has
precisely one ``$1$'' entry. Beame and Machmouchi~\cite{beame-machmouchi12quantum-query-ac0}
proved that for $r=\lfloor n/2\rfloor+1,$ the surjectivity problem
has the maximum possible quantum query complexity, namely, $\Theta(n).$
This led several experts to conjecture that the approximate degree
of surjectivity is also $\Theta(n)$; see, e.g.,~\cite{bun-thaler17adeg-ac0}.
The conjecture was significant because its resolution would give the
first $\mathsf{AC}^{0}$ circuit with approximate degree $\Theta(n),$
closing a long line of research~\cite{nisan-szegedy94degree,aaronson-shi04distinctness,ambainis05collision,bun-thaler17adeg-ac0}. 

Surprisingly, we are able to show that surjectivity has an approximating
polynomial of substantially lower degree, regardless of the range
parameter $r$. Formally, let $\SURJ_{n,r}\colon\zoo_{\leq n}^{n\times r}\to\zoo$
be given by
\[
\SURJ_{n,r}(x)=\bigwedge_{j=1}^{r}\bigvee_{i=1}^{n}x_{i,j}.
\]
In keeping with our other results, our definition of $\SURJ_{n,r}$
allows arbitrary input matrices $\zoo^{n\times r}$ of Hamming weight
at most $n.$ In this generalization of the surjectivity problem,
the input can be thought of as an arbitrary relation rather than a
function. We prove:
\begin{thm}[Surjectivity]
\label{thm:MAIN-surj}For all positive integers $n$ and $r,$
\[
\deg_{1/3}(\SURJ_{n,r})=\begin{cases}
O(\sqrt{n}\cdot r^{1/4}) & \text{if }r\leq n,\\
0 & \text{if \ensuremath{r>n.}}
\end{cases}
\]
Moreover, the approximating polynomial is given explicitly in each
case.
\end{thm}

\noindent In particular, the theorem gives an approximating polynomial
of degree $O(n^{3/4})$ for all $r.$ This upper bound is polynomially
smaller than the problem's quantum query complexity $\Theta(n)$ for
$r=\lfloor n/2\rfloor+1.$ While explicit functions with a polynomial
gap between approximate degree and quantum query complexity have long
been known~\cite{ambainis06adeg-vs-query,ABK16cheat-sheets}, Theorem~\ref{thm:MAIN-surj}
exhibits the first \emph{natural} function with this property. The
functions in previous work~\cite{ambainis06adeg-vs-query,ABK16cheat-sheets}
were constructed with the specific purpose of separating complexity
measures.

\subsection{Symmetric functions}

Key building blocks in our proofs are symmetric functions $f\colon\zoon\to\zoo.$
A classic result due to Paturi~\cite{paturi92approx} states that
the $1/3$-approximate degree of any such function $f$ is $\Theta(\sqrt{n\ell}),$
where $\ell\in\{0,1,2,\ldots,n\}$ is the smallest number such that
$f$ is constant on inputs of Hamming weight in $[\ell,n-\ell].$
When a symmetric function is used in an auxiliary role as part of
a larger construction, it becomes important to have approximating
polynomials for every possible setting of the error parameter, $1/2^{n}\leq\epsilon\leq1/3$.
A complete characterization of the $\epsilon$-approximate degree
of symmetric functions for all $\epsilon$ was obtained by de Wolf~\cite{de-wolf08approx-degree},
who sharpened previous bounds~\cite{kahn96incl-excl,buhrman-et-al99small-error,sherstov07inclexcl-ccc}
using an elegant quantum query algorithm. Prior to our work, no classical,
first-principles proof was known for de Wolf's characterization, which
is telling in view of the basic role that $\AND_{n},\OR_{n},$ and
other symmetric functions play in the area. We are able to give such
a first-principles proof\textemdash in fact, \emph{three} of them.
\begin{thm}[Symmetric functions]
\label{thm:MAIN-symmetric}Let $f\colon\zoon\to\{0,1\}$ be a symmetric
function. Let $\ell\in\{0,1,2,\ldots,n\}$ be an integer such that
$f$ is constant on inputs of Hamming weight in $(\ell,n-\ell).$
Then for $1/2^{n}\leq\epsilon\leq1/3,$ 
\[
\deg_{\epsilon}(f)=O\left(\sqrt{n\ell}+\sqrt{n\log\frac{1}{\epsilon}}\right).
\]
Moreover, the approximating polynomial is given explicitly in each
case.
\end{thm}

\noindent Theorem~\ref{thm:MAIN-symmetric} matches de Wolf's quantum
query result, tightly characterizing the $\epsilon$-approximate degree
of every nonconstant symmetric function.

\subsection{Our techniques}

Our proofs use only basic tools from approximation theory, such as
Chebyshev polynomials. Our constructions additionally incorporate
elements of classic algorithm design, e.g., the divide-and-conquer
paradigm, the inclusion-exclusion principle, and probabilistic reasoning.
The title of our paper, ``Algorithmic Polynomials,'' is a reference
to this combination of classic algorithmic methodology and approximation
theory. The informal message of our work is that algorithmic polynomials
are not only more powerful than quantum algorithms but also easier
to construct. A detailed discussion of Theorems~\ref{thm:MAIN-ed}\textendash \ref{thm:MAIN-symmetric}
follows.

\subsubsection*{Extension theorem.}

As our starting point, we prove an \emph{extension theorem} for polynomial
approximation. This theorem allows one to construct an approximant
for a given function $F$ using an approximant for a restriction $f$
of $F.$ In more detail, let $f\colon\zoo_{\leq m}^{N}\to[-1,1]$
be an arbitrary function, defined on inputs $x\in\zoo^{N}$ of Hamming
weight at most $m.$ Let $F_{n}\colon\zoo_{\leq n}^{N}\to[-1,1]$
be the natural extension of $f$ to inputs of Hamming weight at most
$n,$ defined by $F_{n}=0$ outside the domain of $f.$ From an approximation-theoretic
point of view, a fundamental question to ask is how to efficiently
``extend'' any approximant for $f$ to an approximant for $F_{n}.$
Unfortunately, this naïve formulation of the extension problem has
no efficient solution; we describe a counterexample in Section~\ref{sec:extension}.
We are able to show, however, that the extension problem becomes meaningful
if one works with $F_{2m}$ instead of $f$. In other words, we give
an efficient, explicit, black-box transformation of any approximant
for the extension $F_{2m}$ into an approximant for the extension
$F_{n}$, for any $n\geq2m$. This result is essentially as satisfying
as the ``ideal'' extension theorem in that the domains of $f$ and
$F_{2m}$ almost coincide and can be arbitrarily smaller than the
domain of $F_{n}$. Our proof makes use of extrapolation bounds, extremal
properties of Chebyshev polynomials, and ideas from rational approximation
theory. 

\subsubsection*{Symmetric functions.}

As mentioned earlier, we give three proofs of Theorem~\ref{thm:MAIN-symmetric}
on the $\epsilon$-approximate degree of symmetric functions. Each
of the three proofs is fully constructive. Our simplest proof uses
the extension theorem and is only half-a-page long. Here, we use brute-force
interpolation to compute the function $f$ of interest on inputs of
small Hamming weight, and then apply the extension theorem to effortlessly
extend the interpolant to the full domain of $f.$ Our second proof
of Theorem~\ref{thm:MAIN-symmetric} is an explicit, closed-form
construction that uses Chebyshev polynomials as its only ingredient.
This proof is a refinement of previous, suboptimal approximants for
the AND function~\cite{kahn96incl-excl,sherstov07inclexcl-ccc}.
We eliminate the inefficiency in previous work by using Chebyshev
polynomials to achieve improved control at every point of the domain.
Finally, our third proof of Theorem~\ref{thm:MAIN-symmetric} is
inspired by combinatorics rather than approximation theory. Here,
we use a \emph{sampling experiment} to construct an approximating
polynomial for any symmetric function $f$ from an approximating polynomial
for AND. In more detail, the experiment allows us to interpret $f$
as a linear combination of conjunctions of arbitrary degree, where
the sum of the absolute values of the coefficients is reasonably small.
Once such a representation is available, we simply replace every conjunction
with its approximating polynomial. These substitutions increase the
error of the approximation by a factor bounded by the sum of the absolute
values of the coefficients in the original linear combination, which
is negligible.

\subsubsection*{k-Element distinctness, k-DNF and k-CNF formulas.\emph{ }}

We first establish an auxiliary result on the approximate degree of
composed Boolean functions. Specifically, let $F\colon X\times\zoo_{\leq n}^{N}\to\zoo$
be given by $F(x,y)=\bigvee_{i=1}^{N}y_{i}\wedge f_{i}(x)$ for some
set $X$ and some functions $f_{1},f_{2},\ldots,f_{N}\colon X\to\zoo.$
We bound the $\epsilon$-approximate degree of $F$ in terms of the
approximate degree of $\bigvee_{i\in S}f_{i}$, maximized over all
sets $S\subseteq\{1,2,\ldots,N\}$ of certain size. Crucially for
our applications, the bound that we derive has no dependence on $N.$
The proof uses Chebyshev polynomials and the inclusion-exclusion principle.
Armed with this \emph{composition theorem}, we give a short proof
of Theorem~\ref{thm:MAIN-k-dnf} on the approximate degree of $k$-DNF
and $k$-CNF formulas. The argument proceeds by induction on $k,$
with the composition theorem invoked to implement the inductive step.
The proof of Theorem~\ref{thm:MAIN-ed} on the approximate degree
of $k$-element distinctness is more subtle. It too proceeds by induction,
with the composition theorem playing a central role. This time, however,
the induction is with respect to both $k$ and the range parameter
$r,$ and the extension theorem is required to complete the inductive
step. We note that we are able to bound the $\epsilon$-approximate
degree of $k$-DNF formulas and $k$-element distinctness for every
setting of the error parameter $\epsilon$, rather than just $\epsilon=1/3$
in Theorems~\ref{thm:MAIN-ed} and~\ref{thm:MAIN-k-dnf}.

\subsubsection*{Surjectivity.}

Our proof of Theorem~\ref{thm:MAIN-surj} is surprisingly short,
given how improbable the statement was believed to be. As one can
see from the defining equation for $\SURJ_{n,r}$, this function is
the componentwise composition $\AND_{r}\circ\OR_{n}$ restricted to
inputs of Hamming weight at most $n.$ With this in mind, we start
with a degree-$O(\sqrt{r})$ polynomial $\widetilde{\AND}_{r}$ that
approximates $\AND_{r}$ pointwise within $1/4.$ The approximant
in question is simply a scaled and shifted Chebyshev polynomial. It
follows that the componentwise composition $\widetilde{\AND}_{r}\circ\OR_{n}$,
restricted to inputs of Hamming weight at most $n,$ approximates
$\SURJ_{n,r}$ pointwise within $1/4$. We are not finished, however,
because the degree of $\widetilde{\AND}_{r}\circ\OR_{n}$ is unacceptably
large. Moving on, a few lines of algebra reveal that $\widetilde{\AND}_{r}\circ\OR_{n}$
is a linear combination of conjunctions in which the absolute values
of the coefficients sum to $2^{O(\sqrt{r})}$. It remains to approximate
each of these conjunctions pointwise within $2^{-\Omega(\sqrt{r})}$
by a polynomial of degree $O(\sqrt{n\sqrt{r}})=O(\sqrt{n}\cdot r^{1/4}),$
for which we use our explicit approximant from Theorem~\ref{thm:MAIN-symmetric}
along with the guarantee that the input has Hamming weight at most
$n.$ The proof of Theorem~\ref{thm:MAIN-surj} is particularly emblematic
of our work in its interplay of approximation-theoretic methodology
(Chebyshev polynomials, linear combinations) and algorithmic thinking
(reduction of the problem to the approximation of individual conjunctions).

\medskip{}
We are pleased to report that our $O(n^{3/4})$ upper bound for the
surjectivity problem has just sparked further progress in the area
by Bun, Kothari, and Thaler~\cite{BKT17poly-strikes-back}, who prove
tight or nearly tight lower bounds on the approximate degree of several
key problems in quantum query complexity. In particular, the authors
of~\cite{BKT17poly-strikes-back} prove that our upper bound for
surjectivity is tight. We are confident that the ideas of our work
will inform future research as well.

\section{Preliminaries}

We start with a review of the technical preliminaries. The purpose
of this section is to make the paper as self-contained as possible,
and comfortably readable by a broad audience. The expert reader may
wish to skim it for the notation or skip it altogether.

\subsection{Notation}

We view Boolean functions as mappings $X\to\zoo$ for some finite
set $X.$ This arithmetization of the Boolean values ``true'' and
``false'' makes it possible to use Boolean operations in arithmetic
expressions, as in $1-2\bigvee_{i=1}^{n}x_{i}.$ The familiar functions
$\OR_{n}\colon\zoon\to\zoo$ and $\AND_{n}\colon\zoon\to\zoo$ are
given by $\OR_{n}(x)=\bigvee_{i=1}^{n}x_{i}$ and $\AND_{n}(x)=\bigwedge_{i=1}^{n}x_{i}=\prod_{i=1}^{n}x_{i}.$
The negation of a Boolean function $f$ is denoted as usual by $\overline{f}=1-f.$
The composition of $f$ and $g$ is denoted $f\circ g$, with $(f\circ g)(x)=f(g(x)).$

For a string $x\in\zoon,$ we denote its Hamming weight by $|x|=x_{1}+x_{2}+\cdots+x_{n}.$
We use the following notation for strings of Hamming weight at most
$k,$ greater than $k,$ and exactly $k$: 
\begin{align*}
\zoon_{\leq k} & =\{x\in\zoon:|x|\leq k\},\\
\zoon_{>k} & =\{x\in\zoon:|x|>k\},\\
\zoon_{k} & =\{x\in\zoon:|x|=k\}.
\end{align*}
For a string $x\in\zoon$ and a set $S\subseteq\{1,2,\ldots,n\},$
we let $x|_{S}$ denote the restriction of $x$ to the indices in
$S.$ In other words, $x|_{S}=x_{i_{1}}x_{i_{2}}\ldots x_{i_{|S|}},$
where $i_{1}<i_{2}<\cdots<i_{|S|}$ are the elements of $S.$ The
characteristic vector of a subset $S\subseteq\{1,2,\ldots,n\}$ is
denoted $\1_{S}.$

We let $\NN=\{0,1,2,3,\ldots\}$ and $[n]=\{1,2,\ldots,n\}.$ For
a set $S$ and a real number $k,$ we define
\begin{align*}
\binom{S}{k} & =\{A\subseteq S:|A|=k\},\\
\binom{S}{\mathord{\leq}k} & =\{A\subseteq S:|A|\leq k\}.
\end{align*}
We analogously define $\binom{S}{\geq k},\binom{S}{<k},$ and $\binom{S}{>k}.$
We let $\ln x$ and $\log x$ stand for the natural logarithm of $x$
and the logarithm of $x$ to base $2,$ respectively. The following
bound is well known~\cite[Proposition~1.4]{jukna11extremal-2nd-edition}:
\begin{align}
\sum_{i=0}^{k}{n \choose i}\leq\left(\frac{\e n}{k}\right)^{k}, &  & k=0,1,2,\dots,n,\label{eq:entropy-bound-binomial}
\end{align}
where $\e=2.7182\ldots$ denotes Euler's number. For a logical condition
$C,$ we use the Iverson bracket notation
\[
\I[C]=\begin{cases}
1 & \text{if \ensuremath{C} holds,}\\
0 & \text{otherwise.}
\end{cases}
\]
For a function $f\colon X\to\Re$ on a finite set $X,$ we use the
standard norms
\begin{align*}
 & \|f\|_{\infty}=\max_{x\in X}\,|f(x)|,\\
 & \|f\|_{1}=\sum_{x\in X}\,|f(x)|.
\end{align*}

\subsection{Approximate degree}

Recall that the \emph{total degree} of a multivariate real polynomial
$p\colon\Re^{n}\to\Re$, denoted $\deg p,$ is the largest degree
of any monomial of $p.$ We use the terms ``degree'' and ``total
degree'' interchangeably in this paper. This paper studies the  approximate
representation of functions of interest by polynomials. Specifically,
let $f\colon X\to\Re$ be a given function, for a finite subset $X\subset\Re^{n}.$
Define 
\[
E(f,d)=\min_{p:\deg p\leq d}\|f-p\|_{\infty},
\]
where the minimum is over polynomials of degree at most $d.$ In words,
$E(f,d)$ is the least error to which $f$ can be approximated by
a real polynomial of degree at most $d$. For a real number $\epsilon\geq0,$
the \emph{$\epsilon$-approximate degree} of $f$ is defined as
\[
\deg_{\epsilon}(f)=\min\{d:E(f,d)\leq\epsilon\}.
\]
Thus, $\deg_{\epsilon}(f)$ is the least degree of a real polynomial
that approximates $f$ pointwise to within $\epsilon.$ We refer to
any such polynomial as a\emph{ uniform approximant for $f$ with error
$\epsilon$}. In the study of Boolean functions $f$, the standard
setting of the error parameter is $\epsilon=1/3$. This constant is
chosen mostly for aesthetic reasons and can be replaced by any other
constant in $(0,1/2)$ at the expense of a constant-factor increase
in approximate degree. The following fact on the \emph{exact} representation
of functions by polynomials is well known.
\begin{fact}
\label{fact:trivial-upper-bound-on-deg}For every function $f\colon\zoo_{\leq n}^{N}\to\Re,$
\[
\deg_{0}(f)\leq n.
\]
\end{fact}

\begin{proof}
The proof is by induction on $n.$ The base case $n=0$ is trivial
since $f$ is then a constant function. For the inductive step, let
$n\geq1$ be arbitrary. By the inductive hypothesis, there is a polynomial
$p_{n-1}(x)$ of degree at most $n-1$ such that $f(x)=p_{n-1}(x)$
for inputs $x\in\zoo^{N}$ of Hamming weight at most $n-1.$ Define
\[
p_{n}(x)=p_{n-1}(x)+\sum_{a\in\zoo_{n}^{N}}(f(a)-p_{n-1}(a))\prod_{i:a_{i}=1}x_{i}.
\]
For any fixed input $x$ with $|x|\leq n-1,$ every term in the summation
over $a$ evaluates to zero and therefore $p_{n}(x)=p_{n-1}(x)=f(x).$
For any fixed input $x$ with $|x|=n,$ on the other hand, the summation
over $a$ contributes precisely one nonzero term, corresponding to
$a=x.$ As a result, $p_{n}(x)=p_{n-1}(x)+(f(x)-p_{n-1}(x))=f(x)$
in that case.
\end{proof}

\subsection{Inclusion-exclusion}

All Boolean, arithmetic, and relational operations on functions in
this paper are to be interpreted pointwise. For example, $\bigvee_{i=1}^{n}f_{i}$
refers to the mapping $x\mapsto\bigvee_{i=1}^{n}f_{i}(x).$ Similarly,
$\prod_{i=1}^{n}f_{i}$ is the pointwise product of $f_{1},f_{2},\ldots,f_{n}$.
Recall that in the case of Boolean functions, we have $\bigwedge_{i=1}^{n}f_{i}=\prod_{i=1}^{n}f_{i}.$
The well-known \emph{inclusion-exclusion principle}, stated in terms
of Boolean functions $f_{1},f_{2},\ldots,f_{n},$ asserts that
\[
\bigvee_{i=1}^{n}f_{i}=\sum_{\substack{S\subseteq\{1,2,\ldots,n\}\\
S\ne\varnothing
}
}(-1)^{|S|+1}\prod_{i\in S}f_{i}.
\]
We will need the following less common form of the inclusion-exclusion
principle, where the AND and OR operators are interchanged. 
\begin{fact}
\label{fact:incl-excl-alternative}For any $n\geq1$ and any Boolean
functions $f_{1},f_{2},\ldots,f_{n}\colon X\to\zoo,$
\[
\prod_{i=1}^{n}f_{i}=\sum_{\substack{S\subseteq\{1,2,\ldots,n\}\\
S\ne\varnothing
}
}(-1)^{|S|+1}\bigvee_{i\in S}f_{i}.
\]
\end{fact}

\begin{proof}
We have
\begin{align*}
\prod_{i=1}^{n}f_{i} & =\prod_{i=1}^{n}(1-\overline{f_{i}})\\
 & =\sum_{S\subseteq\{1,2,\ldots,n\}}(-1)^{|S|}\prod_{i\in S}\overline{f_{i}}\\
 & =\sum_{S\subseteq\{1,2,\ldots,n\}}(-1)^{|S|}\left(\prod_{i\in S}\overline{f_{i}}-1\right)\\
 & =\sum_{S\subseteq\{1,2,\ldots,n\}}(-1)^{|S|}\left(-\bigvee_{i\in S}f_{i}\right)\\
 & =\sum_{\substack{S\subseteq\{1,2,\ldots,n\}\\
S\ne\varnothing
}
}(-1)^{|S|+1}\bigvee_{i\in S}f_{i},
\end{align*}
where the third step uses the fact that half of the subsets of $\{1,2,\ldots,n\}$
have odd cardinality and the other half have even cardinality.
\end{proof}

\subsection{Symmetrization}

Let $S_{n}$ denote the symmetric group on $n$ elements. For a permutation
$\sigma\in S_{n}$ and a string $x=(x_{1},x_{2},\ldots,x_{n}),$ we
adopt the shorthand $\sigma x=(x_{\sigma(1)},x_{\sigma(2)},\ldots,x_{\sigma(n)}).$
A function $f(x_{1},x_{2},\ldots,x_{n})$ is called \emph{symmetric}
if it is invariant under permutations of the input variables: $f(x_{1},x_{2},\ldots,x_{n})\equiv f(x_{\sigma(1)},x_{\sigma(2)},\ldots,x_{\sigma(n)})$
for all $x$ and $\sigma.$ Symmetric functions on $\zoon$ are intimately
related to univariate polynomials, as borne out by Minsky and Papert's
\emph{symmetrization argument}~\cite{minsky88perceptrons}.
\begin{prop}[Minsky and Papert]
\label{prop:minsky-papert}Let $p\colon\zoon\to\Re$ be a polynomial
of degree $d.$ Then there is a univariate polynomial $p^{*}$ of
degree at most $d$ such that for all $x\in\zoon,$
\begin{align*}
\Exp_{\sigma\in S_{n}}p(\sigma x)=p^{*}(|x|).
\end{align*}
\end{prop}

\noindent Minsky and Papert's result generalizes to block-symmetric
functions, as pointed out in~\cite[Prop.~2.3]{RS07dc-dnf}:
\begin{prop}
\label{prop:symmetrization} Let $n_{1},\dots,n_{k}$ be positive
integers. Let $p\colon\zoo^{n_{1}}\times\cdots\times\zoo^{n_{k}}\to\Re$
be a polynomial of degree $d.$ Then there is a polynomial $p^{*}\colon\Re^{k}\to\Re$
of degree at most $d$ such that for all $x_{1}\in\zoo^{n_{1}},\ldots,x_{k}\in\zoo^{n_{k}},$
\begin{align*}
\Exp_{\sigma_{1}\in S_{n_{1}},\dots,\sigma_{k}\in S_{n_{k}}}p(\sigma_{1}x_{1},\dots,\sigma_{k}x_{k}) & =p^{*}(|x_{1}|,\ldots,|x_{k}|).
\end{align*}
\end{prop}

\noindent Proposition~\ref{prop:symmetrization} follows in a straightforward
manner from Proposition~\ref{prop:minsky-papert} by induction on
the number of blocks, $k.$

\subsection{\label{subsec:chebyshev}Chebyshev polynomials}

Recall from Euler's identity that
\begin{align}
(\cos x+\iu\sin x)^{d} & =\cos dx+\iu\sin dx, &  & d=0,1,2,\ldots,\label{eq:de-moivre}
\end{align}
where $\iu$ denotes the imaginary unit. Multiplying out the left-hand
side and using $\sin^{2}x=1-\cos^{2}x,$ we obtain a univariate polynomial
$T_{d}$ of degree $d$ such that
\begin{equation}
T_{d}(\cos x)=\cos dx.\label{eq:Chebyshev-definition}
\end{equation}
This unique polynomial is the \emph{Chebyshev polynomial of degree
$d$}. The representation~(\ref{eq:Chebyshev-definition}) immediately
reveals all the roots of $T_{d}$, and all the extrema of $T_{d}$
in the interval $[-1,1]$:
\begin{align}
 & T_{d}\left(\cos\left(\frac{2i-1}{2d}\,\pi\right)\right)=0, &  & i=1,2,\ldots,d,\label{eq:Chebyshev-roots}\\
 & T_{d}\left(\cos\left(\frac{i}{d}\,\pi\right)\right)=(-1)^{i}, &  & i=0,1,\ldots,d,\label{eq:Chebyshev-extrema}\\
\rule{0mm}{4mm} & |T_{d}(t)|\leq1, &  & t\in[-1,1].\label{eqn:chebyshev-containment}
\end{align}
The extremum at $1$ is of particular significance, and we note it
separately:
\begin{equation}
T_{d}(1)=1.\label{eq:chebyshev-at-1}
\end{equation}
In view of~(\ref{eq:de-moivre}), the defining equation~(\ref{eq:Chebyshev-definition})
implies that
\begin{align*}
T_{d}(\cos x) & =\sum_{i=0}^{\lfloor d/2\rfloor}\binom{d}{2i}(-1)^{i}(\sin x)^{2i}(\cos x)^{d-2i}\\
 & =\sum_{i=0}^{\lfloor d/2\rfloor}\binom{d}{2i}(\cos^{2}x-1)^{i}(\cos x)^{d-2i},
\end{align*}
so that the leading coefficient of $T_{d}$ for $d\geq1$ is given
by $\sum_{i=0}^{\lfloor d/2\rfloor}\binom{d}{2i}=2^{d-1}$. As a result,
we have the factored representation 
\begin{align}
T_{d}(t) & =2^{d-1}\prod_{i=1}^{d}\left(t-\cos\left(\frac{2i-1}{2d}\,\pi\right)\right), &  & d\geq1.\label{eq:chebyshev-factored}
\end{align}
By~(\ref{eq:de-moivre}) and~(\ref{eq:Chebyshev-definition}),
\begin{align*}
T_{d}(\cos x) & =\cos dx\\
 & =\frac{1}{2}(\cos x-\iu\sin x)^{d}+\frac{1}{2}(\cos x+\iu\sin x)^{d}\\
 & =\frac{1}{2}(\cos x-\iu\sqrt{1-\cos^{2}x})^{d}+\frac{1}{2}(\cos x+\iu\sqrt{1-\cos^{2}x})^{d},
\end{align*}
whence
\begin{align}
T_{d}(t) & =\frac{1}{2}(t-\sqrt{t^{2}-1})^{d}+\frac{1}{2}(t+\sqrt{t^{2}-1})^{d}, &  & |t|\geq1.\label{eq:chebyshev-beyond-1}
\end{align}
The following fundamental fact follows from~(\ref{eq:chebyshev-beyond-1})
by elementary calculus.
\begin{fact}[Derivative of Chebyshev polynomials]
\label{fact:chebyshev-derivative}For any integer $d\geq0$ and real
$t\geq1,$
\begin{align*}
 & T_{d}'(t)\geq d^{2}.
\end{align*}
\end{fact}

\noindent Together, (\ref{eq:chebyshev-beyond-1}) and Fact~\ref{fact:chebyshev-derivative}
give the following useful lower bound for Chebyshev polynomials on
$[1,\infty).$
\begin{prop}
\label{prop:chebyshev-beyond-1}For any integer $d\geq1,$
\begin{align*}
T_{d}(1+\delta) & \geq1+d^{2}\delta, &  & 0\leq\delta<\infty,\\
T_{d}(1+\delta) & \geq2^{d\sqrt{\delta}-1} &  & 0\leq\delta\leq1.
\end{align*}
\end{prop}

\begin{proof}
The first bound follows from the intermediate value theorem in view
of (\ref{eq:chebyshev-at-1}) and Fact~\ref{fact:chebyshev-derivative}.
For the second bound, use~(\ref{eq:chebyshev-beyond-1}) to write
\begin{align*}
T_{d}(1+\delta) & \geq\frac{1}{2}(1+\delta+\sqrt{(1+\delta)^{2}-1})^{d}\\
 & \geq\frac{1}{2}(1+\sqrt{\delta})^{d}\\
 & \geq\frac{1}{2}\cdot2^{d\sqrt{\delta}},
\end{align*}
where the last step uses $1+x\geq2^{x}$ for $x\in[0,1].$
\end{proof}

\subsection{Coefficient bounds for univariate polynomials}

We let $P_{d}$ stand for the set of univariate polynomials of degree
at most $d.$ For a univariate polynomial $p(t)=a_{d}t^{d}+a_{d-1}t^{d-1}+\cdots+a_{1}t+a_{0},$
we let $\norm p=\sum_{i=0}^{d}|a_{i}|$ denote the sum of the absolute
values of the coefficients of $p.$ Then $\norm\cdot$ is a norm on
the real linear space of polynomials, and it is in addition submultiplicative:
\begin{fact}
\label{fact:poly-norm}For any polynomials $p$ and $q,$
\begin{enumerate}
\item $\norm p\geq0,$ with equality if and only if $p=0;$
\item $\norm{\lambda p}=|\lambda|\cdot\norm p$ for any real $\lambda;$
\item $\norm{p+q}\leq\norm p+\norm q;$
\item \label{item:poly-norm-submultiplicative}$\norm{p\cdot q}\leq\norm p\cdot\norm q.$
\end{enumerate}
\end{fact}

\begin{proof}
All four properties follow directly from the definition.
\end{proof}
We will need a bound on the coefficients of a univariate polynomial
in terms of its degree $d$ and its maximum absolute value on the
interval $[0,1].$ This fundamental problem was solved in the nineteenth
century by V.~A.~Markov~\cite[p.~81]{markov-v-a1982least-dev},
who proved an upper bound of 
\begin{align}
O\left(\frac{(1+\sqrt{2})^{d}}{\sqrt{d}}\right)\label{eqn:markov}
\end{align}
on the size of the coefficients of any degree-$d$ polynomial that
is bounded on $[-1,1]$ in absolute value by~$1.$ Markov further
showed that (\ref{eqn:markov}) is tight. Rather than appeal to this
deep result in approximation theory, we will use the following weaker
bound that suffices for our purposes.
\begin{lem}
\label{lem:bound-on-poly-coeffs-UNIVARIATE} Let $p$ be a univariate
polynomial of degree $d$. Then 
\begin{align}
\norm p\leq8^{d}\max_{i=0,1,\dots,d}\left|p\left(\frac{i}{d}\right)\right|.\label{eqn:bound-poly-coeffs}
\end{align}
\end{lem}

\noindent Lemma~\ref{lem:bound-on-poly-coeffs-UNIVARIATE} is a cosmetic
modification of a lemma from~\cite{sherstov12noisy}, which in our
notation states that $\norm p\leq4^{d}\max_{i=0,1,\ldots,d}|p(1-\frac{2i}{d})|$
for $p\in P_{d}$. We include a detailed proof for the reader's convenience. 
\begin{proof}[Proof of Lemma~\emph{\ref{lem:bound-on-poly-coeffs-UNIVARIATE}.}]
 We use a common approximation-theoretic technique~\cite{cheney-book,rivlin-book}
whereby one expresses $p$ as a linear combination of more structured
polynomials and analyzes the latter objects. For this, define $q_{0},q_{1},\dots,q_{d}\in P_{d}$
by 
\begin{align*}
q_{j}(t)=\frac{(-1)^{d-j}d^{d}}{d!}{d \choose j}\prod_{\substack{i=0\\
i\ne j
}
}^{d}\left(t-\frac{i}{d}\right),\qquad j=0,1,\dots,d.
\end{align*}
One easily verifies that these polynomials behave like delta functions,
in the sense that for $i,j=0,1,2,\dots,d,$ 
\begin{align*}
q_{j}\left(\frac{i}{d}\right)=\begin{cases}
1 & \text{if \ensuremath{i=j,}}\\
0 & \text{otherwise.}
\end{cases}
\end{align*}
Lagrange interpolation gives 
\begin{align}
p=\sum_{j=0}^{d}p\left(\frac{j}{d}\right)q_{j}.\label{eqn:p-lin-comb-q}
\end{align}
By Fact~\ref{fact:poly-norm}, 
\begin{align}
\norm{q_{j}} & \leq\frac{d^{d}}{d!}{d \choose j}\prod_{\substack{i=0\\
i\ne j
}
}^{d}\left(1+\frac{i}{d}\right)\nonumber \\
 & \leq\frac{d^{d}}{d!}{d \choose j}\prod_{i=1}^{d}\left(1+\frac{i}{d}\right)\nonumber \\
 & =\frac{1}{d!}{d \choose j}\frac{(2d)!}{d!} &  & \allowdisplaybreaks\nonumber \\
 & ={d \choose j}\binom{2d}{d}\nonumber \\
 & \leq4^{d}\binom{d}{j}, &  & j=0,1,2,\ldots,d.\label{eq:qj-Pi-norm}
\end{align}
Now
\begin{align*}
\norm p & \leq\left(\max_{j=0,1,\dots,d}\left|p\left(\frac{j}{d}\right)\right|\right)\sum_{j=0}^{d}4^{d}{d \choose j}\\
 & =8^{d}\max_{j=0,1,\dots,d}\left|p\left(\frac{j}{d}\right)\right|,
\end{align*}
where the first step uses~(\ref{eqn:p-lin-comb-q}), (\ref{eq:qj-Pi-norm}),
and Fact~\ref{fact:poly-norm}.
\end{proof}

\subsection{Coefficient bounds for multivariate polynomials}

Let $\phi\colon\Re^{n}\to\Re$ be a multivariate polynomial. Analogous
to the univariate case, we let $\norm\phi$ denote the sum of the
absolute values of the coefficients of $\phi.$ Fact~\ref{fact:poly-norm}
is clearly valid in this multivariate setting as well. Recall that
a multivariate polynomial $\phi$ is \emph{multilinear} if it has
degree at most~$1$ in each variable. The following result is an
analogue of Lemma~\ref{lem:bound-on-poly-coeffs-UNIVARIATE}.
\begin{lem}
\label{lem:bound-on-poly-coeffs-MULTIVARIATE}Let $\phi\colon\Re^{n}\to\Re$
be a symmetric multilinear polynomial. Then 
\[
\norm\phi\leq8^{\deg\phi}\max_{x\in\zoon}|\phi(x)|.
\]
\end{lem}

\begin{proof}
Abbreviate $d=\deg\phi$ and write
\[
\phi(x)=\sum_{i=0}^{d}a_{i}\sum_{S\in\binom{[n]}{i}}\;\prod_{j\in S}x_{j},
\]
where $a_{0},a_{1},\ldots,a_{d}$ are real coefficients. For $0\leq t\leq1,$
let $B(t)$ denote the Bernoulli distribution with success probability
$t.$ Then
\begin{align*}
\norm\phi & =\sum_{i=0}^{d}|a_{i}|\binom{n}{i}\\
 & \leq8^{d}\max_{0\leq t\leq1}\left|\sum_{i=0}^{d}a_{i}\binom{n}{i}t^{i}\right|\\
 & =8^{d}\max_{0\leq t\leq1}\left|\Exp_{x_{1},x_{2},\ldots,x_{n}\sim B(t)}\phi(x)\right|\\
 & \leq8^{d}\max_{x\in\zoon}|\phi(x)|,
\end{align*}
where the second and third steps use Lemma~\ref{lem:bound-on-poly-coeffs-UNIVARIATE}
and multilinearity, respectively.
\end{proof}
The following lemma, due to Razborov and Sherstov~\cite[Lemma~3.2]{RS07dc-dnf},
bounds the value of a polynomial $p$ at a point of large Hamming
weight in terms of $p$'s values at points of low Hamming weight.
\begin{lem}[Extrapolation lemma]
\label{lem:extrapolation}Let $d$ be an integer, $0\leq d\leq n-1.$
Let $\phi:\Re^{n}\to\Re$ be a polynomial of degree at most $d.$
Then 
\[
|\phi(1^{n})|\;\leq\;2^{d}{n \choose d}\,\max_{x\in\zoo_{\leq d}^{n}}|\phi(x)|.
\]
\end{lem}

\noindent As one would expect, one can sharpen the bound of Lemma~\ref{lem:extrapolation}
by maximizing over a larger neighborhood of the Boolean hypercube
than $\zoo_{\leq d}^{n}$. The resulting bound is as follows.
\begin{lem}[Generalized extrapolation lemma]
\label{lem:extrapolation-generalized}Fix positive integers $N>m\geq d.$
Let $\phi\colon\Re^{N}\to\Re$ be a polynomial of degree at most $d.$
Then
\begin{align*}
|\phi(x^{*})| & \leq2^{d}\binom{\lceil|x^{*}|/\lfloor m/d\rfloor\rceil}{d}\max_{x\in\zoo_{\leq m}^{N}}|\phi(x)|, &  & x^{*}\in\zoo_{>m}^{N}.
\end{align*}
\end{lem}

\noindent One recovers Lemma~\ref{lem:extrapolation} as a special
case by taking $N=n,\;m=d,\;$ and $x^{*}=1^{n}.$ 
\begin{proof}[Proof of Lemma~\emph{\ref{lem:extrapolation-generalized}}.]
Consider an arbitrary vector $x^{*}\in\zoo^{N}$ of Hamming weight
$|x^{*}|>m,$ and abbreviate $n=\lceil|x^{*}|/\lfloor m/d\rfloor\rceil.$
Let $S_{1},S_{2},\ldots,S_{n}$ be a partition of $\{i:x_{i}^{*}=1\}$
such that $|S_{i}|\leq\lfloor m/d\rfloor$ for all $i.$ Observe that
\begin{equation}
n>d.\label{eq:extrapolation-n-d}
\end{equation}
Define $L\colon\zoon\to\zoo^{N}$ by
\[
L(z)=\sum_{i=1}^{n}z_{i}\1_{S_{i}}.
\]
Then clearly
\begin{align}
 & L(1^{n})=x^{*},\label{eq:L-x-star}\\
 & |L(z)|\leq|z|\cdot\left\lfloor \frac{m}{d}\right\rfloor .\label{eq:L-hamming-weight-grows}
\end{align}
Moreover, the mapping $z\mapsto\phi(L(z))$ is a real polynomial on
$\zoon$ of degree at most $\deg\phi\leq d.$ As a result,
\begin{align*}
|\phi(x^{*})| & =|\phi(L(1^{n}))|\\
 & \leq2^{d}\binom{n}{d}\max_{|z|\leq d}|\phi(L(z))|\\
 & \leq2^{d}\binom{n}{d}\max_{|x|\leq d\lfloor m/d\rfloor}|\phi(x)|\\
 & \leq2^{d}\binom{n}{d}\max_{|x|\leq m}|\phi(x)|,
\end{align*}
where the first step uses~(\ref{eq:L-x-star}); the second step follows
by (\ref{eq:extrapolation-n-d}) and Lemma~\ref{lem:extrapolation};
and the third step is valid by~(\ref{eq:L-hamming-weight-grows}).
\end{proof}

\subsection{The conjunction norm}

Recall that a \emph{conjunction }in Boolean variables $x_{1},x_{2},\ldots,x_{n}$
is the AND of some subset of the literals $x_{1},\overline{x_{1}},x_{2},\overline{x_{2}},\ldots,x_{n},\overline{x_{n}}.$
Analogously, a \emph{disjunction }is the OR of some subset of $x_{1},\overline{x_{1}},x_{2},\overline{x_{2}},\ldots,x_{n},\overline{x_{n}}.$
We regard conjunctions and disjunctions as Boolean functions $\zoon\to\zoo$
and in particular as a special case of real functions $\zoon\to\Re.$
For a subset $X\subseteq\zoon$ and a function $f\colon X\to\Re$,
we define the \emph{conjunction norm} $\orcomplexity(f)$ to be the
minimum $\Lambda\geq0$ such that
\begin{align*}
f(x) & =\lambda_{1}C_{1}(x)+\lambda_{2}C_{2}(x)+\cdots+\lambda_{N}C_{N}(x) &  & (x\in X)
\end{align*}
for some integer $N,$ some conjunctions $C_{1},C_{2},\ldots,C_{N},$
and some real coefficients $\lambda_{1},\lambda_{2},\ldots,\lambda_{N}$
with $|\lambda_{1}|+|\lambda_{2}|+\cdots+|\lambda_{N}|\leq\Lambda.$
Our choice of the symbol $\Pi,$ for ``product,'' is motivated by
the view of conjunctions as products of literals. In particular, we
have $\orcomplexity(\phi)\leq\norm{\phi}$ for any multivariate polynomial
$\phi\colon\zoon\to\Re.$ The next proposition shows that $\orcomplexity$
is a norm on the space of multivariate real functions and establishes
other useful properties of this complexity measure.
\begin{prop}[Conjunction norm]
\label{prop:orcomplexity}Let $f,g\colon X\to\Re$ be given functions,
for a nonempty set $X\subseteq\zoon$. Then:
\begin{enumerate}
\item \label{item:orcomplexity-nonzero}$\orcomplexity(f)\geq0,$ with equality
if and only if $f=0;$
\item \label{item:orcomplexity-homogeneous}$\orcomplexity(\lambda f)=|\lambda|\orcomplexity(f)$
for any real $\lambda;$
\item \label{item:orcomplexity-triangle}$\orcomplexity(f+g)\leq\orcomplexity(f)+\orcomplexity(g);$
\item \label{item:orcomplexity-product}$\orcomplexity(f\cdot g)\leq\orcomplexity(f)\orcomplexity(g);$
\item \label{item:orcomplexity-trivial-bound}$\orcomplexity(f)\leq\|f\|_{1};$
\item \label{item:orcomplexity-disjunction}$\orcomplexity(f)\leq2$ if
$f$ is a disjunction$;$ 
\item \label{item:orcomplexity-composition-with-univariate}$\orcomplexity(p\circ f)\leq\max\{1,\orcomplexity(f)\}^{d}\,\norm p$
for any polynomial $p\in P_{d}.$
\end{enumerate}
\end{prop}

\begin{proof}
\ref{item:orcomplexity-nonzero}\textendash \ref{item:orcomplexity-triangle}
Immediate from the definitions. 

\ref{item:orcomplexity-product} Express $f$ and $g$ individually
as a linear combination of conjunctions with real coefficients whose
absolute values sum to $\orcomplexity(f)$ and $\Pi(g)$, respectively.
Then, multiply these two linear combinations. Since the product of
conjunctions is again a conjunction, the resulting representation
is a linear combination of conjunctions with real coefficients whose
absolute values sum to at most $\orcomplexity(f)\orcomplexity(g).$

\ref{item:orcomplexity-trivial-bound} By the homogeneity~\ref{item:orcomplexity-homogeneous}
and triangle inequality~\ref{item:orcomplexity-triangle}, we have
\begin{align*}
\orcomplexity(f) & =\orcomplexity\left(\sum_{a\in X}f(a)C_{a}\right)\\
 & \leq\sum_{a\in X}|f(a)|\orcomplexity(C_{a})\\
 & \leq\sum_{a\in X}|f(a)|\\
 & =\|f\|_{1},
\end{align*}
where $C_{a}$ denotes the conjunction that evaluates to true on $a$
and to false on all other inputs in $\zoon.$

\ref{item:orcomplexity-disjunction} We have $\orcomplexity(f)\leq\orcomplexity(f-1)+\orcomplexity(1)=\orcomplexity(1-f)+\orcomplexity(1)\leq2,$
where the first step applies the triangle inequality~\ref{item:orcomplexity-triangle},
the second step uses the homogeneity~\ref{item:orcomplexity-homogeneous},
and the third step uses the fact that $1$ and $1-f$ are conjunctions.

\ref{item:orcomplexity-composition-with-univariate} Let $p(t)=a_{d}t^{d}+a_{d-1}t^{d-1}+\cdots+a_{1}t+a_{0}$
be a given polynomial. Then
\begin{align*}
\orcomplexity(p\circ f) & =\orcomplexity\left(\sum_{i=0}^{d}a_{i}\underbrace{f\cdot f\cdot\cdots\cdot f}_{i}\right)\\
 & \leq\sum_{i=0}^{d}|a_{i}|\orcomplexity(\underbrace{f\cdot f\cdot\cdots\cdot f}_{i})\\
 & \leq\sum_{i=0}^{d}|a_{i}|\orcomplexity(f)^{i}\allowdisplaybreaks\\
 & \leq\max\{1,\orcomplexity(f)^{d}\}\sum_{i=0}^{d}|a_{i}|\\
 & =\max\{1,\orcomplexity(f)\}^{d}\;\norm p,
\end{align*}
where the second step uses~\ref{item:orcomplexity-homogeneous} and~\ref{item:orcomplexity-triangle},
and the third step applies~\ref{item:orcomplexity-product}. 
\end{proof}

\section{\label{sec:extension}The extension theorem}

This section establishes an approximation-theoretic result of independent
interest, the \emph{extension theorem}, that we use several times
in the rest of the paper to construct approximating polynomials. To
set the stage for this result, let $f\colon\zoo_{\leq m}^{N}\to[-1,1]$
be a given function, defined on inputs of Hamming weight up to $m.$
For any integer $n>m,$ consider the extension $F_{n}$ of $f$ to
inputs of Hamming weight up to $n$, given by
\[
F_{n}(x)=\begin{cases}
f(x) & \text{if \ensuremath{|x|\leq m,}}\\
0 & \text{otherwise.}
\end{cases}
\]
From the point of view of approximation theory, a fundamental question
to ask is how to ``extend'' any approximant for $f$ to an approximant
for $F_{n},$ without degrading the quality of the approximation or
significantly increasing the approximant's degree. Ideally, we would
like the approximant for the extension $F_{n}$ to have degree within
a small factor of the original degree, e.g., a factor of $O(n/m)^{\alpha}$
for some constant $0<\alpha<1.$

Unfortunately, the extension problem is hopeless as stated. Indeed,
consider the special case of the constant function $f=1,$ so that
\[
F_{n}(x)=\begin{cases}
1 & \text{if }0\leq|x|\leq m,\\
0 & \text{if }m<|x|\leq n.
\end{cases}
\]
In this example, $\deg_{1/3}(f)=0$ but $\deg_{1/3}(F_{n})=\Omega(\sqrt{n})$
by a well-known result of Nisan and Szegedy~\cite{nisan-szegedy94degree}.
In particular, there is no efficient way to transform an approximant
for a general function $f$ into an approximant for the extension
$F_{n}.$ Our contribution is to show that the extension problem becomes
meaningful and efficiently solvable if one's starting point is an
approximant for $F_{2m}$ rather than for $f.$ In other words, we
give an efficient, black-box transformation of an approximant for
$F_{2m}$ into an approximant for any extension $F_{n},$ where $n\geq2m.$
The formal statement of our result is as follows. 
\begin{thm}[Extension theorem]
\label{thm:extension}Let $f\colon\zoo_{\leq m}^{N}\to[-1,1]$ be
given, where $N\geq m\geq0$ are integers. For integers $n\geq m,$
define $F_{n}\colon\zoo_{\leq n}^{N}\to[-1,1]$ by 
\[
F_{n}(x)=\begin{cases}
f(x) & \text{if \ensuremath{|x|\leq m,}}\\
0 & \text{otherwise.}
\end{cases}
\]
Then for some absolute constant $C>1$ and all $\epsilon,\delta\in(0,1/2)$
and $n\geq m,$ 
\begin{equation}
\deg_{\epsilon+\delta}(F_{n})\leq C\sqrt{\frac{n}{m+1}}\cdot\left(\deg_{\epsilon}(F_{2m})+\log\frac{1}{\delta}\right).\label{eq:degeps-Fn-advertise}
\end{equation}
\end{thm}

\noindent Theorem~\ref{thm:extension} solves the extension problem
with only a factor-$\sqrt{n/m}$ increase in degree. The approximation
quality of the new approximant can be made arbitrarily close to that
of the original at a small additive cost in degree. This overhead
in degree and error is optimal, as we will discover in applications
later in this paper. We also note that the constant $2$ in this result
was chosen exclusively for aesthetic reasons, and (\ref{eq:degeps-Fn-advertise})
holds with $F_{2m}$ replaced by $F_{\lceil cm\rceil}$ for any constant
$c>1.$ The rest of this section is devoted to the proof of Theorem~\ref{thm:extension}.

\subsection{Proof strategy}

In the notation of Theorem~\ref{thm:extension}, let $p_{2m}(x)$
be an approximant for $F_{2m}(x).$ Then clearly 
\begin{equation}
F_{n}(x)\approx p_{2m}(x)\cdot\I[|x|\leq2m]\label{eq:extension-proof-strategy}
\end{equation}
 on the domain of $F_{n},$ where $\I[|x|\leq2m]$ is the characteristic
function of the set of inputs of Hamming weight at most $2m.$ While
$p_{2m}(x)$ can grow rapidly as the Hamming weight $|x|$ increases
beyond $2m,$ that growth is not entirely arbitrary. Specifically,
the generalized extrapolation lemma (Lemma~\ref{lem:extrapolation-generalized})
bounds $|p_{2m}(x)|$ in terms of the Hamming weight $|x|$ and the
degree of $p_{2m}$. In particular, the approximate equality~(\ref{eq:extension-proof-strategy})
is preserved if $\I[|x|\leq2m]$ is replaced by a low-degree approximant.
The construction of such an approximant is the crux of our proof.
More precisely, we construct a low-degree \emph{univariate }approximant
to the characteristic function of any interval. To crystallize our
approach, we first consider the degenerate interval $[0,0]=\{0\}.$
\begin{prop}
\label{prop:ext-theorem-strategy}For any positive integers $n$ and
$d,$ there is a polynomial $p$ with
\begin{align}
 & p(0)=1,\label{eq:or-rapid-decrease-at-0}\\
 & |p(t)|\leq\frac{1}{t^{d}}, &  & t\in[1,n],\label{eq:or-rapid-decrease-at-t}\\
 & \deg p\leq7d\sqrt{n}.\label{eq:or-rapid-decrease-deg}
\end{align}
\end{prop}

\noindent The key property here is (\ref{eq:or-rapid-decrease-at-t}),
whereby the approximating polynomial gets smaller as one moves farther
away from the point of interest, $0$. Reproducing this behavior in
the context of a general interval is much more subtle and is the subject
of Sections~\ref{subsec:reciprocal}\textendash \ref{subsec:characteristic-fn-interval}.
\begin{proof}
Define
\[
T(t)=\left(\prod_{i=0}^{\lceil\log n\rceil}T_{\lceil\sqrt{n/2^{i}}\rceil}\left(1+\frac{2^{i}-t}{n}\right)\right)^{d}.
\]
Fix an arbitrary point $t\in[1,n]$, and let $j$ be the integer such
that $t\in[2^{j},2^{j+1}).$ Then
\begin{align*}
|T(t)| & =\prod_{i=0}^{\lceil\log n\rceil}\left|T_{\lceil\sqrt{n/2^{i}}\rceil}\left(1+\frac{2^{i}-t}{n}\right)\right|^{d}\\
 & \leq\prod_{i=j+1}^{\lceil\log n\rceil}\left|T_{\lceil\sqrt{n/2^{i}}\rceil}\left(1+\frac{2^{i}-t}{n}\right)\right|^{d}\\
 & \leq\prod_{i=j+1}^{\lceil\log n\rceil}\left|T_{\lceil\sqrt{n/2^{i}}\rceil}\left(1+\frac{2^{i}}{n}\right)\right|^{d}\allowdisplaybreaks\\
 & =|T(0)|\;\prod_{i=0}^{j}\left|T_{\lceil\sqrt{n/2^{i}}\rceil}\left(1+\frac{2^{i}}{n}\right)\right|^{-d}\\
 & \leq|T(0)|\;\prod_{i=0}^{j}2^{-d}\\
 & \leq\frac{|T(0)|}{t^{d}},
\end{align*}
where the second step uses~(\ref{eqn:chebyshev-containment}), the
third step follows from~(\ref{eq:chebyshev-at-1}) and Fact~\ref{fact:chebyshev-derivative},
and the next-to-last step applies Proposition~\ref{prop:chebyshev-beyond-1}.
Moreover,
\begin{align*}
\deg T & =d\sum_{i=0}^{\lceil\log n\rceil}\left\lceil \sqrt{\frac{n}{2^{i}}}\right\rceil \\
 & \leq d\sum_{i=0}^{\infty}\sqrt{\frac{n}{2^{i}}}\cdot2\\
 & \leq7d\sqrt{n}.
\end{align*}
As a result,~(\ref{eq:or-rapid-decrease-at-0})\textendash (\ref{eq:or-rapid-decrease-deg})
hold for $p(t)=T(t)/T(0).$
\end{proof}

\subsection{\label{subsec:reciprocal}Approximating 1/\emph{t}}

To handle actual intervals rather than singleton points, we need to
develop a number of auxiliary results. To start with, we construct
an approximant for the reciprocal function $1/t$ on $[1,n]$. We
are specifically interested in approximation within a \emph{multiplicative}
factor close to $1$, which is a more demanding regime than pointwise
approximation. 
\begin{lem}
\label{lem:reciprocal}For any integer $d\geq0$ and real $n>1,$
there is an $($explicitly given$)$ polynomial $p\in P_{d}$ such
that
\begin{align}
\frac{1-\epsilon}{t} & \leq p(t)\leq\frac{1+\epsilon}{t}, &  & 1\leq t\leq n,\label{eq:approximate-t-multiplicatively}
\end{align}
where 
\[
\epsilon=\frac{1}{T_{d+1}\bigl(\frac{n+1}{n-1}\bigr)}.
\]
\end{lem}

\begin{proof}
Property~(\ref{eq:approximate-t-multiplicatively}) can be restated
as $\max_{1\leq t\leq n}|1-tp(t)|\leq\epsilon$. Thus, the existence
of $p\in P_{d}$ that obeys~(\ref{eq:approximate-t-multiplicatively})
is equivalent to the existence of $q\in P_{d+1}$ that obeys $q(0)=1$
and $\max_{1\leq t\leq n}|q(t)|\leq\epsilon$. Now, define $q\in P_{d+1}$
by
\[
q(t)=\frac{T_{d+1}\left(1-2\cdot\frac{t-1}{n-1}\right)}{T_{d+1}\bigl(\frac{n+1}{n-1}\bigr)}.
\]
Then $q(0)=1$ by definition. Moreover,
\begin{align*}
\max_{1\leq t\leq n}|q(t)| & \leq\max_{-1\leq t\leq1}\frac{|T_{d+1}(t)|}{\bigl|T_{d+1}\bigl(\frac{n+1}{n-1}\bigr)\bigr|}\\
 & \leq\frac{1}{\bigl|T_{d+1}\bigl(\frac{n+1}{n-1}\bigr)\bigr|}\\
 & =\frac{1}{T_{d+1}\bigl(\frac{n+1}{n-1}\bigr)},
\end{align*}
where the last two steps use~(\ref{eqn:chebyshev-containment}) and~(\ref{eq:chebyshev-beyond-1}),
respectively.
\end{proof}
It is well known~\cite[Theorem~1.10]{rivlin-book} that among all
polynomials of degree at most $d$ that are bounded on $[-1,1]$ in
absolute value by $1,$ the Chebyshev polynomial $T_{d}$ takes on
the largest possible value at every point of $[1,\infty).$ Using
this fact, it is straightforward to verify that Lemma~\ref{lem:reciprocal}
gives the best possible bound on $\epsilon$ in terms of $n$ and
$d.$
\begin{cor}
\label{cor:reciprocal}For any real $n>1,$ there is an $($explicitly
given$)$ univariate polynomial $p$ of degree at most $\sqrt{2(n-1)}$
such that 
\begin{align*}
 & \frac{1}{2t}\leq p(t)\leq\frac{1}{t}, &  & 1\leq t\leq n.
\end{align*}
\end{cor}

\begin{proof}
By Proposition~\ref{prop:chebyshev-beyond-1},
\begin{align*}
T_{\lfloor\sqrt{2(n-1)}\rfloor+1}\left(\frac{n+1}{n-1}\right) & \geq5.
\end{align*}
As a result, it suffices to invoke Lemma~\ref{lem:reciprocal} with
$d=\lfloor\sqrt{2(n-1)}\rfloor$.
\end{proof}

\subsection{Approximating 1/\emph{t}\protect\textsuperscript{\emph{i}}}

We now construct approximants for powers of the reciprocal function,
focusing this time on absolute rather than relative error. Here, we
are interested only in approximation in the neighborhood of $1$.
In the following construction, increasing the approximant's degree
makes the neighborhood larger and the approximation more accurate.
\begin{lem}
\label{lem:powers-of-reciprocal}Let $d\geq1$ be a given integer.
Then for every integer $D\geq0,$ there is an $($explicitly given$)$
polynomial $p$ with
\begin{align}
 & \left|\frac{1}{t^{d}}-p(t)\right|\leq|1-t|^{D+1}\binom{D+d}{d}d, &  & t\in\left[\frac{d}{d+D},2-\frac{d}{d+D}\right],\label{eq:inverter-error}\\
 & |p(t)|\leq\binom{D+d}{d}, &  & t\in[0,2],\label{eq:inverter-norm}\\
 & \deg p\leq D.\label{eq:inverter-degree}
\end{align}
\end{lem}

\begin{proof}
Define
\[
p(t)=\sum_{i=0}^{D}\binom{i+d-1}{i}(1-t)^{i}.
\]
Then~(\ref{eq:inverter-degree}) is immediate. For~(\ref{eq:inverter-norm}),
it suffices to observe that
\begin{align*}
\sum_{i=0}^{D}\binom{i+d-1}{i} & =\binom{D+d}{D}\\
 & =\binom{D+d}{d},
\end{align*}
where the first equality is well-known and can be verified by using
Pascal's triangle or by interpreting the left-hand side as the number
of ways to distribute at most $D$ identical balls into $d$ distinct
bins.

It remains to settle~(\ref{eq:inverter-error}). For $0<t<2,$ we
have the Maclaurin expansion
\begin{align*}
\frac{1}{t^{d}} & =\left(\sum_{i=0}^{\infty}(1-t)^{i}\right)^{d}\\
 & =\sum_{i=0}^{\infty}\binom{i+d-1}{i}(1-t)^{i}.
\end{align*}
Therefore,
\begin{align*}
\left|\frac{1}{t^{d}}-p(t)\right| & =\left|\sum_{i=D+1}^{\infty}\binom{i+d-1}{i}(1-t)^{i}\right|\\
 & \leq\sum_{i=D+1}^{\infty}\binom{i+d-1}{i}|1-t|^{i}\allowdisplaybreaks\\
 & \le|1-t|^{D+1}\binom{D+d}{D+1}\sum_{i=0}^{\infty}\left(\frac{D+d}{D+1}\right)^{i}|1-t|^{i}\allowdisplaybreaks\\
 & \le|1-t|^{D+1}\binom{D+d}{D+1}\sum_{i=0}^{\infty}\left(\frac{D}{D+1}\right)^{i}\\
 & =|1-t|^{D+1}\binom{D+d}{d}d,
\end{align*}
where the fourth step is legitimate in view of the range of $t$ in~(\ref{eq:inverter-error}).
\end{proof}

\subsection{\label{subsec:characteristic-fn-interval}Approximating the characteristic
function of an interval}

The following lemma is the last prerequisite to our construction of
a low-degree approximant for the characteristic function of an interval.
Without loss of generality, it suffices to consider the interval $[0,1].$
The lemma below \emph{almost} solves our problem except that it gives
a flat bound on the approximant's value outside the interval, not
taking into account how far one is from the interval.
\begin{lem}
\label{lem:or-continuous}For any reals $n\geq1$ and $0<\epsilon<1/2,$
there is an $($explicitly given$)$ univariate polynomial $p$ such
that
\begin{align}
 & |p(t)-1|\leq\epsilon, &  & t\in[0,1],\label{eq:or-continuous-p-starts-large}\\
 & |p(t)|\leq1, &  & t\in(1,2],\label{eq:or-continuous-p-is-bounded}\\
 & |p(t)|\leq\epsilon, &  & t\in(2,n],\label{eq:or-continuous-p-ends-small}\\
 & \deg p=O\left(\sqrt{n}\log\frac{1}{\epsilon}\right).\label{eq:or-continuous-p-degree-logeps-factor-above-chebyshev}
\end{align}
\end{lem}

\begin{proof}
For $n<2,$ the lemma holds trivially with $p=1.$ In what follows,
we treat the complementary case $n\geq2.$ Consider the univariate
polynomial
\[
q(t)=T_{\lceil\sqrt{n}\rceil}\left(1+\frac{2-t}{n}\right).
\]
Using $n\geq2,$ we obtain
\begin{align}
q([0,n]) & \subseteq\left[-1,T_{\lceil\sqrt{n}\rceil}\left(1+\frac{2}{n}\right)\right]\nonumber \\
 & \subseteq\left[-1,\left(1+\frac{2}{n}+\sqrt{\left(1+\frac{2}{n}\right)^{2}-1}\right)^{\sqrt{n}+1}\right]\nonumber \\
 & \subseteq\left[-1,\left(1+\frac{2}{n}+\sqrt{\frac{6}{n}}\right)^{\sqrt{n}+1}\right]\nonumber \\
 & \subset\left[-1,\exp\left(\left(\frac{2}{n}+\sqrt{\frac{6}{n}}\right)\left(\sqrt{n}+1\right)\right)\right]\nonumber \\
 & \subset[-1,\e^{7}-1],\label{eq:q-global-bound}
\end{align}
where the first step is legitimate in view of~(\ref{eqn:chebyshev-containment}),
(\ref{eq:chebyshev-at-1}), and Fact~\ref{fact:chebyshev-derivative};
and the second step uses~(\ref{eq:chebyshev-beyond-1}). By Proposition~\ref{prop:chebyshev-beyond-1},
\begin{align}
\min_{0\leq t\leq1}q(t) & =\min_{1\leq t\leq2}\;T_{\lceil\sqrt{n}\rceil}\left(1+\frac{t}{n}\right)\nonumber \\
 & \geq2.\label{eq:q-large-on-good-set}
\end{align}
By~(\ref{eqn:chebyshev-containment}),
\begin{align}
\max_{2\leq t\leq n}|q(t)| & \leq\max_{0\leq t\leq1}\;|T_{\lceil\sqrt{n}\rceil}(t)|\nonumber \\
 & \leq1.\label{eq:q-small-on-bad-set}
\end{align}
In view of~(\ref{eq:q-global-bound})\textendash (\ref{eq:q-small-on-bad-set}),
the normalized polynomial $q^{*}(t)=(q(t)+1)/\e^{7}$ obeys
\begin{align}
q^{*}([0,n]) & \subseteq[0,1],\label{eq:q-star-range}\\
q^{*}([0,1]) & \subseteq[3\e^{-7},1],\\
q^{*}([2,n]) & \subseteq[0,2\e^{-7}].
\end{align}

To complete the proof, we use a technique due to Buhrman et al.~\cite{BNRW05robust}.
Consider the univariate polynomial
\[
B_{d}(t)=\sum_{i=\lceil2.5\,\e^{-7}d\rceil}^{d}\binom{d}{i}t^{i}(1-t)^{i}.
\]
In words, $B_{d}(t)$ is the probability of observing at least $2.5\,\e^{-7}d$
heads in a sequence of $d$ independent coin flips, each coming up
heads with probability $t.$ For large enough $d=O(\log(1/\epsilon)),$
the Chernoff bound guarantees that
\begin{align}
B_{d}([0,1]) & \subseteq[0,1],\\
B_{d}([0,2\e^{-7}]) & \subseteq[0,\epsilon],\\
B_{d}([3\e^{-7},1]) & \subseteq[1-\epsilon,1].\label{eq:B-d-large}
\end{align}
Now define $p(t)=B_{d}(q^{*}(t)).$ Then the degree bound~(\ref{eq:or-continuous-p-degree-logeps-factor-above-chebyshev})
is immediate, whereas the remaining properties~(\ref{eq:or-continuous-p-starts-large})\textendash (\ref{eq:or-continuous-p-ends-small})
follow from~(\ref{eq:q-star-range})\textendash (\ref{eq:B-d-large}).
\end{proof}
Finally, we are now in a position to construct the desired approximant
for the characteristic function of an interval. As mentioned above,
we may without loss of generality focus on the interval $[0,1].$
\begin{thm}
\label{thm:sigmoid}For all integers $n,d\geq0$ and all $0<\epsilon<1/2,$
there is an $($explicitly given$)$ univariate polynomial $p$ such
that
\begin{align}
 & |p(t)-1|\leq\epsilon, &  & t\in[0,1],\label{eq:sigmoid-starts-small}\\
 & |p(t)|\leq1+\epsilon, &  & t\in(1,2],\label{eq:sigmoid-bounded}\\
 & |p(t)|\leq\frac{\epsilon}{t^{d}} &  & t\in(2,n],\label{eq:sigmoid-ends-small}\\
 & \deg p=O\left(\sqrt{n}\left(d+\log\frac{1}{\epsilon}\right)\right).\label{eq:sigmoid-deg}
\end{align}
\end{thm}

\begin{proof}
For $n<2,$ the theorem holds trivially by taking $p=1.$ In what
follows, we focus on the complementary case $n\geq2.$

Corollary~\ref{cor:reciprocal} gives an explicit univariate polynomial
$p_{1}$ such that
\begin{align}
 & \frac{1}{2(t+1)}\leq p_{1}(t)\leq\frac{1}{t+1}, &  & 0\leq t\leq n,\label{eq:reciprocal-error-restated}\\
 & \deg p_{1}\leq\sqrt{2n}.\label{eq:reciprocal-deg-restated}
\end{align}
Let $D$ be an integer parameter to be chosen later, $D>5d.$ Then
Lemma~\ref{lem:powers-of-reciprocal} provides an explicit polynomial
$p_{2}$ such that
\begin{align}
 & \left|\frac{1}{t^{d}}-p_{2}(t)\right|\leq\left(\frac{5}{6}\right)^{D+1}\binom{D+d}{d}d, &  & t\in\left[\frac{1}{6},1\right],\label{eq:inverter-error-restated}\\
 & |p_{2}(t)|\leq\binom{D+d}{d}, &  & t\in[0,2],\label{eq:inverter-norm-restated}\\
 & \deg p_{2}\leq D.\label{eq:inverter-degree-restated}
\end{align}
As our last building block, Lemma~\ref{lem:or-continuous} constructs
an explicit polynomial $p_{3}$ with
\begin{align}
 & |p_{3}(t)-1|\leq\epsilon2^{-D-d}, &  & t\in[0,1],\label{eq:or-continuous-p-starts-large-restated}\\
 & |p_{3}(t)|\leq1, &  & t\in(1,2],\label{eq:or-continuous-p-is-bounded-restated}\\
 & |p_{3}(t)|\leq\epsilon2^{-D-d}, &  & t\in(2,n],\label{eq:or-continuous-p-ends-small-restated}\\
 & \deg p_{3}=O\left(\sqrt{n}\left(D+d+\log\frac{1}{\epsilon}\right)\right).\label{eq:or-continuous-p-degree-logeps-factor-above-chebyshev-restated}
\end{align}
In the rest of the proof, we will show that the conclusion of the
theorem holds for the polynomial
\[
p(t)=p_{1}(t)^{d}p_{2}(p_{1}(t))p_{3}(t).
\]

To begin with,
\begin{align}
\max_{0\leq t\leq1}|p(t)-1| & \leq\max_{0\leq t\leq1}(1+|p_{1}(t)^{d}p_{2}(p_{1}(t))-1|)\cdot(1+|1-p_{3}(t)|)-1\nonumber \\
 & \leq\left(1+\frac{\epsilon}{2}\right)\max_{0\leq t\leq1}(1+|p_{1}(t)^{d}p_{2}(p_{1}(t))-1|)-1\nonumber \\
 & \leq\left(1+\frac{\epsilon}{2}\right)\left(1+\max_{1/4\leq t\leq1}|t^{d}p_{2}(t)-1|\right)-1\nonumber \\
 & \leq\left(1+\frac{\epsilon}{2}\right)\left(1+\max_{1/4\leq t\leq1}\left|p_{2}(t)-\frac{1}{t^{d}}\right|\right)-1\nonumber \\
 & \leq\left(1+\frac{\epsilon}{2}\right)\left(1+\left(\frac{5}{6}\right)^{D+1}\binom{D+d}{d}d\right)-1,\label{eq:sigmoid-starts-small-restated}
\end{align}
where the first step uses the inequality $|ab-1|\leq(1+|a-1|)(1+|b-1|)-1$
for any real $a,b;$ the second step is valid by~(\ref{eq:or-continuous-p-starts-large-restated});
the third applies~(\ref{eq:reciprocal-error-restated}); and the
final step is legitimate by~(\ref{eq:inverter-error-restated}).
Continuing,
\begin{align}
\max_{1\leq t\leq2}|p(t)| & =\max_{1\leq t\leq2}|p_{1}(t)^{d}p_{2}(p_{1}(t))p_{3}(t)|\nonumber \\
 & \leq\max_{1\leq t\leq2}|p_{1}(t)^{d}p_{2}(p_{1}(t))|\nonumber \\
 & \leq\max_{1/6\leq t\leq1/2}|t^{d}p_{2}(t)|\nonumber \\
 & \leq\max_{1/6\leq t\leq1/2}|t^{d}p_{2}(t)-1|+1\nonumber \\
 & \leq\max_{1/6\leq t\leq1/2}\left|p_{2}(t)-\frac{1}{t^{d}}\right|+1\nonumber \\
 & \leq1+\left(\frac{5}{6}\right)^{D+1}\binom{D+d}{d}d,\label{eq:sigmoid-bounded-restated}
\end{align}
where the second step uses~(\ref{eq:or-continuous-p-is-bounded-restated}),
the third step applies~(\ref{eq:reciprocal-error-restated}), the
fourth step is immediate from the triangle inequality, and the last
step follows from~(\ref{eq:inverter-error-restated}). Moreover,
\begin{align}
\max_{2\leq t\leq n}|t^{d}p(t)| & \leq\max_{2\leq t\leq n}|t^{d}p_{1}(t)^{d}|\cdot\max_{2\leq t\leq n}|p_{2}(p_{1}(t))|\cdot\max_{2\leq t\leq n}|p_{3}(t)|\nonumber \\
 & \leq\max_{2\leq t\leq n}|t^{d}p_{1}(t)^{d}|\cdot\max_{2\leq t\leq n}|p_{2}(p_{1}(t))|\cdot\epsilon2^{-D-d}\nonumber \\
 & \leq\max_{2\leq t\leq n}|p_{2}(p_{1}(t))|\cdot\epsilon2^{-D-d}\nonumber \\
 & \leq\max_{0\leq t\leq1/3}|p_{2}(t)|\cdot\epsilon2^{-D-d}\nonumber \\
 & \le\binom{D+d}{d}\cdot\epsilon2^{-D-d}\nonumber \\
 & \leq\epsilon,\label{eq:sigmoid-ends-small-restated}
\end{align}
where the second step is legitimate by (\ref{eq:or-continuous-p-ends-small-restated}),
the third and fourth steps use~(\ref{eq:reciprocal-error-restated}),
and the fifth step is immediate from~(\ref{eq:inverter-norm-restated}).
Finally, (\ref{eq:reciprocal-deg-restated}), (\ref{eq:inverter-degree-restated}),
and (\ref{eq:or-continuous-p-degree-logeps-factor-above-chebyshev-restated})
imply that
\begin{equation}
\deg p=O\left(\sqrt{n}\left(D+d+\log\frac{1}{\epsilon}\right)\right).\label{eq:sigmoid-deg-restated}
\end{equation}
Now the claimed bounds~(\ref{eq:sigmoid-starts-small})\textendash (\ref{eq:sigmoid-deg})
in the theorem statement follow immediately from~(\ref{eq:sigmoid-starts-small-restated})\textendash (\ref{eq:sigmoid-deg-restated})
by taking
\begin{align*}
D & =c\left\lceil d+\log\frac{1}{\epsilon}\right\rceil 
\end{align*}
for a sufficiently large absolute constant $c>1.$ 
\end{proof}

\subsection{Proof of the extension theorem}

Using the approximant constructed in Theorem~\ref{thm:sigmoid},
we now prove the extension theorem. We restate it below for the reader's
convenience.
\begin{thm*}[restatement of Theorem~\ref{thm:extension}]
Let $f\colon\zoo_{\leq m}^{N}\to[-1,1]$ be given, where $N\geq m\geq0$
are integers. For integers $n\geq m,$ define $F_{n}\colon\zoo_{\leq n}^{N}\to[-1,1]$
by 
\[
F_{n}(x)=\begin{cases}
f(x) & \text{if \ensuremath{|x|\leq m,}}\\
0 & \text{otherwise.}
\end{cases}
\]
Then for some absolute constant $C>1$ and all $\epsilon,\delta\in(0,1/2)$
and $n\geq m,$ 
\begin{equation}
\deg_{\epsilon+\delta}(F_{n})\leq C\sqrt{\frac{n}{m+1}}\cdot\left(\deg_{\epsilon}(F_{2m})+\log\frac{1}{\delta}\right).\label{eq:degeps-Fn}
\end{equation}
\end{thm*}
\begin{proof}
To simplify the presentation, we first settle two degenerate cases.
For $m=0,$ consider the polynomial
\[
T(t)=\left(T_{\lceil\sqrt{n}\rceil}\left(1+\frac{1-t}{n}\right)\right)^{\lceil\log\frac{1}{\delta}\rceil}.
\]
Then $T(0)\geq1/\delta$ by Proposition~\ref{prop:chebyshev-beyond-1},
and $\max_{1\leq t\leq n}|T(t)|\leq1$ by~(\ref{eqn:chebyshev-containment}).
Therefore, in this case $F_{n}$ is approximated pointwise within
$\delta$ by the degree-$O(\sqrt{n}\log(1/\delta))$ polynomial $F_{n}(0^{N})T(|x|)/T(0).$
Another degenerate possibility is $n\leq2m,$ in which case $\deg_{\epsilon}(F_{n})\leq\deg_{\epsilon}(F_{2m})$
and the theorem holds trivially. In what follows, we focus on the
general case when 
\begin{align*}
 & m\geq1,\\
 & n>2m.
\end{align*}

Abbreviate $d=\max\{\deg_{\epsilon}(F_{2m}),1\}$. By Fact~\ref{fact:trivial-upper-bound-on-deg},
\begin{equation}
1\leq d\leq2m.\label{eq:extension-thm-d-at-most-m}
\end{equation}
Fix a polynomial $\phi\colon\zoo^{N}\to\Re$ such that
\begin{align}
 & |F_{2m}(x)-\phi(x)|\leq\epsilon, &  & x\in\zoo_{\leq2m}^{N},\label{eq:extension-phi-approximates-f}\\
 & \deg\phi\leq d.\label{eq:extension-phi-deg}
\end{align}
Let $0<\alpha<1/2$ be a parameter to be chosen later. Then Theorem~\ref{thm:sigmoid}
gives an explicit univariate polynomial $p$ such that
\begin{align}
 & |p(t)-1|\leq\alpha, &  & t\in[0,1],\label{eq:sigmoid-starts-small-apply}\\
 & |p(t)|\leq1+\alpha, &  & t\in(1,2],\label{eq:sigmoid-bounded-apply}\\
 & |p(t)|\leq\frac{\alpha}{t^{d}} &  & t\in\left(2,\frac{n}{m}\right],\label{eq:sigmoid-ends-small-apply}\\
 & \deg p=O\left(\sqrt{\frac{n}{m}}\left(d+\log\frac{1}{\alpha}\right)\right).\label{eq:sigmoid-deg-apply}
\end{align}
Consider the polynomial $\Phi\colon\zoo^{N}\to\Re$ given by
\[
\Phi(x)=\phi(x)\;p\!\left(\frac{|x|}{m}\right).
\]
By~(\ref{eq:extension-phi-deg}) and~(\ref{eq:sigmoid-deg-apply}),
\begin{equation}
\deg\Phi=O\left(\sqrt{\frac{n}{m}}\left(d+\log\frac{1}{\alpha}\right)\right).\label{eq:deg-Phi}
\end{equation}
As the notation suggests, $\Phi$ is meant to be an extension of the
approximant $\phi$ to inputs $x\in\zoo^{N}$ of Hamming weight up
to $n.$ To analyze the accuracy of this new approximant, we will
examine three cases depending on the Hamming weight~$|x|.$ 

To start with,
\begin{align}
\max_{|x|\leq m}|F_{n}(x)-\Phi(x)| & =\max_{|x|\leq m}|F_{2m}(x)-\Phi(x)|\nonumber \\
 & \leq\max_{|x|\leq m}\{|F_{2m}(x)-\phi(x)|+|\phi(x)-\Phi(x)|\}\nonumber \\
 & \leq\epsilon+\max_{|x|\leq m}|\phi(x)-\Phi(x)|\nonumber \\
 & \leq\epsilon+\max_{|x|\leq m}|\phi(x)|\,\max_{0\leq t\leq1}|1-p(t)|\nonumber \\
 & \leq\epsilon+(1+\epsilon)\max_{0\leq t\leq1}|1-p(t)|\nonumber \\
 & \leq\epsilon+(1+\epsilon)\cdot\alpha,\label{eq:Phi-approximates-on-small-inputs}
\end{align}
where the third and fifth steps use~(\ref{eq:extension-phi-approximates-f}),
and the last step uses~(\ref{eq:sigmoid-starts-small-apply}). Continuing,
\begin{align}
\max_{m<|x|\leq2m}|F_{n}(x)-\Phi(x)| & =\max_{m<|x|\leq2m}|\Phi(x)|\nonumber \\
 & \leq\max_{m<|x|\leq2m}|\phi(x)|\;\max_{1<t\leq2}|p(t)|\nonumber \\
 & \leq\max_{m<|x|\leq2m}(|F_{2m}(x)|+\epsilon)\,\max_{1<t\leq2}|p(t)|\nonumber \\
 & \leq\epsilon\cdot(1+\alpha),\label{eq:Phi-approximation-med-inputs}
\end{align}
where the last two steps use~(\ref{eq:extension-phi-approximates-f})
and~(\ref{eq:sigmoid-bounded-apply}), respectively. Finally,
\begin{align}
\max_{2m<|x|\leq n} & |F_{n}(x)-\Phi(x)|\nonumber \\
 & =\max_{2m<|x|\leq n}|\Phi(x)|\nonumber \\
 & =\max_{2m<|x|\leq n}|\phi(x)|\;p\left(\frac{|x|}{m}\right)\nonumber \\
 & \leq\max_{2m<|x|\leq n}|\phi(x)|\cdot\alpha\cdot\left(\frac{m}{|x|}\right)^{d}\allowdisplaybreaks\nonumber \\
 & \leq\max_{2m<|x|\leq n}\left\{ 2^{d}\binom{\lceil|x|/\lfloor2m/d\rfloor\rceil}{d}\max_{|x'|\leq2m}|\phi(x')|\cdot\alpha\cdot\left(\frac{m}{|x|}\right)^{d}\right\} \nonumber \\
 & \leq\max_{2m<t\leq n}\left\{ 2^{d}\binom{\lceil t/\lfloor2m/d\rfloor\rceil}{d}(1+\epsilon)\cdot\alpha\cdot\left(\frac{m}{t}\right)^{d}\right\} \nonumber \\
 & \leq(4\e)^{d}(1+\epsilon)\cdot\alpha,\label{eq:Phi-approximation-large-inputs}
\end{align}
where the third step uses~(\ref{eq:sigmoid-ends-small-apply}), the
fourth step applies (\ref{eq:extension-thm-d-at-most-m}) and the
generalized extrapolation lemma (Lemma~\ref{lem:extrapolation-generalized}),
the fifth step follows from~(\ref{eq:extension-phi-approximates-f}),
and the last step uses~(\ref{eq:entropy-bound-binomial}) and~(\ref{eq:extension-thm-d-at-most-m}).
Now (\ref{eq:degeps-Fn}) follows from~(\ref{eq:deg-Phi})\textendash (\ref{eq:Phi-approximation-large-inputs})
by taking $\alpha=\delta(4\e)^{-d-1}.$
\end{proof}

\section{Symmetric functions}

In this section, we study the approximation of symmetric functions.
This class includes $\AND_{n}$ and $\OR_{n}$, which are fundamental
building blocks of our constructions in the rest of the paper. Our
result here is as follows.
\begin{thm}
\label{thm:SYMMETRIC}Let $f\colon\zoon\to[-1,1]$ be an arbitrary
symmetric function. Let $k$ be a nonnegative integer such that $f$
is constant on inputs of Hamming weight in $(k,n-k).$ Then for $0<\epsilon<1/2,$
\begin{equation}
\deg_{\epsilon}(f)=O\left(\sqrt{nk}+\sqrt{n\log\frac{1}{\epsilon}}\right).\label{eq:degeps-f-symmetric-1-1}
\end{equation}
Moreover, the approximating polynomial is given explicitly in each
case.
\end{thm}

\noindent Theorem~\ref{thm:SYMMETRIC} is tight~\cite{sherstov07inclexcl-ccc}
for every $\epsilon\in[1/2^{n},1/3]$ and every symmetric function
$f\colon\zoon\to\zoo$, with the obvious exception of the constant
functions $f=0$ and $f=1.$ Prior to our work, de Wolf~\cite{de-wolf08approx-degree}
proved the upper bound~(\ref{eq:degeps-f-symmetric-1-1}) by giving
an $\epsilon$-error quantum query algorithm for any symmetric function
$f.$ The novelty of Theorem~\ref{thm:SYMMETRIC} is the construction
of an explicit, closed-form approximating polynomial that achieves
de Wolf's upper bound. We give three proofs of Theorem~\ref{thm:SYMMETRIC},
corresponding to Sections 4.1\textendash 4.3 below.

\subsection{Approximation using the extension theorem}

Our first proof of Theorem~\ref{thm:SYMMETRIC} is based on the extension
theorem, and is the shortest of the three. The centerpiece of the
proof is the following technical lemma, in which we construct a closed-form
approximant for any function supported on inputs of low Hamming weight. 
\begin{lem}
\label{lem:small-support}Let $f\colon\zoon\to[-1,1]$ be given. Let
$k$ be a nonnegative integer such that $f(x)=0$ for $|x|>k.$ Then
for $0<\epsilon<1/2,$ 
\[
\deg_{\epsilon}(f)=O\left(\sqrt{nk}+\sqrt{n\log\frac{1}{\epsilon}}\right).
\]
Moreover, the approximating polynomial is given explicitly in each
case.
\end{lem}

\begin{proof}
Abbreviate
\[
m=\left\lceil k+\log\frac{1}{\epsilon}\right\rceil .
\]
If $m\geq n,$ the bound in the theorem statement follows trivially
from $\deg_{0}(f)\leq n.$ In the rest of the proof, we focus on the
complementary case $m<n.$

For $i\geq m,$ define $F_{i}\colon\zoon_{\leq i}\to[-1,1]$ by
\[
F_{i}(x)=\begin{cases}
f(x) & \text{if }|x|\leq m,\\
0 & \text{otherwise.}
\end{cases}
\]
Then
\begin{align*}
\deg_{\epsilon}(f) & =\deg_{\epsilon}(F_{n})\\
 & \leq\sqrt{\frac{n}{m}}\cdot O\left(\deg_{0}(F_{2m})+\log\frac{1}{\epsilon}\right)\\
 & \leq\sqrt{\frac{n}{m}}\cdot O\left(2m+\log\frac{1}{\epsilon}\right)\\
 & =O\left(\sqrt{nk}+\sqrt{n\log\frac{1}{\epsilon}}\right),
\end{align*}
where the first step uses $f=F_{n},$ the second step applies the
extension theorem~(Theorem~\ref{thm:extension}), and the third
step is valid by Fact~\ref{fact:trivial-upper-bound-on-deg}. Moreover,
the approximating polynomial is given explicitly because the extension
theorem and Fact~\ref{fact:trivial-upper-bound-on-deg} are fully
constructive.
\end{proof}
We are now in a position to prove the claimed result on the approximation
of arbitrary symmetric functions.
\begin{thm}
\label{thm:small-support}Let $f\colon\zoon\to[-1,1]$ be given. Let
$k$ be a nonnegative integer such that $f$ is constant on inputs
of Hamming weight in $(k,n-k).$ Then for $0<\epsilon<1/2,$ 
\[
\deg_{\epsilon}(f)=O\left(\sqrt{nk}+\sqrt{n\log\frac{1}{\epsilon}}\right).
\]
Moreover, the approximating polynomial is given explicitly in each
case.
\end{thm}

\noindent A powerful feature of Theorem~\ref{thm:small-support}
is that the function of interest is only assumed to be symmetric on
inputs of Hamming weight in $(k,n-k).$ In particular, Theorem~\ref{thm:small-support}
is significantly more general than Theorem~\ref{thm:SYMMETRIC}.
\begin{proof}[Proof of Theorem~\emph{\ref{thm:small-support}}.]
 If $k\geq n/2,$ the theorem follows from the trivial bound $\deg_{0}(f)\leq n.$
For the complementary case $k<n/2$, write 
\[
f(x_{1},\ldots,x_{n})=\lambda+f'(x_{1},\ldots,x_{n})+f''(\overline{x_{1}},\ldots,\overline{x_{n}}),
\]
where $\lambda\in[-1,1]$ and $f',f''\colon\zoon\to[-2,2]$ are functions
that vanish on $\zoon_{>k}.$ Then
\begin{align*}
\deg_{\epsilon}(f) & \leq\max\{\deg_{\epsilon/2}(f'),\deg_{\epsilon/2}(f'')\}\\
 & \leq\max\left\{ \deg_{\epsilon/4}\!\left(\frac{f'}{2}\right),\deg_{\epsilon/4}\!\left(\frac{f''}{2}\right)\right\} \\
 & =O\left(\sqrt{nk}+\sqrt{n\log\frac{1}{\epsilon}}\right),
\end{align*}
where the last step uses Lemma~\ref{lem:small-support}.
\end{proof}

\subsection{Approximation from first principles}

We now present our second proof of Theorem~\ref{thm:SYMMETRIC}.
This proof proceeds from first principles, using Chebyshev polynomials
as its only ingredient. To convey the construction as clearly as possible,
we first present an approximant for the simplest and most important
symmetric function, $\AND_{n}$. For this, we adopt the strategy of
previous constructions~\cite{kahn96incl-excl,sherstov07inclexcl-ccc},
whereby one first zeroes out as many of the integer points $n-1,n-2,n-3,\ldots$
as possible and then uses a Chebyshev polynomial to approximate $\AND_{n}$
on the remaining points of $\{0,1,2,\ldots,n\}.$ We depart from the
previous work in the implementation of the first step. Specifically,
we produce the zeroes using a product of Chebyshev polynomials, each
of which is stretched and shifted so as to obtain an extremum at $n$
and a root at one of the points $n-1,n-2,n-3,\ldots.$ The use of
Chebyshev polynomials allows us to avoid explosive growth at the nonzeroes,
thereby eliminating a key source of inefficiency in~\cite{kahn96incl-excl,sherstov07inclexcl-ccc}.
The lemma below shows how to produce a single zero, at any given point
$m$.
\begin{lem}
\label{lem:eliminate-points-close-to-n}Let $n$ and $m$ be given
integers, $0\leq m<n.$ Then there is a univariate polynomial $T_{n,m}$
such that
\begin{align}
 & T_{n,m}(n)=1,\label{eq:Tnm-at-n}\\
 & T_{n,m}(m)=0,\label{eq:Tnm-at-m}\\
 & |T_{n,m}(t)|\leq1, &  & 0\leq t\leq n,\label{eq:Tnm-bounded}\\
 & \deg(T_{n,m})\leq\left\lceil \frac{\pi}{4}\sqrt{\frac{n}{n-m}}\right\rceil .\label{eq:deg-Tnm}
\end{align}
\end{lem}

\begin{proof}
As mentioned above, the construction involves starting with a Chebyshev
polynomial and stretching and shifting it so as to move an extremum
to $n$ and a root to $m.$ In more detail, let
\[
d=\left\lceil \frac{\pi}{4}\sqrt{\frac{n}{n-m}}\right\rceil .
\]
Consider the linear map $L$ that sends
\begin{align}
 & L(n)=1,\label{eq:L-at-n}\\
 & L(m)=\cos\left(\frac{\pi}{2d}\right).\label{eq:L-at-m}
\end{align}
Observe that under $L$, the length of any given interval of the real
line changes by a factor of 
\begin{align*}
\frac{1}{n-m}\left(1-\cos\left(\frac{\pi}{2d}\right)\right) & \leq\frac{1}{n-m}\left(1-\left(1-\frac{\pi^{2}}{8d^{2}}\right)\right)\\
 & =\frac{\pi^{2}}{8d^{2}(n-m)}\\
 & \leq\frac{2}{n},
\end{align*}
where the first step uses $\cos x\geq1-\frac{1}{2}x^{2}$ for $x\in\Re.$
In particular,
\begin{align}
L([0,n]) & \subseteq\left[L(n)-\frac{2}{n}\cdot n,L(n)\right]\nonumber \\
 & \subseteq[-1,1].\label{eq:L-maps-into-01}
\end{align}

We now show that the sought properties~(\ref{eq:Tnm-at-n})\textendash (\ref{eq:deg-Tnm})
hold for the polynomial $T_{n,m}(t)=T_{d}(L(t)),$ where $T_{d}$
denotes as usual the Chebyshev polynomial of degree $d.$ To start
with,
\begin{align*}
T_{n,m}(n) & =T_{d}(L(n))\\
 & =T_{d}(1)\\
 & =1,
\end{align*}
where the last two steps use~(\ref{eq:L-at-n}) and~(\ref{eq:chebyshev-at-1}),
respectively. Similarly,
\begin{align*}
T_{n,m}(m) & =T_{d}(L(m))\\
 & =T_{d}\left(\cos\left(\frac{\pi}{2d}\right)\right)\\
 & =\cos\left(\frac{\pi}{2}\right)\\
 & =0,
\end{align*}
where the second and third steps follow from~(\ref{eq:L-at-m}) and~(\ref{eq:Chebyshev-definition}),
respectively. Continuing, 
\begin{align*}
T_{n,m}([0,n]) & =T_{d}(L([0,n]))\\
 & \subseteq T_{d}([-1,1])\\
 & \subseteq[-1,1],
\end{align*}
where the last two steps follow from~(\ref{eq:L-maps-into-01}) and~(\ref{eqn:chebyshev-containment}),
respectively. Finally, the degree bound~(\ref{eq:deg-Tnm}) is immediate
from the choice of $d.$
\end{proof}
We now obtain the desired approximant for AND and OR, using the two-stage
approach described earlier. The reader interested exclusively in the
general case may wish to skip to Theorem~\ref{lem:orcomplexity-symmetric}.
\begin{thm}
\label{thm:AND}For some constant $c>0$ and all integers $n\geq1$
and $d\geq0,$ there is an $($explicitly given$)$ univariate polynomial
$p$ such that
\begin{align}
 & p(n)=1,\label{eq:p-at-n}\\
 & |p(t)|\leq\exp\left(-\frac{cd^{2}}{n}\right), &  & t=0,1,2,\ldots,n-1,\label{eq:p-at-i}\\
 & |p(t)|\leq1, &  & t\in[0,n],\label{eq:p-bounded}\\
 & \deg p\leq d.\label{eq:deg-p}
\end{align}
In particular,
\begin{align}
 & E(\AND_{n},d)\leq\frac{1}{2}\exp\left(-\frac{c}{2}\cdot\frac{d^{2}}{n}\right), &  & d=0,1,2,3,\ldots,\label{eq:E-ANDn}\\
 & \deg_{\epsilon}(\AND_{n})\leq O\left(\sqrt{n\log\frac{1}{\epsilon}}\right), &  & 0<\epsilon<\frac{1}{2},\label{eq:degeps-ANDn}
\end{align}
and analogously
\begin{align}
 & E(\OR_{n},d)\leq\frac{1}{2}\exp\left(-\frac{c}{2}\cdot\frac{d^{2}}{n}\right), &  & d=0,1,2,3,\ldots\label{eq:E-ORn}\\
 & \deg_{\epsilon}(\OR_{n})\leq O\left(\sqrt{n\log\frac{1}{\epsilon}}\right), &  & 0<\epsilon<\frac{1}{2}.\label{eq:degeps-ORn}
\end{align}
\end{thm}

\begin{proof}
For $d\geq n,$ we may simply take $p(t)=t(t-1)(t-2)\cdots(t-n+1)/n!.$
In what follows, we focus on the construction of $p$ for $d<n.$
Let $\ell,r$ be integer parameters to be chosen later, where $1\leq\ell\leq n-1$
and $1\leq r\leq n.$ We define
\[
p(t)=\frac{T_{r}(t/(n-\ell))}{T_{r}(n/(n-\ell))}\prod_{i=n-\ell+1}^{n-1}T_{n,i}(t),
\]
where $T_{n,i}$ is as constructed in Lemma~\ref{lem:eliminate-points-close-to-n},
and $T_{r}$ stands as usual for the Chebyshev polynomial of degree
$r.$ By~(\ref{eq:Tnm-at-n}) and~(\ref{eq:Tnm-at-m}),
\begin{align}
p(n) & =1,\label{eq:p-at-n-restated}\\
p(t) & =0, &  & t=n-\ell+1,\ldots,n-1.\label{eq:p-at-t}
\end{align}
Moreover,
\begin{align}
\max_{0\leq t\leq n-\ell}|p(t)| & =\max_{0\leq t\leq n-\ell}\;\;\frac{|T_{r}(t/(n-\ell))|}{|T_{r}(n/(n-\ell))|}\prod_{i=n-\ell+1}^{n-1}|T_{n,i}(t)|\nonumber \\
 & \leq\frac{1}{|T_{r}(n/(n-\ell))|}\nonumber \\
 & \leq\frac{1}{\max\left\{ 1+\frac{r^{2}\ell}{n},2^{r\sqrt{\ell/n}-1}\right\} }\nonumber \\
 & \leq\frac{1}{\min\left\{ \exp\left(\frac{r^{2}\ell}{3n}\right),\exp\left(\frac{r\sqrt{\ell}}{3\sqrt{n}}\right)\right\} },\label{eq:p-intermediate}
\end{align}
where the second step uses~(\ref{eqn:chebyshev-containment}) and
(\ref{eq:Tnm-bounded}); the third step follows from Proposition~\ref{prop:chebyshev-beyond-1};
and the last step uses $1+x\geq\exp(x/3)$ for $0\leq x\leq4,$ and
$2^{\sqrt{x}-1}\geq\exp(\sqrt{x}/3)$ for $x\geq4.$ Next,
\begin{align}
\max_{0\leq t\leq n}|p(t)| & =\max_{0\leq t\leq n}\;\;\frac{|T_{r}(t/(n-\ell))|}{|T_{r}(n/(n-\ell))|}\prod_{i=n-\ell+1}^{n-1}|T_{n,i}(t)|\nonumber \\
 & \leq\max_{0\leq t\leq n}\;\;\frac{|T_{r}(t/(n-\ell))|}{|T_{r}(n/(n-\ell))|}\nonumber \\
 & \leq1,\label{eq:p-bounded-on-0-to-n}
\end{align}
where the second inequality uses~(\ref{eq:Tnm-bounded}), and the
third inequality follows from~(\ref{eqn:chebyshev-containment}),
(\ref{eq:chebyshev-at-1}), and Fact~\ref{fact:chebyshev-derivative}.
Finally,
\begin{align}
\deg p & \leq r+\sum_{i=n-\ell+1}^{n-1}\deg(T_{n,i})\nonumber \\
 & \leq r+\sum_{i=1}^{\ell-1}\left(\frac{\pi}{4}\sqrt{\frac{n}{i}}+1\right)\nonumber \\
 & \leq r+\ell-1+\frac{\pi\sqrt{n}}{4}\int_{0}^{\ell-1}\frac{dt}{\sqrt{t}}\allowdisplaybreaks\nonumber \\
 & =r+\ell-1+\frac{\pi\sqrt{n(\ell-1)}}{2}\nonumber \\
 & \leq r+3\sqrt{n(\ell-1)},\label{eq:deg-p-restated}
\end{align}
where the second step uses~(\ref{eq:deg-Tnm}). Now (\ref{eq:p-at-n})\textendash (\ref{eq:deg-p})
follow from (\ref{eq:p-at-n-restated})\textendash (\ref{eq:deg-p-restated})
by setting $r=\lceil d/2\rceil$ and $\ell=\lfloor d^{2}/(36n)\rfloor+1.$

The remaining claims in the theorem statement follow in a straightforward
manner from~(\ref{eq:p-at-n})\textendash (\ref{eq:deg-p}). For~(\ref{eq:E-ANDn}),
we have
\begin{align*}
E(\AND_{n},d) & \leq\max_{x\in\zoon}\left|\AND_{n}(x)-\frac{p(\sum_{i=1}^{n}x_{i})}{1+\exp(-cd^{2}/n)}\right|\\
 & \leq\frac{\exp(-cd^{2}/n)}{1+\exp(-cd^{2}/n)}\\
 & \leq\frac{1}{2}\exp\left(-\frac{c}{2}\cdot\frac{d^{2}}{n}\right),
\end{align*}
where the last step uses $a/(1+a)\leq\sqrt{a}/2$ for any $a\geq0.$
This in turn settles~(\ref{eq:E-ORn}) since $\OR_{n}(x)=1-\AND_{n}(1-x_{1},\ldots,1-x_{n}).$
Finally,~(\ref{eq:degeps-ANDn}) and~(\ref{eq:degeps-ORn}) are
immediate from~(\ref{eq:E-ANDn}) and~(\ref{eq:E-ORn}), respectively.
\end{proof}
To generalize Theorem~\ref{thm:AND} to an arbitrary symmetric function
$f$, it is helpful to think of $f$ as a linear combination of the
characteristic functions of individual levels of the Boolean hypercube.
Specifically, define $\EXACT_{n,k}\colon\zoon\to\zoo$ by
\[
\EXACT_{n,k}(x)=\begin{cases}
1 & \text{if }|x|=k,\\
0 & \text{otherwise.}
\end{cases}
\]
In this notation, Theorem~\ref{thm:AND} treats the special case
$\AND_{n}=\EXACT_{n,n}$. The technique of that theorem is easily
adapted to yield the following more general result.
\begin{thm}
\label{thm:EXACT}For any $0<\epsilon<1/2$ and any integers $n\geq m\geq k\geq0,$
there is a univariate polynomial $p$ such that
\begin{align*}
 & p(|x|)=\EXACT_{n,n-k}(x), &  & |x|\leq m,\\
 & p(|x|)=\EXACT_{n,n-k}(x), &  & |x|\geq n-m,\\
 & |p(|x|)-\EXACT_{n,n-k}(x)|\leq\epsilon, &  & x\in\zoon,\\
 & \deg p=O\left(\sqrt{nm}+\sqrt{n\log\frac{1}{\epsilon}}\right).
\end{align*}
\end{thm}

\begin{proof}
Define
\begin{align}
\ell & =\left\lceil m+\log\frac{2}{\epsilon}\right\rceil ,\label{eq:arb-symm-ell}\\
r & =\left\lceil \sqrt{n\log\frac{2}{\epsilon}}\right\rceil .\label{eq:arb-symm-r}
\end{align}
If $\ell\geq n/2,$ the theorem holds trivially for the degree-$n$
polynomial
\[
p(t)=\prod_{\substack{i=0\\
i\ne n-k
}
}^{n}\frac{t-i}{n-k-i}.
\]
In the complementary case $\ell<n/2,$ define
\begin{multline*}
p(t)=\frac{T_{r}(t/(n-\ell))}{T_{r}((n-k)/(n-\ell))}\cdot\prod_{i=0}^{\ell}T_{n-k,i}(t)\cdot\prod_{i=n-\ell}^{n-k-1}T_{n-k,i}(t)\\
\times\prod_{i=n-k+1}^{n}(1-T_{i,n-k}(t)^{2}),
\end{multline*}
where $T_{n-k,i}$ and $T_{i,n-k}$ are as constructed in Lemma~\ref{lem:eliminate-points-close-to-n},
and $T_{r}$ denotes as usual the Chebyshev polynomial of degree $r.$
Then~(\ref{eq:Tnm-at-n}) and~(\ref{eq:Tnm-at-m}) imply that
\begin{align}
p(t) & =0, &  & t\in\{0,1,\ldots,\ell\},\label{eq:arb-symm-low}\\
p(t) & =0, &  & t\in\{n-\ell,\ldots,n-1,n\}\setminus\{n-k\},\label{eq:arb-symm-high}
\end{align}
and
\begin{align}
p(n-k) & =\prod_{i=0}^{\ell}T_{n-k,i}(n-k)\cdot\prod_{i=n-\ell}^{n-k-1}T_{n-k,i}(n-k)\nonumber \\
 & \qquad\qquad\qquad\qquad\qquad\qquad\qquad\times\prod_{i=n-k+1}^{n}(1-T_{i,n-k}(n-k)^{2})\nonumber \\
 & =\prod_{i=0}^{\ell}1\cdot\prod_{i=n-\ell}^{n-k-1}1\cdot\prod_{i=n-k+1}^{n}(1-0^{2})\nonumber \\
 & =1.\label{eq:arb-symm-at-n-k}
\end{align}
Moreover,
\begin{align}
\max_{0\leq t\leq n-\ell}|p(t)| & =\max_{0\leq t\leq n-\ell}\;\left|\frac{T_{r}(t/(n-\ell))}{T_{r}((n-k)/(n-\ell))}\cdot\prod_{i=0}^{\ell}T_{n-k,i}(t)\right.\nonumber \\
 & \qquad\qquad\qquad\times\left.\prod_{i=n-\ell}^{n-k-1}T_{n-k,i}(t)\cdot\prod_{i=n-k+1}^{n}(1-T_{i,n-k}(t)^{2})\right|\nonumber \\
 & \leq\max_{0\leq t\leq n-\ell}\left|\frac{T_{r}(t/(n-\ell))}{T_{r}((n-k)/(n-\ell))}\right|\nonumber \\
 & \leq\frac{1}{|T_{r}((n-k)/(n-\ell))|}\nonumber \\
 & \leq2^{-r\sqrt{(\ell-k)/n}+1}\nonumber \\
 & \leq\epsilon,\label{eq:arb-symm-intermediate}
\end{align}
where the second step uses~(\ref{eq:Tnm-bounded}), the third step
uses~(\ref{eqn:chebyshev-containment}), the fourth step applies
Proposition~\ref{prop:chebyshev-beyond-1}, and the final step substitutes
the parameters~(\ref{eq:arb-symm-ell}) and~(\ref{eq:arb-symm-r}).
Finally,
\begin{align}
\deg p & \leq r+\sum_{i=0}^{\ell}\deg(T_{n-k,i})+\sum_{i=n-\ell}^{n-k-1}\deg(T_{n-k,i})+2\sum_{i=n-k+1}^{n}\deg(T_{i,n-k})\nonumber \\
 & \leq r+\sum_{i=0}^{\ell}\left(\frac{\pi}{4}\sqrt{\frac{n-k}{n-k-i}}+1\right)+\sum_{i=1}^{\ell-k}\left(\frac{\pi}{4}\sqrt{\frac{n-k}{i}}+1\right)\nonumber \\
 & \qquad\qquad\qquad\qquad+2\sum_{i=1}^{k}\left(\frac{\pi}{4}\sqrt{\frac{n-k+i}{i}}+1\right)\nonumber \\
 & \leq r+3\ell+\pi\sum_{i=1}^{\ell+1}\sqrt{\frac{n}{i}}\allowdisplaybreaks\nonumber \\
 & \leq r+3\ell+\pi\sqrt{n}\int_{0}^{\ell+1}\frac{dt}{\sqrt{t}}\nonumber \\
 & =r+3\ell+2\pi\sqrt{n(\ell+1)}\nonumber \\
 & =O\left(\sqrt{nm}+\sqrt{n\log\frac{1}{\epsilon}}\right),\label{eq:arb-symm-degree}
\end{align}
where the second and third steps use~(\ref{eq:Tnm-bounded}) and~$k<\ell<n-k$,
respectively. In view of~(\ref{eq:arb-symm-low})\textendash (\ref{eq:arb-symm-degree}),
the proof is complete.
\end{proof}
We are now in a position to handle arbitrary symmetric functions by
expressing them as a linear combination of $\EXACT_{n,i}$ for $i=0,1,2,\ldots,n.$
This result provides a new proof of Theorem~\ref{thm:SYMMETRIC}.
\begin{thm}
\label{thm:SYMMETRIC-proof-first-principles}Let $f\colon\zoon\to[-1,1]$
be an arbitrary symmetric function. Let $k$ be a nonnegative integer
such that $f$ is constant on inputs of Hamming weight in $(k,n-k).$
Then for $0<\epsilon<1/2,$ 
\begin{equation}
\deg_{\epsilon}(f)=O\left(\sqrt{nk}+\sqrt{n\log\frac{1}{\epsilon}}\right).\label{eq:degeps-f-symmetric-first-principles}
\end{equation}
More precisely, there is an $($explicitly given$)$ polynomial $\tilde{f}\colon\zoon\to\Re$
such that
\begin{align}
 & f(x)=\tilde{f}(x), &  & |x|\leq k,\label{eq:f-tilde-f-equal-on-small-Hamm-weight-first-principles}\\
 & f(x)=\tilde{f}(x), &  & |x|\geq n-k,\label{eq:f-tilde-f-equal-on-large-Hamm-weight-first-principles}\\
 & |f(x)-\tilde{f}(x)|\leq\epsilon, &  & x\in\zoon,\label{eq:f-tilde-f-close-on-med-Hamm-weight-first-principles}\\
 & \deg\tilde{f}=O\left(\sqrt{nk}+\sqrt{n\log\frac{1}{\epsilon}}\right).\label{eq:tilde-f-degree-symmetric-first-principles}
\end{align}
\end{thm}

\begin{proof}
If $k\geq n/2,$ the theorem follows from the trivial bound $\deg_{0}(f)\leq n.$
For the complementary case $k<n/2,$ write
\begin{align*}
f(x) & =\lambda+\sum_{i=0}^{k}\lambda'_{i}\cdot\EXACT_{n,i}(x)+\sum_{i=0}^{k}\lambda''_{i}\cdot\EXACT_{n,n-i}(x)\\
 & =\lambda+\sum_{i=0}^{k}\lambda'_{i}\cdot\EXACT_{n,n-i}(\overline{x_{1}},\ldots,\overline{x_{n}})+\sum_{i=0}^{k}\lambda''_{i}\cdot\EXACT_{n,n-i}(x),
\end{align*}
where $\lambda,\lambda_{0}',\lambda_{0}'',\ldots,\lambda_{k}',\lambda_{k}''\in[-2,2]$
are fixed reals. By Theorem~\ref{thm:EXACT}, each of the functions
$\EXACT_{n,n-i}$ in this linear combination can be approximated pointwise
to within $\epsilon/(2k+2)$ by a polynomial of degree $O(\sqrt{nk}+\sqrt{n\log(1/\epsilon)}).$
Moreover, the lemma guarantees that in each case, the approximation
is exact on $\zoon_{\leq k}$ and $\zoon_{\geq n-k}.$ Now (\ref{eq:degeps-f-symmetric-first-principles})\textendash (\ref{eq:f-tilde-f-close-on-med-Hamm-weight-first-principles})
are immediate.
\end{proof}

\subsection{Approximation using a sampling argument}

We now give a third proof of Theorem~\ref{thm:SYMMETRIC}, inspired
by combinatorics rather than approximation theory. Here, we show how
to approximate an arbitrary symmetric function $f$ using an approximant
for AND (cf.~Theorem~\ref{thm:AND}) and a sampling argument. Suppose
for the sake of concreteness that $f$ is supported on inputs of Hamming
weight at most $k.$ Given a string $x\in\zoon,$ consider the  experiment
whereby one chooses $\lfloor n/k\rfloor$ bits of $x$ independently
and uniformly at random, and outputs the disjunction of those bits.
To approximate $f,$ we feed the expected value of the sampling experiment
to a suitable univariate polynomial constructed by Lagrange interpolation.
The expected value of the experiment as a function of $x$ has $\orcomplexity$-norm
at most $2,$ which by Proposition~\ref{prop:orcomplexity} means
that the overall composition has small $\Pi$-norm as well. The complete
details of this construction are provided in Lemma~\ref{lem:orcomplexity-symmetric}.
To finish the proof, we expand the composition as a linear combination
of conjunctions and replace each conjunction by a corresponding approximant
from Theorem~\ref{thm:AND}.
\begin{lem}
\label{lem:orcomplexity-symmetric}Let $k\geq0$ be a given integer.
Let $f\colon\zoon\to[-1,1]$ be a symmetric function that vanishes
on $\zoon_{>k}.$ Then for every $0<\epsilon<1/2,$ there exists an
$($explicitly given$)$ function $\tilde{f}\colon\zoon\to\Re$ such
that
\begin{align}
 & f(x)=\tilde{f}(x), &  & |x|\leq k,\label{eq:f-tilde-f-equal-on-small-Hamm-weight}\\
 & f(x)=\tilde{f}(x), &  & |x|\geq n-k,\label{eq:f-tilde-f-equal-on-large-Hamm-weight}\\
 & |f(x)-\tilde{f}(x)|\leq\epsilon, &  & x\in\zoon,\label{eq:f-tilde-f-close-on-med-Hamm-weight}\\
 & \orcomplexity(\tilde{f})\leq C^{k+\log(1/\epsilon)},\label{eq:tilde-f-orcomplexity}
\end{align}
where $C>1$ is an absolute constant independent of $f,n,k,\epsilon.$
\end{lem}

\begin{proof}
If $k=0,$ the only possibilities are $f(x)\equiv0$ and $f(x)=\bigwedge\overline{x_{i}},$
and therefore we may take $\tilde{f}=f$. If $k\geq n/4,$ we again
may take $\tilde{f}=f$ since $\orcomplexity(f)\leq2^{n}$ by Proposition~\ref{prop:orcomplexity}\ref{item:orcomplexity-trivial-bound}.
In what follows, we treat the remaining case 
\begin{equation}
1\leq k<\frac{n}{4}.\label{eq:k-not-too-large}
\end{equation}

Consider the points $0=t_{0}\leq t_{1}\leq t_{2}\leq\cdots\leq t_{n}=1,$
where
\begin{align*}
t_{i} & =1-\left(1-\frac{i}{n}\right)^{\left\lfloor \frac{n}{2k}\right\rfloor },\qquad\qquad i=0,1,2,\ldots,n.
\end{align*}
The derivative of $t\mapsto1-(1-\frac{t}{n})^{\lfloor n/(2k)\rfloor}$
on $[0,2k]$ ranges in $[\frac{1}{6k},\frac{1}{2k}].$ Therefore,
the mean value theorem gives
\begin{align}
\frac{|i-j|}{6k}\leq|t_{i}-t_{j}| & \leq\frac{|i-j|}{2k}, &  & i,j=0,1,2,\ldots,2k.\label{eq:ti-tj}
\end{align}
In particular, 
\begin{align}
\frac{i}{6k} & \leq t_{i}\leq\frac{i}{2k}, &  & i=0,1,2,\ldots,2k.\label{eq:ti}
\end{align}

Consider the univariate polynomials
\begin{align*}
p(t) & =(1-t)^{d}\prod_{i=n-k}^{n}(t-t_{i}),\\
q(t) & =\sum_{i=0}^{k}\frac{f(1^{i}0^{n-i})}{p(t_{i})}\prod_{\substack{j=0\\
j\ne i
}
}^{2k}\frac{t-t_{j}}{t_{i}-t_{j}},
\end{align*}
where
\begin{equation}
d=5\left\lceil 8k+\ln\frac{1}{\epsilon}\right\rceil .\label{eq:d-defined}
\end{equation}
Our definitions ensure that $p(t_{i})q(t_{i})=f(1^{i}0^{n-i})$ for
$i=0,1,2,\ldots,k.$ Moreover, we have $p(t_{i})q(t_{i})=0$ for $i=\{k+1,k+2,\ldots,2k\}\cup\{n-k,n-k+1,\ldots,n\}.$
Since $f$ vanishes on inputs of Hamming weight greater than $k,$
we conclude that
\begin{multline}
p(t_{i})q(t_{i})=f(1^{i}0^{n-i}),\\
i=\{0,1,\ldots,2k\}\cup\{n-k,n-k+1,\ldots,n\}.\qquad\label{eq:pq-equals-f-on-small-and-large-inputs}
\end{multline}
A routine calculation reveals the following additional properties
of $p$ and $q$.
\begin{claim}
\label{claim:pq-close-to-f}$|p(t_{i})q(t_{i})-f(1^{i}0^{n-i})|\leq\epsilon$
for $i\geq2k.$
\end{claim}

\begin{claim}
\label{claim:pq-norm}$\norm{p\cdot q}=2^{O(k+\log(1/\epsilon))}.$
\end{claim}

We will settle these claims once we complete the main proof. Define
$\tilde{f}\colon\zoon\to\Re$ by $\tilde{f}(x)=p(t_{|x|})q(t_{|x|}).$
Then~(\ref{eq:f-tilde-f-equal-on-small-Hamm-weight})\textendash (\ref{eq:f-tilde-f-close-on-med-Hamm-weight})
follow directly from~(\ref{eq:pq-equals-f-on-small-and-large-inputs})
and Claim~\ref{claim:pq-close-to-f}. For~(\ref{eq:tilde-f-orcomplexity}),
observe that 
\[
\tilde{f}(x)=p\left(\Exp_{S}\;\,\bigvee_{i\in S}x_{i}\right)q\left(\Exp_{S}\;\,\bigvee_{i\in S}x_{i}\right)
\]
where the expectation is over a multiset $S$ of $\lfloor\frac{n}{2k}\rfloor$
elements that are chosen independently and uniformly at random from
$\{1,2,\ldots,n\}.$ As a result,~(\ref{eq:tilde-f-orcomplexity})
follows from Claim~\ref{claim:pq-norm} and Proposition~\ref{prop:orcomplexity}~\ref{item:orcomplexity-homogeneous},~\ref{item:orcomplexity-triangle},~\ref{item:orcomplexity-disjunction},~\ref{item:orcomplexity-composition-with-univariate}.
\end{proof}
\begin{proof}[Proof of Claim~\emph{\ref{claim:pq-close-to-f}.}]
 Fix an arbitrary point $t\in[$$t_{2k},t_{n}]=[t_{2k},1]$. Recall
from~(\ref{eq:k-not-too-large}) that $k<n/4.$ As a result,
\begin{align*}
|p(t)q(t)| & \leq|p(t_{2k})q(t)|\\
 & \leq|p(t_{2k})|\sum_{i=0}^{k}\frac{1}{\min\{|p(t_{0})|,\ldots,|p(t_{k})|\}}\prod_{\substack{j=0\\
j\ne i
}
}^{2k}\frac{|1-t_{j}|}{|t_{i}-t_{j}|}\\
 & \leq\frac{|p(t_{2k})|}{|p(t_{k})|}\sum_{i=0}^{k}\prod_{\substack{j=0\\
j\ne i
}
}^{2k}\frac{|1-t_{j}|}{|t_{i}-t_{j}|}\\
 & \leq\left(\frac{1-t_{2k}}{1-t_{k}}\right)^{d}\;\sum_{i=0}^{k}\prod_{\substack{j=0\\
j\ne i
}
}^{2k}\frac{|1-t_{j}|}{|t_{i}-t_{j}|}\\
 & =\left(1-\frac{t_{2k}-t_{k}}{1-t_{k}}\right)^{d}\;\sum_{i=0}^{k}\prod_{\substack{j=0\\
j\ne i
}
}^{2k}\frac{|1-t_{j}|}{|t_{i}-t_{j}|}\\
 & \leq\exp\left(-\frac{t_{2k}-t_{k}}{1-t_{k}}\cdot d\right)\sum_{i=0}^{k}\prod_{\substack{j=0\\
j\ne i
}
}^{2k}\frac{|1-t_{j}|}{|t_{i}-t_{j}|}.
\end{align*}
Using the lower bounds in~(\ref{eq:ti-tj}) and~(\ref{eq:ti}),
we obtain
\begin{align*}
|p(t)q(t)| & \leq\exp\left(-\frac{(2k-k)/6k}{1-(k/6k)}\cdot d\right)\sum_{i=0}^{k}\prod_{\substack{j=0\\
j\ne i
}
}^{2k}\frac{1-(j/6k)}{|i-j|/6k}\\
 & =\exp\left(-\frac{d}{5}\right)\sum_{i=0}^{k}\prod_{\substack{j=0\\
j\ne i
}
}^{2k}\frac{6k-j}{|i-j|}\\
 & \leq\exp\left(-\frac{d}{5}\right)\sum_{i=0}^{k}\frac{(6k)!/(4k)!}{i!\,(2k-i)!}\\
 & =\exp\left(-\frac{d}{5}\right)\sum_{i=0}^{k}\binom{6k}{4k}\binom{2k}{i}\\
 & \leq\exp\left(-\frac{d}{5}\right)\binom{6k}{4k}\cdot2^{2k}\\
 & \leq\exp\left(-\frac{d}{5}\right)\cdot2^{8k}\\
 & \leq\epsilon,
\end{align*}
where the last step follows from the definition of $d$ in~(\ref{eq:d-defined}).
Hence, $|p(t_{i})q(t_{i})-f(1^{i}0^{n-i})|=|p(t_{i})q(t_{i})|\leq\epsilon$
for $i\geq2k.$
\end{proof}
\begin{proof}[Proof of Claim~\emph{\ref{claim:pq-norm}.}]
 Recall from~(\ref{eq:k-not-too-large}) that $k<n/4.$ As a result,
\begin{align}
\min_{i=0,1,\ldots,k}|p(t_{i})| & =|p(t_{k})|\nonumber \\
 & =|1-t_{k}|^{d}\prod_{i=n-k}^{n}|t_{k}-t_{i}|\nonumber \\
 & \geq|1-t_{k}|^{d}\cdot|t_{k}-t_{2k}|^{k+1}\nonumber \\
 & \geq\frac{1}{2^{d}\cdot6^{k+1}},\label{eq:min-pt-i}
\end{align}
where the last step uses the estimates in~(\ref{eq:ti-tj}) and~(\ref{eq:ti}).
As a result,
\begin{align*}
\norm{p\cdot q} & \leq(1+1)^{d}\prod_{i=n-k}^{n}(1+t_{i})\cdot\sum_{i=0}^{k}\frac{|f(1^{i}0^{n-i})|}{|p(t_{i})|}\prod_{\substack{j=0\\
j\ne i
}
}^{2k}\frac{1+t_{j}}{|t_{i}-t_{j}|}\\
 & \leq2^{d}\cdot2^{k+1}\sum_{i=0}^{k}\frac{1}{2^{-d}\cdot6^{-k-1}}\prod_{\substack{j=0\\
j\ne i
}
}^{2k}\frac{2}{|t_{i}-t_{j}|}\\
 & \leq4^{d}\cdot12^{k+1}\sum_{i=0}^{k}\prod_{\substack{j=0\\
j\ne i
}
}^{2k}\frac{2\cdot6k}{|i-j|}\allowdisplaybreaks\\
 & =4^{d}\cdot12^{k+1}\cdot6^{2k}\cdot\frac{(2k)^{2k}}{(2k)!}\sum_{i=0}^{k}\binom{2k}{i}\\
 & \leq4^{d}\cdot12^{k+1}\cdot6^{2k}\cdot\frac{(2k)^{2k}}{(2k)!}\cdot2^{2k}\\
 & =2^{O(d+k)},
\end{align*}
where the first step is valid by Fact~\ref{fact:poly-norm}; the
second step uses $0\leq t_{i}\leq1$ and~(\ref{eq:min-pt-i}); the
third step follows from the lower bound in~(\ref{eq:ti-tj}); and
the last step is legitimate by Stirling's approximation. In view of~(\ref{eq:d-defined}),
the proof is complete. 
\end{proof}
We have reached the promised construction of an approximating polynomial
for any symmetric function.
\begin{thm*}[restatement of Theorem~\ref{thm:SYMMETRIC}]
Let $f\colon\zoon\to[-1,1]$ be an arbitrary symmetric function.
Let $k$ be a nonnegative integer such that $f$ is constant on inputs
of Hamming weight in $(k,n-k).$ Then for $0<\epsilon<1/2,$
\begin{equation}
\deg_{\epsilon}(f)=O\left(\sqrt{nk}+\sqrt{n\log\frac{1}{\epsilon}}\right).\label{eq:degeps-f-symmetric}
\end{equation}
Moreover, the approximating polynomial is given explicitly in each
case.
\end{thm*}
\begin{proof}
If $k\geq n/2,$ the theorem follows from the trivial bound $\deg_{0}(f)\leq n.$
For the complementary case $k<n/2$, write 
\[
f(x_{1},\ldots,x_{n})=\lambda+f'(x_{1},\ldots,x_{n})+f''(\overline{x_{1}},\ldots,\overline{x_{n}}),
\]
where $\lambda\in[-1,1]$ and $f',f''\colon\zoon\to[-2,2]$ are symmetric
functions that vanish on $\zoon_{>k}.$ Lemma~\ref{lem:orcomplexity-symmetric}
shows that $f'/2$ and $f''/2$ are each approximated pointwise to
within $\epsilon/5$ by a linear combination of conjunctions, with
real coefficients whose absolute values sum to $2^{O(k+\log(1/\epsilon))}.$
By Theorem~\ref{thm:AND}, each such conjunction can in turn be approximated
pointwise by a polynomial of degree $d$ to within $2^{-\Theta(d^{2}/n)}$.
Summarizing,
\begin{align*}
E(f,d) & \leq E(f',d)+E(f'',d)\\
 & \leq2E\left(\frac{f'}{2},d\right)+2E\left(\frac{f''}{2},d\right)\\
 & \leq2\left(2\cdot\frac{\epsilon}{5}+2^{O(k+\log(1/\epsilon))}\cdot2^{-\Theta(d^{2}/n)}\right),
\end{align*}
whence~(\ref{eq:degeps-f-symmetric}). Moreover, the approximating
polynomial is given explicitly because Theorem~\ref{thm:AND} and
Lemma~\ref{lem:orcomplexity-symmetric} provide closed-form expressions
for the approximants involved.
\end{proof}

\subsection{Generalizations}

Theorem~\ref{thm:AND} on the approximation of AND and OR obviously
generalizes to arbitrary conjunctions and disjunctions. Somewhat less
obviously, it generalizes in an optimal manner to conjunctions and
disjunctions whose domain of definition is restricted to the first
few levels of the hypercube. We record this generalization for later
use.
\begin{thm}
\label{thm:DISJUNCTION}Let $f\colon\zoo_{\leq n}^{N}\to\zoo$ be
given by
\[
f(x)=\left(\bigvee_{i\in A}x_{i}\right)\vee\left(\bigvee_{i\in B}\overline{x_{i}}\right),
\]
for some subsets $A,B\subseteq\{1,2,\ldots,N\}.$ Then
\begin{align*}
E(f,d) & \leq\frac{1}{2}\exp\left(-\frac{cd^{2}}{n}\right), &  & d=0,1,2,\ldots,
\end{align*}
where $c>0$ is an absolute constant. Moreover, the approximating
polynomial is given explicitly in each case.
\end{thm}

\begin{proof}
If $|B|>n,$ then $f\equiv1$ on its domain of definition and hence
$E(f,0)=0$.

In the complementary case when $|B|\leq n,$ we have 
\begin{align}
\sum_{i\in A}x_{i}+\sum_{i\in B}(1-x_{i}) & \in\{0,1,2,\ldots,2n\}, &  & x\in\zoo_{\leq n}^{N}.\label{eq:sum-contained}
\end{align}
Theorem~\ref{thm:AND} gives an explicit univariate polynomial $p$
of degree $d$ such that
\begin{align}
 & p(2n)=1,\label{eq:p-at-0-restated}\\
 & |p(t)|\leq\exp\left(-\frac{Cd^{2}}{n}\right), &  & t=0,1,2,\ldots,2n-1,\label{eq:p-at-i-restated}
\end{align}
where $C>0$ is an absolute constant. Define 
\[
P(x)=1-\frac{1}{1+\exp(-Cd^{2}/n)}\cdot p\left(2n-\sum_{i\in A}x_{i}-\sum_{i\in B}(1-x_{i})\right).
\]
Then
\begin{align*}
\max_{x\in\zoo_{\leq n}^{N}}|f(x)-P(x)| & \leq\frac{\exp(-Cd^{2}/n)}{1+\exp(-Cd^{2}/n)}\\
 & \leq\frac{1}{2}\exp\left(-\frac{Cd^{2}}{2n}\right),
\end{align*}
where the first step follows from~(\ref{eq:sum-contained})\textendash (\ref{eq:p-at-i-restated}),
and the second step uses $a/(1+a)\leq\sqrt{a}/2$ for any $a\geq0.$
\end{proof}
\begin{cor}
\label{cor:CONJUNCTION}Let $f\colon\zoo_{\leq n}^{N}\to\zoo$ be
given by
\[
f(x)=\left(\bigwedge_{i\in A}x_{i}\right)\wedge\left(\bigwedge_{i\in B}\overline{x_{i}}\right),
\]
for some subsets $A,B\subseteq\{1,2,\ldots,N\}.$ Then
\begin{align*}
E(f,d) & \leq\frac{1}{2}\exp\left(-\frac{cd^{2}}{n}\right), &  & d=0,1,2,\ldots,
\end{align*}
where $c>0$ is an absolute constant. Moreover, the approximating
polynomial is given explicitly in each case.
\end{cor}

\begin{proof}
Apply Theorem~\ref{thm:DISJUNCTION} to $1-f.$
\end{proof}

\section{\emph{\label{sec:k-DNF-and-k-CNF}k}-DNF and \emph{k}-CNF formulas}

Recall that a \emph{$k$-DNF formula} in Boolean variables $x_{1},x_{2},\ldots,x_{N}$
is the disjunction of zero or more \emph{terms}, where each term is
the conjunction of at most $k$ literals from among $x_{1},\overline{x_{1}},x_{2},\overline{x_{2}},\ldots,x_{N},\overline{x_{N}}.$
As a convention, we consider the constant functions $0$ and $1$
to be valid $k$-DNF formulas for every $k\geq0$. Analogously, a
\emph{$k$-CNF formula} in Boolean variables $x_{1},x_{2},\ldots,x_{N}$
is the conjunction of zero or more \emph{clauses}, where each clause
is the disjunction of at most $k$ literals from among $x_{1},\overline{x_{1}},x_{2},\overline{x_{2}},\ldots,x_{N},\overline{x_{N}}.$
Again, we consider the constant functions $0$ and $1$ to be valid
$k$-CNF formulas for all $k\ge0.$ Recall that a function $f$ is
representable by a $k$-DNF formula if and only if its negation $\overline{f}$
is representable by a $k$-CNF formula. Note also that the definition
of $k$-DNF formulas is hereditary in the sense that a $k$-DNF formula
is also a $k'$-DNF formula for any $k'\geq k,$ and analogously for
CNF formulas. 

The contribution of this section is to settle Theorem~\ref{thm:MAIN-k-dnf}
on the approximate degree of every $k$-DNF and $k$-CNF formula.
We will in fact prove the following more precise result, for every
setting of the error parameter.
\begin{thm}
\label{thm:k-DNF}Let $f\colon\zoo_{\leq n}^{N}\to\zoo$ be representable
on its domain by a $k$-DNF or $k$-CNF formula. Then
\begin{equation}
\deg_{\epsilon}(f)\leq c\cdot(\sqrt{2})^{k}\,n^{\frac{k}{k+1}}\left(\log\frac{1}{\epsilon}\right)^{\frac{1}{k+1}}\label{eq:k-dnf-main}
\end{equation}
for all $0<\epsilon<1/2,$ where $c>1$ is an absolute constant independent
of $f,N,n,k,\epsilon.$ Moreover, the approximating polynomial is
given explicitly in each case.
\end{thm}

\noindent We present the proof of this theorem in Sections~\ref{subsec:k-dnf-key-quantities}\textendash \ref{subsec:k-dnf-solving-the-recurrence}
below.

\subsection{Key quantities\label{subsec:k-dnf-key-quantities}}

For nonnegative integers $n$ and $k$ and a real number $\Delta\geq1,$
we define 
\[
D(n,k,\Delta)=\max_{f}\deg_{2^{-\Delta}}(f),
\]
where the maximum is over all functions $f\colon\zoo_{\leq n}^{N}\to\zoo$
for some $N\geq n$ that are representable by a $k$-DNF formula.
Fact~\ref{fact:trivial-upper-bound-on-deg} gives the upper bound
\begin{equation}
D(n,k,\Delta)\leq n.\label{eq:D-nkr-trivial}
\end{equation}
Since the only $0$-DNF formulas are the constant functions $0$ and
$1,$ we obtain
\begin{equation}
D(n,0,\Delta)=0.\label{eq:D-nkr-base}
\end{equation}
We will prove Theorem~\ref{thm:k-DNF} by induction of $k$, with
(\ref{eq:D-nkr-base}) serving as the base case.

\subsection{A composition theorem for approximate degree}

The inductive step in our analysis of $D(n,k,\Delta)$ relies on a
certain general bound on approximate degree for a class of composed
functions, as follows.
\begin{lem}
\label{lem:selector-homogeneous}Let $F\colon X\times\zoo_{n}^{N}\to\zoo$
be given by
\[
F(x,y)=\bigvee_{i=1}^{N}y_{i}\wedge f_{i}(x)
\]
for some functions $f_{1},f_{2},\ldots,f_{N}\colon X\to\zoo.$ Let
$b$ be an integer with $b\mid n$ and $b\mid N.$ Then 
\[
\deg_{\epsilon}(F)\leq C\sqrt{nb\log\frac{1}{\epsilon}}+\max_{\substack{S\subseteq\{1,\ldots,N\}\\
|S|\leq C\sqrt{nb\log\frac{1}{\epsilon}}
}
}\deg_{\epsilon\exp\left(-C\sqrt{\frac{n}{b}\log\frac{1}{\epsilon}}\right)}\left(\bigvee_{i\in S}f_{i}\right)
\]
for all $0<\epsilon\leq1/2,$ where $C>1$ is an absolute constant
independent of $F,N,n,b,\epsilon.$
\end{lem}

\noindent As we will see shortly, the bound of Lemma~\ref{lem:selector-homogeneous}
generalizes to functions $F\colon X\times\zoo_{\leq n}^{N}\to\zoo$
and to arbitrary reals $b\geq1$. It is this more general, and more
natural, result on the approximate degree of composed functions that
we need for our analysis of $D(n,k,\Delta).$ However, establishing
Lemma~\ref{lem:selector-homogeneous} first considerably improves
the readability and modularity of the proof. By way of notation, we
remind the reader that the symbol $\bigvee_{i\in S}f_{i}$ denotes
the mapping $x\mapsto\bigvee_{i\in S}f_{i}(x).$ The reader will also
recall the shorthand $[n]=\{1,2,\ldots,n\}.$ In particular, $\binom{[n]}{\leq d}$
denotes the family of subsets of $\{1,2,\ldots,n\}$ of cardinality
at most $d.$
\begin{proof}[Proof of Lemma~\emph{\ref{lem:selector-homogeneous}.}]
The proof is constructive and uses as its building blocks two main
components: an ``outer'' approximant (for the OR function) and ``inner''
approximants (for disjunctions of small sets of $f_{i}$). We first
describe these components individually and then present the overall
construction and error analysis.\bigskip{}

\textsc{Step 1: Outer approximant.} Theorem~\ref{thm:AND} provides
a symmetric multilinear polynomial $\widetilde{\OR}_{n/b}\colon\zoo^{n/b}\to[0,1]$
of degree $d=O(\sqrt{n\log(1/\epsilon)/b})$ that approximates $\OR_{n/b}$
pointwise to within $\epsilon/2.$ More specifically, there are real
coefficients $a_{0},a_{1},a_{2},\ldots$ such that
\begin{align}
\left|\bigvee_{i=1}^{n/b}z_{i}-\sum_{S\in\binom{[n/b]}{\leq d}}a_{|S|}\prod_{i\in S}z_{i}\right| & \leq\frac{\epsilon}{2}, &  & z\in\zoo^{n/b},\label{eq:decomp-tilde-OR-error}
\end{align}
where
\begin{align}
 & 1\leq d\leq c\sqrt{\frac{n}{b}\log\frac{1}{\epsilon}}\label{eq:decomp-tilde-OR-degree}
\end{align}
for some absolute constant $c>1$. By Lemma~\ref{lem:bound-on-poly-coeffs-MULTIVARIATE},
\begin{align}
\sum_{\ell=0}^{d}\binom{n/b}{\ell}|a_{\ell}| & \leq8^{d}.\label{eq:tilde-OR-norm}
\end{align}
\bigskip{}

\textsc{Step 2: Inner approximants.}\emph{ }For a subset $S\subseteq\{1,2,\ldots,N\},$
define $f_{S}\colon X\to\zoo$ by
\[
f_{S}(x)=\bigvee_{i\in S}f_{i}(x).
\]
Fix a polynomial $\tilde{f}_{S}\colon X\to\Re$ of the smallest possible
degree such that 
\begin{align}
 & \|f_{S}-\tilde{f}_{S}\|_{\infty}\leq\frac{\epsilon}{2}\left(\sum_{\ell=0}^{d}\binom{n/b}{\ell}2^{\ell}|a_{\ell}|\right)^{-1}.\label{eq:decomp-approx-fS-error}
\end{align}
To avoid notational clutter in the formulas below, we will frequently
write $f_{S}$ and $\tilde{f}_{S}$ instead of $f_{S}(x)$ and $\tilde{f}_{S}(x)$,
respectively, when referring to the value of these functions at a
given point $x\in X.$ We have
\begin{align}
\deg(\tilde{f}_{S}) & \leq\deg_{\epsilon/(2\cdot16^{d})}\left(\bigvee_{i\in S}f_{i}\right)\nonumber \\
 & \leq\deg_{\epsilon\exp\left(-4c\sqrt{\frac{n}{b}\log\frac{1}{\epsilon}}\right)}\left(\bigvee_{i\in S}f_{i}\right),\label{eq:tilde-fs-degree}
\end{align}
where the first and second steps use~(\ref{eq:tilde-OR-norm}) and~(\ref{eq:decomp-tilde-OR-degree}),
respectively.\bigskip{}

\textsc{Step 3: Overall approximant. }By appropriately composing the
outer approximant with the inner approximants, we obtain an approximant
for the overall function $F$. Specifically, define $\tilde{F}\colon X\times\zoo_{n}^{N}\to\Re$
by
\begin{multline}
\tilde{F}(x,y)=a_{0}+\sum_{\ell=1}^{d}a_{\ell}\binom{n/b}{\ell}\binom{N}{b\ell}\binom{n}{b\ell}^{-1}\\
\times\Exp_{B_{1},\ldots,B_{n/b}}\left[\left(\sum_{\substack{S\subseteq\{1,2,\ldots,\ell\}\\
S\ne\varnothing
}
}(-1)^{|S|+1}\;\widetilde{f}_{\bigcup_{i\in S}\!B_{i}}\right)\prod_{i\in B_{1}\cup\cdots\cup B_{\ell}}y_{i}\right],\label{eq:decomp-approximant-defined}
\end{multline}
where expectation is taken over a uniformly random tuple of sets $B_{1},\ldots,B_{n/b}\subseteq\{1,2,\ldots,N\}$
that are pairwise disjoint and have cardinality $b$ each. Then
\begin{align}
\deg(\tilde{F}) & \leq\max_{B_{1},\ldots,B_{n/b}}\;\max_{S\subseteq\{1,2,\ldots,d\}}\left\{ \deg(\tilde{f}_{\bigcup_{i\in S}B_{i}})+\sum_{i=1}^{d}|B_{i}|\right\} \nonumber \\
 & \leq\max_{S\in\binom{[N]}{\leq db}}\left\{ \deg(\tilde{f}_{S})\right\} +db\nonumber \\
 & \leq\max_{\substack{S\subseteq\{1,\ldots,N\}\\
|S|\leq c\sqrt{nb\log\frac{1}{\epsilon}}
}
}\left\{ \deg_{\epsilon\exp\left(-4c\sqrt{\frac{n}{b}\log\frac{1}{\epsilon}}\right)}\left(\bigvee_{i\in S}f_{i}\right)\right\} +c\sqrt{nb\log\frac{1}{\epsilon}},\label{eq:tilde-F-degree}
\end{align}
where the final step uses~(\ref{eq:decomp-tilde-OR-degree}) and~(\ref{eq:tilde-fs-degree}).
\bigskip{}

\textsc{Step 4: Error analysis.} For the rest of the proof, fix $y\in\zoo_{n}^{N}$
arbitrarily. Let $L=\{i:y_{i}=1\}$. In the defining equation~(\ref{eq:decomp-approximant-defined}),
the product $\prod_{i\in B_{1}\cup\cdots\cup B_{\ell}}y_{i}$ acts
like an indicator random variable for the event that $B_{1}\cup\ldots\cup B_{\ell}\subseteq L,$
which occurs with probability precisely 
\[
\binom{|L|}{b\ell}\binom{N}{b\ell}^{-1}=\binom{n}{b\ell}\binom{N}{b\ell}^{-1}.
\]
 Therefore,
\begin{align}
 & \hspace{-5mm}\tilde{F}(x,y)\nonumber \\
 & =a_{0}+\sum_{\ell=1}^{d}a_{\ell}\binom{n/b}{\ell}\Exp_{B_{1},\ldots,B_{n/b}}\left[\sum_{\substack{S\subseteq[\ell]\\
S\ne\varnothing
}
}(-1)^{|S|+1}\;\widetilde{f}_{\bigcup_{i\in S}\!B_{i}}\;\middle|\;B_{1},\ldots,B_{\ell}\subseteq L\right]\nonumber \\
 & =a_{0}+\sum_{\ell=1}^{d}a_{\ell}\binom{n/b}{\ell}\Exp_{B_{1},\ldots,B_{n/b}}\left[\sum_{\substack{S\subseteq[\ell]\\
S\ne\varnothing
}
}(-1)^{|S|+1}\;\widetilde{f}_{\bigcup_{i\in S}\!B_{i}}\;\middle|\;\bigcup_{i=1}^{n/b}B_{i}=L\right]\allowdisplaybreaks\nonumber \\
 & =a_{0}+\sum_{\ell=1}^{d}a_{\ell}\Exp_{B_{1},\ldots,B_{n/b}}\left[\sum_{T\in\binom{[n/b]}{\ell}}\sum_{\substack{S\subseteq T\\
S\ne\varnothing
}
}(-1)^{|S|+1}\;\widetilde{f}_{\bigcup_{i\in S}\!B_{i}}\;\;\middle|\;\;\bigcup_{i=1}^{n/b}B_{i}=L\right]\nonumber \\
 & =\Exp_{B_{1},\ldots,B_{n/b}}\left[a_{0}+\sum_{\ell=1}^{d}a_{\ell}\sum_{T\in\binom{[n/b]}{\ell}}\sum_{\substack{S\subseteq T\\
S\ne\varnothing
}
}(-1)^{|S|+1}\;\widetilde{f}_{\bigcup_{i\in S}\!B_{i}}\;\;\middle|\;\;\bigcup_{i=1}^{n/b}B_{i}=L\right],\label{eq:decomp-tilde-Fxy-simplified}
\end{align}
where the second step is valid because a uniformly random tuple of
pairwise disjoint sets $B_{1},\ldots,B_{\ell}\subseteq L$ of cardinality
$b$ each can be generated by partitioning $L$ uniformly at random
into parts of size $b$ and using the first $\ell$ parts of that
partition; the third step is valid in view of the symmetry of the
distribution of $B_{1},\ldots,B_{n/b}$; and the last step uses the
linearity of expectation. Analogously,
\begin{align}
F(x,y) & =\bigvee_{i=1}^{N}y_{i}\wedge f_{i}\nonumber \\
 & =\bigvee_{i\in L}f_{i}\nonumber \\
 & =f_{L}\nonumber \\
 & =\Exp_{B_{1},\ldots,B_{n/b}}\left[f_{B_{1}\cup\cdots\cup B_{n/b}}\;\;\middle|\;\;\bigcup_{i=1}^{n/b}B_{i}=L\right].\label{eq:decomp-Fxy-simplified}
\end{align}
As a result,
\begin{align}
 & |F(x,y)-\tilde{F}(x,y)|\nonumber \\
 & \qquad\leq\max_{B_{1},\ldots,B_{n/b}}\left|f_{B_{1}\cup\cdots\cup B_{n/b}}-a_{0}-\sum_{\ell=1}^{d}a_{\ell}\sum_{T\in\binom{[n/b]}{\ell}}\sum_{\substack{S\subseteq T\\
S\ne\varnothing
}
}(-1)^{|S|+1}\tilde{f}_{\bigcup_{i\in S}\!B_{i}}\right|\nonumber \\
 & \qquad\leq\max_{B_{1},\ldots,B_{n/b}}\left|f_{B_{1}\cup\cdots\cup B_{n/b}}-a_{0}-\sum_{\ell=1}^{d}a_{\ell}\sum_{T\in\binom{[n/b]}{\ell}}\sum_{\substack{S\subseteq T\\
S\ne\varnothing
}
}(-1)^{|S|+1}f_{\bigcup_{i\in S}\!B_{i}}\right|\nonumber \\
 & \qquad\qquad+\max_{B_{1},\ldots,B_{n/b}}\;\;\sum_{\ell=1}^{d}|a_{\ell}|\sum_{T\in\binom{[n/b]}{\ell}}\sum_{\substack{S\subseteq T\\
S\ne\varnothing
}
}\left|f_{\bigcup_{i\in S}B_{i}}-\tilde{f}_{\bigcup_{i\in S}\!B_{i}}\right|\allowdisplaybreaks\nonumber \\
 & \qquad\leq\max_{B_{1},\ldots,B_{n/b}}\left|f_{B_{1}\cup\cdots\cup B_{n/b}}-a_{0}-\sum_{\ell=1}^{d}a_{\ell}\sum_{T\in\binom{[n/b]}{\ell}}\sum_{\substack{S\subseteq T\\
S\ne\varnothing
}
}(-1)^{|S|+1}f_{\bigcup_{i\in S}\!B_{i}}\right|\nonumber \\
 & \qquad\qquad+\frac{\epsilon}{2}\nonumber \\
 & \qquad=\max_{B_{1},\ldots,B_{n/b}}\left|\bigvee_{i=1}^{n/b}f_{B_{i}}-a_{0}-\sum_{\ell=1}^{d}a_{\ell}\sum_{T\in\binom{[n/b]}{\ell}}\sum_{\substack{S\subseteq T\\
S\ne\varnothing
}
}(-1)^{|S|+1}\bigvee_{i\in S}f_{B_{i}}\right|+\frac{\epsilon}{2}\nonumber \\
 & \qquad=\max_{B_{1},\ldots,B_{n/b}}\left|\bigvee_{i=1}^{n/b}f_{B_{i}}-\sum_{\ell=0}^{d}a_{\ell}\sum_{T\in\binom{[n/b]}{\ell}}\prod_{i\in T}f_{B_{i}}\right|+\frac{\epsilon}{2}\nonumber \\
 & \qquad\leq\epsilon,\label{eq:decomp-error}
\end{align}
where the first step is immediate from~(\ref{eq:decomp-tilde-Fxy-simplified})
and~(\ref{eq:decomp-Fxy-simplified}), the second step applies the
triangle inequality, the third step is valid by~(\ref{eq:decomp-approx-fS-error}),
the fourth step is a change of notation, the fifth step uses the inclusion-exclusion
formula (Fact~\ref{fact:incl-excl-alternative}), and the last step
is justified by~(\ref{eq:decomp-tilde-OR-error}). By~(\ref{eq:tilde-F-degree})
and~(\ref{eq:decomp-error}), the proof of Lemma~\ref{lem:selector-homogeneous}
is complete by taking $C=4c.$
\end{proof}
To remove the homogeneity and divisibility assumptions in Lemma~\ref{lem:selector-homogeneous},
we now show how to reduce the approximation of any function on $\zoo_{\leq n}^{N}$
to the approximation of a closely related function on $\zoo_{n}^{N+n}.$
This connection is surprising at first but has a short proof based
on Minsky and Papert's symmetrization argument.
\begin{lem}[Homogenization lemma]
\label{lem:homogenization}Let $f\colon X\times\zoo_{\leq n}^{N}\to\Re$
be given. Define $f'\colon X\times\zoo_{n}^{N+n}\to\Re$ by
\begin{equation}
f'(x,y_{1}\ldots y_{N+n})=f(x,y_{1}\ldots y_{N}).\label{eq:homogenization}
\end{equation}
Then for all $\epsilon\geq0,$
\[
\deg_{\epsilon}(f')=\deg_{\epsilon}(f).
\]
\end{lem}

\begin{proof}
The upper bound $\deg_{\epsilon}(f')\leq\deg_{\epsilon}(f)$ is immediate
from the defining equation~(\ref{eq:homogenization}). For a matching
lower bound, fix a polynomial $\phi'\colon X\times\Re^{N+n}\to\Re$
such that
\begin{align}
 & |f'(x,y)-\phi'(x,y)|\leq\epsilon, &  & x\in X,\;y\in\zoo_{n}^{N+n},\label{eq:homog-approx}\\
 & \deg\phi'=\deg_{\epsilon}(f').\label{eq:homog-deg-tilde-f}
\end{align}
Minsky and Papert's symmetrization argument~(Proposition~\ref{prop:symmetrization})
yields a polynomial $\phi^{*}\colon X\times\Re^{N}\times\Re\to\Re$
such that for $t=0,1,2,\ldots,n,$
\begin{align}
 & \phi^{*}(x,y,t)=\Exp_{z\in\zoo_{t}^{n}}\phi'(x,yz), &  & x\in X,\;y\in\zoo^{N},\label{eq:homog-symm}\\
 & \deg\phi^{*}\leq\deg\phi'.\label{eq:homog-deg}
\end{align}
We are now in a position to construct the desired approximant for
$f.$ For any $x\in X$ and $y\in\zoo_{\leq n}^{N},$ we have
\begin{align*}
 & \left|f(x,y)-\phi^{*}\!\left(x,y,n-\sum_{i=1}^{N}y_{i}\right)\right|\\
 & \qquad\qquad\qquad\leq\left|f(x,y)-\Exp_{\substack{z\in\zoo_{n-|y|}^{n}}
}f'(x,yz)\right|\\
 & \qquad\qquad\qquad\qquad\qquad+\left|\Exp_{\substack{z\in\zoo_{n-|y|}^{n}}
}f'(x,yz)-\phi^{*}\!\left(x,y,n-\sum_{i=1}^{N}y_{i}\right)\right|\\
 & \qquad\qquad\qquad=\left|\Exp_{\substack{z\in\zoo_{n-|y|}^{n}}
}f'(x,yz)-\phi^{*}\!\left(x,y,n-\sum_{i=1}^{N}y_{i}\right)\right|\\
 & \qquad\qquad\qquad=\left|\Exp_{\substack{z\in\zoo_{n-|y|}^{n}}
}[f'(x,yz)-\phi'(x,yz)]\right|\\
 & \qquad\qquad\qquad\leq\epsilon,
\end{align*}
where the first step applies the triangle inequality, the second step
is immediate from the definition of $f',$ the third step uses~(\ref{eq:homog-symm}),
and the last step follows from~(\ref{eq:homog-approx}). In summary,
we have shown that $\deg_{\epsilon}(f)\leq\deg\phi^{*},$ which in
view of~(\ref{eq:homog-deg}) and~(\ref{eq:homog-deg-tilde-f})
completes the proof. 
\end{proof}
We are now in a position to remove the divisibility assumption in
Lemma~\ref{lem:selector-homogeneous} and additionally generalize
it to the nonhomogeneous setting.
\begin{thm}
\label{thm:selector}Let $F\colon X\times\zoo_{\leq n}^{N}\to\zoo$
be given by
\[
F(x,y)=\bigvee_{i=1}^{N}y_{i}\wedge f_{i}(x)
\]
for some functions $f_{1},f_{2},\ldots,f_{N}\colon X\to\zoo.$ Then
\begin{multline}
\deg_{\epsilon}(F)\leq C\sqrt{nb\log\frac{1}{\epsilon}}\\
+\max_{\substack{S\subseteq\{1,\ldots,N\}\\
|S|\leq C\sqrt{nb\log\frac{1}{\epsilon}}
}
}\deg_{\epsilon\exp\left(-C\sqrt{\frac{n}{b}\log\frac{1}{\epsilon}}\right)}\left(\bigvee_{i\in S}f_{i}\right)\qquad\label{eq:selector-general}
\end{multline}
for all reals $b\geq1$ and $0<\epsilon\leq1/2,$ where $C>1$ is
an absolute constant independent of $F,N,n,b,\epsilon.$
\end{thm}

\begin{proof}
We first examine the case $1\leq b\leq n.$ Consider the function
$F'\colon X\times\zoo_{n'}^{N'}\to\zoo$ given by
\[
F'(x,y)=\bigvee_{i=1}^{N'}y_{i}\wedge f_{i}(x),
\]
where
\begin{align*}
 & n'=\lfloor b\rfloor\left\lceil \frac{n}{\lfloor b\rfloor}\right\rceil ,\\
 & N'=\lfloor b\rfloor\left\lceil \frac{N}{\lfloor b\rfloor}\right\rceil +\lfloor b\rfloor\left\lceil \frac{n}{\lfloor b\rfloor}\right\rceil ,\\
\rule{0mm}{4mm} & f_{N+1}=f_{N+2}=\cdots=f_{N'}=0.
\end{align*}
Then
\begin{align*}
\deg_{\epsilon}(F) & \leq\deg_{\epsilon}(F')\\
 & \leq c\sqrt{n'\lfloor b\rfloor\log\frac{1}{\epsilon}}+\max_{\substack{S\subseteq\{1,\ldots,N'\}\\
|S|\leq c\sqrt{n'\lfloor b\rfloor\log\frac{1}{\epsilon}}
}
}\deg_{\epsilon\exp\left(-c\sqrt{\frac{n'}{\lfloor b\rfloor}\log\frac{1}{\epsilon}}\right)}\left(\bigvee_{i\in S}f_{i}\right),
\end{align*}
for some absolute constant $c\geq1,$ where the first step uses the
homogenization lemma (Lemma~\ref{lem:homogenization}) and the second
step follows from Lemma~\ref{lem:selector-homogeneous}. This settles~(\ref{eq:selector-general})
for $C=2c.$

For the complementary case $b\geq n,$ define $F'\colon X\times\zoo_{n}^{N+n}\to\zoo$
by
\[
F'(x,y)=\bigvee_{i=1}^{N+n}y_{i}\wedge f_{i}(x),
\]
where $f_{N+1}=f_{N+2}=\cdots=f_{N+n}=0.$ Then
\begin{equation}
\deg_{\epsilon}(F)=\deg_{\epsilon}(F')\label{eq:case-b-large}
\end{equation}
by the homogenization lemma (Lemma~\ref{lem:homogenization}). On
the other hand,
\begin{equation}
F'(x,y)=\sum_{S\in\binom{[N+n]}{n}}\left(\bigvee_{i\in S}f_{i}(x)\right)\prod_{i\in S}y_{i}.\label{eq:F'-exact}
\end{equation}
For any input $y$ of Hamming weight $n,$ every term in this summation
vanishes except for the term corresponding to $S=\{i:y_{i}=1\}.$
This means that an approximant for $F'$ with error $\epsilon$ can
be obtained by replacing each disjunction in~(\ref{eq:F'-exact})
with a polynomial that approximates that disjunction to within $\epsilon.$
As a result,
\[
\deg_{\epsilon}(F')\leq n+\max_{S\in\binom{[N]}{\leq n}}\deg_{\epsilon}\left(\bigvee_{i\in S}f_{i}\right).
\]
This upper bound along with~(\ref{eq:case-b-large}) settles~(\ref{eq:selector-general})
for $b\geq n.$
\end{proof}

\subsection{A recursive bound}

Using Theorem~\ref{thm:selector} as our main tool, we now derive
the promised recurrence for $D(n,k,\Delta).$
\begin{lem}
\label{lem:k-dnf-recursion}There is a constant $C\geq1$ such that
for all integers $n,k\geq1$ and reals $\Delta\geq1,$
\begin{equation}
D(n,k,\Delta)\leq\max_{b\geq1}\left\{ C\sqrt{nb\Delta}+D\!\left(n,k-1,\Delta+C\sqrt{\frac{n\Delta}{b}}\right)\right\} .\label{eq:k-dnf-recursion}
\end{equation}
\end{lem}

\begin{proof}
Let $f\colon\zoo_{\leq n}^{N}\to\zoo$ be a $k$-DNF formula. Our
objective is to bound $\deg_{2^{-\Delta}}(f)$ by the right-hand side
of~(\ref{eq:k-dnf-recursion}). We may assume that 
\begin{equation}
f\not\equiv1,\label{eq:f-not-1}
\end{equation}
since the bound holds trivially for the constant function $f=1$. 

Write $f=f'\vee f'',$ where $f'$ is a $k$-DNF formula in which
every term has an unnegated variable, and $f''$ is a $k$-DNF formula
whose terms feature only negated variables. Collecting like terms
in $f',$ we immediately obtain
\begin{equation}
f'(x)=\bigvee_{i=1}^{N}x_{i}\wedge f'_{i}(x)\label{eq:k-dnf-recursion-f'}
\end{equation}
for some $(k-1)$-DNF formulas $f_{1}',f_{2}',\ldots,f_{N}'.$

We now turn to $f''$. By~(\ref{eq:f-not-1}), there exists $x^{*}\in\zoo_{\leq n}^{N}$
such that $f''(x^{*})=0.$ Consider the subset $I=\{i:x_{i}^{*}=1\}$,
of cardinality 
\begin{equation}
|I|\leq n.\label{eq:k-dnf-recursion-I}
\end{equation}
 Since every occurrence of a variable in $f''(x)$ is negated, we
conclude that every term in $f''(x)$ features some literal $\overline{x_{i}}$
with $i\in I.$ Collecting like terms, we obtain the representation
\begin{equation}
f''(x)=\bigvee_{i\in I}\overline{x_{i}}\wedge f''_{i}(x),\label{eq:k-dnf-recursion-f''}
\end{equation}
where each $f''_{i}$ is a $(k-1)$-DNF formula.

To summarize~(\ref{eq:k-dnf-recursion-f'})\textendash (\ref{eq:k-dnf-recursion-f''}),
the function $f=f'\vee f''$ is a subfunction of some $F\colon\zoo_{\leq n}^{N}\times\zoo_{\leq2n}^{N+n}\to\zoo$
of the form
\[
F(x,y)=\bigvee_{i=1}^{N+n}y_{i}\wedge f_{i}(x),
\]
where each $f_{i}$ is a $(k-1)$-DNF formula. Now
\begin{align*}
\deg_{2^{-\Delta}}(f) & \leq\deg_{2^{-\Delta}}(F)\\
 & \leq\max_{b\geq1}\left\{ c\sqrt{2nb\Delta}+\max_{\substack{S\subseteq\{1,2\ldots,N+n\}}
}\deg_{2^{-\Delta}\exp(-c\sqrt{2n\Delta/b})}\left(\bigvee_{i\in S}f_{i}\right)\right\} \\
 & \leq\max_{b\geq1}\left\{ c\sqrt{2nb\Delta}+D\!\left(n,k-1,\Delta+\frac{c}{\ln2}\sqrt{\frac{2n\Delta}{b}}\right)\right\} ,
\end{align*}
where the second step follows from Theorem~\ref{thm:selector} for
a suitable absolute constant $c\geq1,$ and the third step is justified
by the fact that each $\bigvee_{i\in S}f_{i}\colon\zoo_{\leq n}^{N}\to\zoo$
is a $(k-1)$-DNF formula. In conclusion, (\ref{eq:k-dnf-recursion})
holds with $C=c\sqrt{2}/\ln2.$
\end{proof}

\subsection{\label{subsec:k-dnf-solving-the-recurrence}Solving the recurrence}

It remains to solve the recurrence for $D(n,k,\Delta)$ given by~(\ref{eq:D-nkr-base})
and Lemma~\ref{lem:k-dnf-recursion}.
\begin{thm}
\label{thm:D-nkr-final}There is a constant $c\geq1$ such that for
all integers $n,k\geq0$ and reals $\Delta\geq1,$
\begin{equation}
D(n,k,\Delta)\leq c\cdot(\sqrt{2})^{k}\,n^{\frac{k}{k+1}}\,\Delta^{\frac{1}{k+1}}.\label{eq:k-dnf-recurrence-solved}
\end{equation}
\end{thm}

\noindent This result settles Theorem~\ref{thm:k-DNF}. Indeed, if
$f\colon\zoo_{\leq n}^{N}\to\zoo$ is representable by a $k$-DNF
formula, then (\ref{eq:k-dnf-main}) is immediate from (\ref{eq:k-dnf-recurrence-solved}).
The same bound applies to $k$-CNF formulas because they are negations
of $k$-DNF formulas, and $\deg_{\epsilon}(f)=\deg_{\epsilon}(1-f)$
for any $f.$
\begin{proof}[Proof of Theorem~\emph{\ref{thm:D-nkr-final}}.]
 We will prove~(\ref{eq:k-dnf-recurrence-solved}) for $c=2(C+1)^{2},$
where $C\geq1$ is the absolute constant from Lemma~\ref{lem:k-dnf-recursion}.
The proof is by induction on $k.$ The base $k=0$ is valid due to~(\ref{eq:D-nkr-base}).
For the inductive step, let $k\geq1$ be arbitrary. For $\Delta\geq n,$
the claim is immediate from~(\ref{eq:D-nkr-trivial}), and we focus
on the complementary case
\begin{equation}
1\leq\Delta\leq n.\label{eq:eps-bounded}
\end{equation}
For every $b\geq1,$
\begin{align}
D(n,k,\Delta) & \leq\min\left\{ n,C\sqrt{nb\Delta}+D\!\left(n,k-1,\Delta+C\sqrt{\frac{n\Delta}{b}}\right)\right\} \nonumber \\
 & \leq\min\left\{ n,C\sqrt{nb\Delta}+2(C+1)^{2}\,2^{\frac{k-1}{2}}n^{\frac{k-1}{k}}\left(\Delta+C\sqrt{\frac{n\Delta}{b}}\right)^{\frac{1}{k}}\right\} \nonumber \\
 & \leq(C+1)\sqrt{nb\Delta}+(C+1)^{2}\,2^{\frac{k+1}{2}}n^{\frac{k-1}{k}}\left((C+1)\sqrt{\frac{n\Delta}{b}}\right)^{\frac{1}{k}},\label{eq:D-nkr-solving-recurrence}
\end{align}
where the first step uses (\ref{eq:D-nkr-trivial}) and~Lemma~\ref{lem:k-dnf-recursion};
the second step applies the inductive hypothesis; and the last step
can be verified in a straightforward manner by examining the cases
$\Delta\leq n/b$ and $\Delta\geq n/b.$ Setting
\[
b=(C+1)^{2}\,2^{k}\left(\frac{n}{\Delta}\right)^{1-\frac{2}{k+1}}
\]
in~(\ref{eq:D-nkr-solving-recurrence}) now yields~(\ref{eq:k-dnf-recurrence-solved}),
completing the inductive step. Note that our choice of parameter meets
the requirement $b\geq1$, as one can see from~(\ref{eq:eps-bounded}).
\end{proof}

\section{\emph{\label{sec:k-Distinctness}k}-Element distinctness}

For an integer $k$, recall that the\emph{ threshold function} $\THR_{k}\colon\zoo^{*}\to\zoo$
is given by
\[
\THR_{k}(x)=\begin{cases}
1 & \text{if }|x|\geq k,\\
0 & \text{otherwise.}
\end{cases}
\]
As a generate case, we have
\begin{equation}
\THR_{0}\equiv1.\label{eq:THR0}
\end{equation}
In the \emph{$k$-element distinctness problem}, the input is a list
of $n$ integers from some range of size $r,$ and the objective is
to determine whether some integer occurs at least $k$ times. Traditionally,
the input to $k$-element distinctness is represented by a Boolean
matrix $x\in\zoo^{n\times r}$ with precisely one nonzero entry in
each row. We depart from tradition by allowing the input $x\in\zoo^{n\times r}$
to be an arbitrary matrix with at most $n$ ones. Formally, we define
the $k$-element distinctness function $\ED_{n,r,k}\colon\zoo_{\leq n}^{nr}\to\zoo$
by 
\[
\ED_{n,r,k}(x)=\neg\bigvee_{i=1}^{r}\THR_{k}(x_{1,i}x_{2,i}\ldots x_{n,i}).
\]
Since our focus is on upper bounds, working with the more general
domain makes our results stronger. Our main result in this section
is as follows.
\begin{thm}
\label{thm:ED}Let $k\geq1$ be a fixed integer. Then for all integers
$n,r\geq1$ and all reals $0<\epsilon\leq1/2,$
\begin{multline*}
\deg_{\epsilon}(\ED_{n,r,k})=O\left(\sqrt{n}\min\{n,r\}^{\frac{1}{2}-\frac{1}{4(1-2^{-k})}}\,\left(\log\frac{1}{\epsilon}\right)^{\frac{1}{4(1-2^{-k})}}\right)\\
+O\left(\sqrt{n\log\frac{1}{\epsilon}}\right).
\end{multline*}
Moreover, the approximating polynomial is given explicitly in each
case.
\end{thm}

\noindent Taking $\epsilon=1/3$ in this result settles Theorem~\ref{thm:MAIN-ed}
from the introduction. To prove Theorem~\ref{thm:ED}, we will need
to consider a more general class of functions. For nonnegative integers
$n,r,k$ and a real number $\Delta\geq1,$ we define
\[
D(n,r,k,\Delta)=\max_{F}\;\deg_{2^{-\Delta}}(F),
\]
where the maximum is over all functions $F\colon\zoo_{\leq n}^{N}\to\zoo$
for some $N$ that are expressible as
\[
F(x)=\bigvee_{i=1}^{r}\THR_{k_{i}}(x|_{S_{i}})
\]
for some pairwise disjoint sets $S_{1},S_{2},\ldots,S_{r}\subseteq\{1,2,\ldots,N\}$
and some $k_{1},k_{2},\ldots,k_{r}\in\{0,1,2,\ldots,k\}.$ The four-argument
quantity $D$ that we have just defined is unrelated to the three-argument
quantity $D$ from Section~\ref{sec:k-DNF-and-k-CNF}. We abbreviate
\[
D(n,\infty,k,\Delta)=\max_{r\geq1}D(n,r,k,\Delta).
\]
By definition,
\begin{align}
\deg_{\epsilon}(\ED_{n,r,k}) & \leq D\left(n,r,k,\log\frac{1}{\epsilon}\right), &  & 0<\epsilon\leq\frac{1}{2}.\label{eq:Dnrkd-vs-ED}
\end{align}
Our analysis of $D(n,r,k,\Delta)$ proceeds by induction on $k.$
As the base cases, we have
\begin{align}
 & D(n,\infty,0,\Delta)=0\label{eq:ed-recursive-base-case}
\end{align}
by~(\ref{eq:THR0}), and 
\begin{equation}
D(n,\infty,1,\Delta)=C\sqrt{n\Delta}\label{eq:ed-recursive-base-case-OR}
\end{equation}
 by Theorem~\ref{thm:DISJUNCTION} for some constant $C\geq1.$ Also,
Fact~\ref{fact:trivial-upper-bound-on-deg} implies that
\begin{equation}
D(n,\infty,k,\Delta)\leq n.\label{eq:ed-recursive-trivial}
\end{equation}

\subsection{A recursive bound for small range}

To implement the inductive step, we derive two complementary recursive
bounds for $D(n,r,k,\Delta)$. The first of these bounds, presented
below, is tailored to the case when $n\geq kr.$
\begin{lem}
\label{lem:ed-recursive-small-range}There is a constant $C\geq1$
such that for all positive integers $n,r,k$ and all reals $\Delta\geq1,$
\[
D(n,r,k,\Delta)\leq C\cdot\sqrt{1+\frac{n}{kr}}\cdot(D(2kr,r,k,\Delta+1)+\Delta).
\]
\end{lem}

\begin{proof}
Since $D$ is monotonically increasing in every argument, the lemma
holds trivially for $n<kr.$ In what follows, we consider the complementary
case
\begin{equation}
n\geq kr.\label{eq:n-kr}
\end{equation}
Consider an arbitrary function $F\colon\zoo_{\leq n}^{N}\to\zoo$
of the form
\begin{equation}
F(x)=\bigvee_{i=1}^{r}\THR_{k_{i}}(x|_{S_{i}})\label{eq:ed-recursive2-F}
\end{equation}
for some pairwise disjoint sets $S_{1},S_{2},\ldots,S_{r}\subseteq\{1,2,\ldots,N\}$
and $k_{1},k_{2},\ldots,k_{r}\in\{0,1,2,\ldots,k\}.$ By discarding
any irrelevant variables among $x_{1},x_{2},\ldots,x_{N}$, we may
assume that $S_{1}\cup S_{2}\cup\cdots\cup S_{r}=\{1,2,\ldots,N\}.$
Then by the pigeonhole principle, any input $x$ with Hamming weight
at least $kr$ satisfies at least one of the disjuncts in~(\ref{eq:ed-recursive2-F}).
Therefore,
\begin{align}
F(x) & =1, &  & x\in\zoo_{\geq kr}^{N}.\label{eq:F-const-beyond-kr}
\end{align}

For $i\geq kr,$ define $F_{i}\colon\zoo_{\leq i}^{N}\to\zoo$ by
\[
F_{i}(x)=\begin{cases}
F(x) & \text{if }|x|\leq kr,\\
1 & \text{otherwise.}
\end{cases}
\]
Then
\begin{align*}
\deg_{2^{-\Delta}}(F) & =\deg_{2^{-\Delta}}(F_{n})\\
 & =\deg_{2^{-\Delta}}(1-F_{n})\\
 & \leq c\sqrt{\frac{n}{kr}}\cdot(\deg_{2^{-\Delta-1}}(1-F_{2kr})+\Delta+1)\\
 & =c\sqrt{\frac{n}{kr}}\cdot(\deg_{2^{-\Delta-1}}(F_{2kr})+\Delta+1)\\
 & \leq c\sqrt{\frac{n}{kr}}\cdot(D(2kr,r,k,\Delta+1)+\Delta+1)
\end{align*}
for some absolute constant $c\geq1$ and all $\Delta\geq1,$ where
the first and last steps use~(\ref{eq:F-const-beyond-kr}), and the
third step applies~(\ref{eq:n-kr}) and the extension theorem (Theorem~\ref{thm:extension})
with $m=kr$ and $\epsilon=\delta=2^{-\Delta-1}$. As a result, the
lemma holds with $C=2c.$
\end{proof}

\subsection{A recursive bound for large range}

We now derive an alternate upper bound on $D(n,r,k,\Delta),$ with
no dependence on the range parameter $r$. This result addresses the
case of large $r$ and complements Lemma~\ref{lem:ed-recursive-small-range}.
\begin{lem}
\label{lem:ed-recursive-large-range}There is a constant $C\geq1$
such that for all integers $n,k\geq1$ and all reals $\Delta,b\geq1,$
\begin{multline}
D(n,\infty,k,\Delta)\leq C\sqrt{nb\Delta}+C\left(1+\frac{1}{\sqrt{k}}\left(\frac{n}{b\Delta}\right)^{1/4}\right)\times\\
\times\left(D\left(\lfloor Ck\sqrt{nb\Delta}\rfloor,\infty,k-1,C\sqrt{\frac{n\Delta}{b}}+1\right)+\sqrt{\frac{n\Delta}{b}}\right).\label{eq:recursive1}
\end{multline}
\end{lem}

\begin{proof}
Consider an arbitrary function $F\colon\zoo_{\leq n}^{N}\to\zoo$
of the form
\begin{equation}
F(x)=\bigvee_{i=1}^{r}\THR_{k_{i}}(x|_{S_{i}})\label{eq:F-or-of-thrs}
\end{equation}
for some integer $r\geq1,$ some pairwise disjoint sets $S_{1},S_{2},\ldots,S_{r}\subseteq\{1,2,\ldots,N\}$,
and some $k_{1},k_{2},\ldots,k_{r}\in\{0,1,2,\ldots,k\}.$ If $k_{i}=0$
for some $i,$ then the corresponding term in~(\ref{eq:F-or-of-thrs})
is the constant function~$1,$ resulting in $\deg_{0}(F)=0.$ In
what follows, we treat the complementary case when $k_{i}\geq1$ for
each $i.$

Rewriting~(\ref{eq:F-or-of-thrs}),
\begin{equation}
F(x)=\bigvee_{i=1}^{r}\bigvee_{j\in S_{i}}x_{j}\wedge\THR_{k_{i}-1}(x|_{S_{i}\setminus\{j\}}).\label{eq:ed-recursive1-F}
\end{equation}
As this representation suggests, our intention is to bound the approximate
degree of $F$ by appeal to Theorem~\ref{thm:selector}.
\begin{claim}
\label{claim:ed-recursive1}Fix a subset $S_{i}'\subseteq S_{i}$
for each $i=1,2,\ldots,r.$ Then for $\Delta\geq1,$
\begin{multline*}
\deg_{2^{-\Delta}}\left(\bigvee_{i=1}^{r}\bigvee_{j\in S'_{i}}\THR_{k_{i}-1}(x|_{S_{i}\setminus\{j\}})\right)\\
\leq D\left(n,\sum_{i=1}^{r}|S_{i}'|,k-1,\Delta\right)+\sum_{i=1}^{r}|S'_{i}|.\qquad\qquad
\end{multline*}
 
\end{claim}

\noindent We will settle Claim~\ref{claim:ed-recursive1} once we
complete the main proof. In light of this claim, the representation~(\ref{eq:ed-recursive1-F})
shows that
\[
F(x)=\bigvee_{i=1}^{N}x_{i}\wedge f_{i}(x)
\]
for some functions $f_{i}$ such that
\begin{align}
\deg_{2^{-\Delta}}\left(\bigvee_{i\in S}f_{i}\right) & \leq D(n,|S|,k-1,\Delta)+|S|\label{eq:ors-of-thrs}
\end{align}
for all $S\subseteq\{1,2,\ldots,N\}$ and all $\Delta\geq1.$ Then
for some absolute constants $c',c''\geq1$ and all $\Delta\geq1$
and $b\geq1$, we have
\begin{align*}
\deg_{2^{-\Delta}}(F) & \leq c'\sqrt{nb\Delta}+\max_{\substack{S\subseteq\{1,\ldots,N\}\\
|S|\leq c'\sqrt{nb\Delta}
}
}\deg_{2^{-\Delta}\exp\left(-c'\sqrt{n\Delta/b}\right)}\left(\bigvee_{i\in S}f_{i}\right)\\
 & \leq2c'\sqrt{nb\Delta}+D\left(n,\lceil c'\sqrt{nb\Delta}\rceil,k-1,\Delta+\frac{c'}{\ln2}\sqrt{\frac{n\Delta}{b}}\right)\\
 & \leq2c'\sqrt{nb\Delta}+c''\cdot\sqrt{1+\frac{n}{k\cdot c'\sqrt{nb\Delta}}}\times\phantom{a}\\
 & \qquad\times\left(D\left(2k\lceil c'\sqrt{nb\Delta}\rceil,\infty,k-1,\Delta+\frac{c'}{\ln2}\sqrt{\frac{n\Delta}{b}}+1\right)\right.\\
 & \qquad\qquad\qquad\qquad\qquad\qquad\qquad\qquad\qquad\left.\phantom{a}+\Delta+\frac{c'}{\ln2}\sqrt{\frac{n\Delta}{b}}\right),
\end{align*}
where the first step applies Theorem~\ref{thm:selector}, the second
step uses~(\ref{eq:ors-of-thrs}), and the final step follows from~(\ref{eq:ed-recursive-base-case})
for $k=1$ and from Lemma~\ref{lem:ed-recursive-small-range} for
$k\geq2$. This directly implies~(\ref{eq:recursive1}) for $\Delta\leq n/b.$
In the complementary case $\Delta>n/b,$ the right-hand side of (\ref{eq:recursive1})
exceeds $n$ and therefore the bound follows trivially from~(\ref{eq:ed-recursive-trivial}).
\end{proof}
\begin{proof}[Proof of Claim~\emph{\ref{claim:ed-recursive1}}]
To start with,
\begin{align*}
\bigvee_{i=1}^{r}\bigvee_{j\in S'_{i}}\THR_{k_{i}-1}(x|_{S_{i}\setminus\{j\}}) & =\bigvee_{i:S'_{i}\ne\varnothing}\;\bigvee_{j\in S'_{i}}\THR_{k_{i}-1}(x|_{S_{i}\setminus\{j\}})\\
 & =\bigvee_{i:S'_{i}\ne\varnothing}\THR_{k_{i}-1-\min\left\{ \left|x|_{S'_{i}}\right|,\,|S_{i}'|-1\right\} }(x|_{S_{i}\setminus S_{i}'}).\\
\end{align*}
Considering the possible values for the Hamming weight of each $x|_{S'_{i}}$,
we arrive at the representation
\begin{multline}
\bigvee_{i=1}^{r}\;\bigvee_{j\in S'_{i}}\THR_{k_{i}-1}(x|_{S_{i}\setminus\{j\}})=\sum_{\ell_{1}=0}^{|S'_{1}|}\cdots\sum_{\ell_{r}=0}^{|S'_{r}|}\I[|x|_{S'_{i}}|=\ell_{i}\text{ for each }i]\\
\times\left(\bigvee_{i:S_{i}'\ne\varnothing}\THR_{k_{i}-1-\min\{\ell_{i},\,|S_{i}'|-1\}}(x|_{S_{i}\setminus S_{i}'})\right).\label{eq:approx-disjunction-of-thrs}
\end{multline}
The indicator functions in this summation are mutually exclusive in
that for any given value of $x,$ precisely one of them is nonzero.
As a result, the right-hand side of~(\ref{eq:approx-disjunction-of-thrs})
can be approximated pointwise to within $2^{-\Delta}$ by replacing
each parenthesized expression with its $2^{-\Delta}$-error approximant,
which by definition can be chosen to have degree at most $D(n,\sum|S_{i}'|,k-1,\Delta).$
This completes the proof since each indicator function in~(\ref{eq:approx-disjunction-of-thrs})
depends on only $\sum|S_{i}'|$ Boolean variables and is therefore
a polynomial of degree at most $\sum|S'_{i}|.$
\end{proof}

\subsection{Solving the recurrence}

It remains to solve the newly obtained recurrences. We first solve
the recurrence given by (\ref{eq:ed-recursive-base-case-OR}) and
Lemma~\ref{lem:ed-recursive-large-range}, corresponding to the infinite-range
case.
\begin{thm}[Range-independent bound]
\label{thm:ed-range-independent}There is a constant $c\geq1$ such
that for all positive integers $n$ and $k,$ and all reals $\Delta\geq1,$
\begin{align}
D(n,\infty,k,\Delta) & \leq c^{k}\sqrt{k!}\cdot n^{1-\frac{1}{4(1-2^{-k})}}\Delta^{\frac{1}{4(1-2^{-k})}}.\label{eq:main-recursive-large-range-k>0}
\end{align}
\end{thm}

\begin{proof}
We will prove~(\ref{eq:main-recursive-large-range-k>0}) for $c=(4C)^{2},$
where $C\geq1$ is the larger of the constants in~(\ref{eq:ed-recursive-base-case-OR})
and Lemma~\ref{lem:ed-recursive-large-range}. The proof is by induction
on $k.$ The base case $k=1$ is immediate from~(\ref{eq:ed-recursive-base-case-OR}).
For the inductive step, let $k\geq2$ be arbitrary. When $\Delta>n,$
the right-hand side of~(\ref{eq:main-recursive-large-range-k>0})
exceeds $n$ and therefore the bound is immediate from~(\ref{eq:ed-recursive-trivial}).
In what follows, we assume that 
\begin{equation}
1\leq\Delta\leq n.\label{eq:ed-Delta-not-too-large}
\end{equation}
Let $b\geq1$ be a parameter to be fixed later. By Lemma~\ref{lem:ed-recursive-large-range},
\begin{multline*}
D(n,\infty,k,\Delta)\leq C\sqrt{nb\Delta}+C\left(1+\frac{1}{\sqrt{k}}\left(\frac{n}{b\Delta}\right)^{\frac{1}{4}}\right)\left(\sqrt{\frac{n\Delta}{b}}\right.\\
\left.\phantom{a}+D\left(\lfloor Ck\sqrt{nb\Delta}\rfloor,\infty,k-1,C\sqrt{\frac{n\Delta}{b}}+1\right)\right).
\end{multline*}
It follows that
\begin{multline*}
D(n,\infty,k,\Delta)\leq C\sqrt{knb\Delta}+2C\left(\frac{n}{kb\Delta}\right)^{\frac{1}{4}}\left(\sqrt{\frac{n\Delta}{b}}\right.\\
\left.\phantom{a}+D\left(\lfloor Ck\sqrt{nb\Delta}\rfloor,\infty,k-1,C\sqrt{\frac{n\Delta}{b}}+1\right)\right),
\end{multline*}
as one can verify from the previous step if $n\geq kb\Delta$ and
from~(\ref{eq:ed-recursive-trivial}) if $n<kb\Delta.$ Applying
the inductive hypothesis,
\begin{multline*}
D(n,\infty,k,\Delta)\leq C\sqrt{knb\Delta}+2C\left(\frac{n}{kb\Delta}\right)^{\frac{1}{4}}\left(\sqrt{\frac{n\Delta}{b}}\right.\\
\left.\phantom{a}+c^{k-1}\sqrt{(k-1)!}\cdot(Ck\sqrt{nb\Delta})^{1-\frac{1}{4(1-2^{-k+1})}}\left(C\sqrt{\frac{n\Delta}{b}}+1\right)^{\frac{1}{4(1-2^{-k+1})}}\right).
\end{multline*}
Now the bound
\begin{multline*}
D(n,\infty,k,\Delta)\leq C\sqrt{knb\Delta}+2C\left(\frac{n}{kb\Delta}\right)^{\frac{1}{4}}\times\\
\times4c^{k-1}\sqrt{(k-1)!}\cdot(Ck\sqrt{nb\Delta})^{1-\frac{1}{4(1-2^{-k+1})}}\left(C\sqrt{\frac{n\Delta}{b}}\right)^{\frac{1}{4(1-2^{-k+1})}}
\end{multline*}
is immediate from the previous step if $n\geq b/\Delta$ and from~(\ref{eq:ed-recursive-trivial})
if $n<b/\Delta.$ Rearranging, we find that
\begin{multline}
D(n,\infty,k,\Delta)\\
\leq C\sqrt{knb\Delta}\left(1+C\left(\frac{n}{\Delta}\right)^{\frac{1}{4}}\cdot8c^{k-1}\sqrt{(k-1)!}\;b^{-\frac{1}{4}-\frac{1}{4(1-2^{-k+1})}}\right).\label{eq:D-infinity-induction-simplified}
\end{multline}
The right-hand side is minimized at 
\[
b=\left(C\left(\frac{n}{\Delta}\right)^{\frac{1}{4}}\cdot8c^{k-1}\sqrt{(k-1)!}\right)^{\frac{2^{k+1}-4}{2^{k}-1}},
\]
which in view of~(\ref{eq:ed-Delta-not-too-large}) is a real number
in $[1,\infty)$ and therefore a legitimate parameter setting. Making
this substitution in~(\ref{eq:D-infinity-induction-simplified}),
we arrive at
\begin{align*}
D(n,\infty,k,\Delta) & \leq2C\sqrt{kn\Delta}\left(C\left(\frac{n}{\Delta}\right)^{\frac{1}{4}}\cdot8c^{k-1}\sqrt{(k-1)!}\right)^{\frac{2^{k}-2}{2^{k}-1}}\\
 & \leq2C^{2}\cdot8c^{k-1}\sqrt{k!\,n\Delta\left(\frac{n}{\Delta}\right)^{\frac{2^{k-1}-1}{2^{k}-1}}}\\
 & =c^{k}\sqrt{k!}\,n^{1-\frac{1}{4(1-2^{-k})}}\Delta^{\frac{1}{4(1-2^{-k})}}.
\end{align*}
This completes the inductive step and settles~(\ref{eq:main-recursive-large-range-k>0}).
\end{proof}
By combining the previous result with an application of Lemma~\ref{lem:ed-recursive-small-range},
we will now prove our main bound on $D(n,r,k,\Delta).$
\begin{thm}[Range-dependent bound]
\label{thm:ed-range-dependent}There is a constant $c\geq1$ such
that for all positive integers $n,r,k$ and all reals $\Delta\geq1,$
\begin{align*}
D(n,r,k,\Delta) & \leq c^{k}\sqrt{k!}\left(\sqrt{n}\min\{n,kr\}^{\frac{1}{2}-\frac{1}{4(1-2^{-k})}}\,\Delta^{\frac{1}{4(1-2^{-k})}}+\sqrt{n\Delta}\right).
\end{align*}
\end{thm}

\begin{proof}
The bound follows from Theorem~\ref{thm:ed-range-independent} if
$kr\geq n$; and from~(\ref{eq:ed-recursive-trivial}) if $\Delta\geq n.$
As a result, we may assume that
\begin{align}
n & >kr,\label{eq:ed-range-dep-n-kr}\\
n & >\Delta.\label{eq:ed-range-dep-n-Delta}
\end{align}
In what follows, let $C\geq1$ denote the larger of the constants
in Lemma~\ref{lem:ed-recursive-small-range} and Theorem~\ref{thm:ed-range-independent}.
Then
\begin{align*}
 & D(n,r,k,\Delta)\\
 & \qquad\leq D\left(n,r+\left\lceil \frac{\Delta}{k}\right\rceil ,k,\Delta\right)\\
 & \qquad\leq C\cdot\sqrt{1+\frac{n}{kr+k\lceil\Delta/k\rceil}}\cdot\left(D\left(2kr+2k\left\lceil \frac{\Delta}{k}\right\rceil ,\infty,k,\Delta+1\right)+\Delta\right)\\
 & \qquad\leq2C\cdot\sqrt{\frac{n}{kr+k\lceil\Delta/k\rceil}}\cdot\left(D\left(2kr+2k\left\lceil \frac{\Delta}{k}\right\rceil ,\infty,k,\Delta+1\right)+\Delta\right)\\
 & \qquad\leq2C\cdot\sqrt{\frac{n}{kr+k\lceil\Delta/k\rceil}}\\
 & \qquad\qquad\qquad\times\left(C^{k}\sqrt{k!}\left(2kr+2k\left\lceil \frac{\Delta}{k}\right\rceil \right)^{1-\frac{1}{4(1-2^{-k})}}(\Delta+1)^{\frac{1}{4(1-2^{-k})}}+\Delta\right)\allowdisplaybreaks\\
 & \qquad\leq2C\cdot\sqrt{\frac{n}{kr+k\lceil\Delta/k\rceil}}\\
 & \qquad\qquad\qquad\times2C^{k}\sqrt{k!}\left(2kr+2k\left\lceil \frac{\Delta}{k}\right\rceil \right)^{1-\frac{1}{4(1-2^{-k})}}(\Delta+1)^{\frac{1}{4(1-2^{-k})}}\\
 & \qquad=4C^{k+1}\sqrt{k!}\cdot\sqrt{2n}\left(2kr+2k\left\lceil \frac{\Delta}{k}\right\rceil \right)^{\frac{1}{2}-\frac{1}{4(1-2^{-k})}}(\Delta+1)^{\frac{1}{4(1-2^{-k})}},
\end{align*}
where the first step is valid because $D$ is monotonically increasing
in every argument; the second step applies Lemma~\ref{lem:ed-recursive-small-range};
the third step uses~(\ref{eq:ed-range-dep-n-kr}) and~(\ref{eq:ed-range-dep-n-Delta});
and the fourth step applies Theorem~\ref{thm:ed-range-independent}.
This completes the proof of the theorem for $n>kr.$
\end{proof}
\noindent Equation~(\ref{eq:Dnrkd-vs-ED}) and Theorem~\ref{thm:ed-range-dependent}
establish the main result of this section, Theorem~\ref{thm:ED}.
We note that with a more careful analysis, the multiplicative factor
$\sqrt{k!}$ in Theorems~\ref{thm:ed-range-independent} and~\ref{thm:ed-range-dependent}
can be improved to a slightly smaller quantity, still of the order
of $k^{O(k)}.$

\section{Surjectivity}

For positive integers $n$ and $r$, the \emph{surjectivity problem}
is to determine whether a given mapping $\{1,2,\ldots,n\}\to\{1,2,\ldots,r\}$
is surjective. Traditionally, the input to this problem is represented
by a Boolean matrix $x\in\zoo^{n\times r}$ with precisely one nonzero
entry in every row. Analogous to our work on element distinctness
in the previous section, we depart from tradition by allowing arbitrary
matrices $x\in\zoo^{n\times r}$ with at most $n$ ones. Specifically,
we define the surjectivity function $\SURJ_{n,r}\colon\zoo_{\leq n}^{nr}\to\zoo$
by
\[
\SURJ_{n,r}(x)=\bigwedge_{j=1}^{r}\bigvee_{i=1}^{n}x_{i,j}.
\]
This formalism corresponds to determining the surjectivity of arbitrary
relations on $\{1,2,\ldots,n\}\times\{1,2,\ldots,r\}$, including
functions $\{1,2,\ldots,n\}\to\{1,2,\ldots,r\}$ as a special case.
Since we are interested in upper bounds, working in this more general
setting makes our results stronger.

\subsection{Approximation to 1/3}

For clarity of exposition, we first bound the approximate degree of
surjectivity with the error parameter set to $\epsilon=1/3.$ This
setting covers most applications of interest and allows for a shorter
and simpler proof. Readers with an interest in general $\epsilon$
can skip directly to Section~\ref{subsec:Approximation-to-any-eps}.
\begin{thm*}[restatement of Theorem~\ref{thm:MAIN-surj}]
For all positive integers $n$ and $r,$
\begin{align}
 & \deg_{1/3}(\SURJ_{n,r})=O(\sqrt{n}\cdot r^{1/4}) &  & (r\leq n),\label{eq:surj-small-range-0.33333}\\
 & \deg_{1/3}(\SURJ_{n,r})=0 &  & (r>n).\label{eq:surj-large-range-0.33333}
\end{align}
Moreover, the approximating polynomial is given explicitly in each
case.
\end{thm*}
\noindent The theorem shows that $\deg_{1/3}(\SURJ_{n,r})=O(n^{3/4})$
for all $r,$ disproving the conjecture of Bun and Thaler~\cite{bun-thaler17adeg-ac0}
that the $1/3$-approximate degree of $\SURJ_{n,\Omega(n)}$ is linear
in $n.$ 
\begin{proof}
The identity $\SURJ_{n,r}\equiv0$ for $r>n$ implies (\ref{eq:surj-large-range-0.33333})
directly. The proof of~(\ref{eq:surj-small-range-0.33333}) involves
two steps. First, we construct an explicit real-valued function $\widetilde{\SURJ}_{n,r}$
that approximates $\SURJ_{n,r}$ pointwise and is representable by
a linear combination of conjunctions with reasonably small coefficients.
Then, we replace each conjunction in this linear combination by an
approximating polynomial of low degree.

In more detail, let $m\geq1$ be an integer parameter to be chosen
later. Recall from~(\ref{eqn:chebyshev-containment}) and~Proposition~\ref{prop:chebyshev-beyond-1}
that the Chebyshev polynomial $T_{m}$ obeys 
\begin{align*}
 & |T_{m}(t)|\leq1, &  & -1\leq t\leq t,\\
 & T_{m}\left(1+\frac{1}{r}\right)\geq1+\frac{m^{2}}{r}.
\end{align*}
As a result, $\SURJ_{n,r}$ is approximated pointwise within $1/(1+\frac{m^{2}}{r})$
by
\[
\widetilde{\SURJ}_{n,r}(x)=\frac{1}{T_{m}(1+\frac{1}{r})}\cdot T_{m}\left(\frac{1}{r}+\frac{1}{r}\sum_{j=1}^{r}\bigvee_{i=1}^{n}x_{i,j}\right).
\]
Therefore,
\begin{align}
E(\SURJ_{n,r},d) & \leq\frac{1}{1+\frac{m^{2}}{r}}+E(\widetilde{\SURJ}_{n,r},d), &  & d=1,2,3,\ldots.\label{eq:approximate-surj}
\end{align}

To estimate the rightmost term in~(\ref{eq:approximate-surj}), use
the factored representation~(\ref{eq:chebyshev-factored}) to write
\begin{align*}
\widetilde{\SURJ}_{n,r}(x) & =\frac{2^{m-1}}{T_{m}(1+\frac{1}{r})}\cdot\prod_{i=1}^{m}\left(\frac{1}{r}+\frac{1}{r}\sum_{j=1}^{r}\left(\bigvee_{i=1}^{n}x_{i,j}\right)-\cos\frac{(2i-1)\pi}{2m}\right)\\
 & =\frac{2^{m-1}}{T_{m}(1+\frac{1}{r})}\cdot\prod_{i=1}^{m}\left(\frac{1}{r}+1-\frac{1}{r}\sum_{j=1}^{r}\prod_{i=1}^{n}\overline{x_{i,j}}-\cos\frac{(2i-1)\pi}{2m}\right).
\end{align*}
Multiplying out shows that $\widetilde{\SURJ}_{n,r}(x)$ is a linear
combination of conjunctions with real coefficients whose absolute
values sum to $2^{O(m)}.$ By Corollary~\ref{cor:CONJUNCTION}, each
of these conjunctions can be approximated by a polynomial of degree
$d$ to within  $2^{-\Theta(d^{2}/n)}$ pointwise. We conclude that
\[
E(\widetilde{\SURJ}_{n,r},d)\leq2^{O(m)}\cdot2^{-\Theta(d^{2}/n)},
\]
which along with (\ref{eq:approximate-surj}) gives
\[
E(\SURJ_{n,r},d)\leq\frac{1}{1+\frac{m^{2}}{r}}+2^{O(m)}\cdot2^{-\Theta(d^{2}/n)}.
\]
Now~(\ref{eq:surj-small-range-0.33333}) follows by taking $m=\lceil\sqrt{3r}\rceil$
and $d=\Theta(\sqrt{n}\cdot r^{1/4})$. The approximating polynomial
in question is given explicitly because every stage of our proof,
including the appeal to Corollary~\ref{cor:CONJUNCTION},  is constructive.
\end{proof}

\subsection{\label{subsec:Approximation-to-any-eps}Approximation to arbitrary
error}

We now generalize the previous theorem to arbitrary $\epsilon$. The
proof closely mirrors the case of $\epsilon=1/3$ but features additional
ingredients, such as Lemma~\ref{lem:bound-on-poly-coeffs-UNIVARIATE}.
\begin{thm}
\label{thm:SURJ}For all positive integers $n$ and $r,$ and all
reals $0<\epsilon<1/2,$
\begin{align}
 & \deg_{\epsilon}(\SURJ_{n,r})=O\left(\sqrt{n}\left(r\log\frac{1}{\epsilon}\right)^{1/4}+\sqrt{n\log\frac{1}{\epsilon}}\right) &  & (r\leq n),\label{eq:surj-small-range}\\
 & \deg_{\epsilon}(\SURJ_{n,r})=0 &  & (r>n).\label{eq:surj-large-range}
\end{align}
Moreover, the approximating polynomial is given explicitly in each
case.
\end{thm}

\begin{proof}
As before, we need only prove (\ref{eq:surj-small-range}) since $\SURJ_{n,r}\equiv0$
for $r>n.$ Theorem~\ref{thm:AND} provides, after rescaling, an
explicit univariate polynomial $p$ such that
\begin{align}
 & p(1)=1,\label{eq:surj-general-p-at-1}\\
 & |p(t)|\leq\frac{\epsilon}{2}, &  & t\in\left\{ 0,\frac{1}{r},\frac{2}{r},\ldots,\frac{r-1}{r}\right\} ,\label{eq:surj-general-p-intermediate}\\
 & |p(t)|\leq1, &  & t\in[0,1],\label{eq:surj-general-p-bounded}\\
 & \deg p=O\left(\sqrt{r\log\frac{1}{\epsilon}}\right).\label{eq:surj-general-p-deg}
\end{align}
Now define $\widetilde{\SURJ}_{n,r}\colon\zoo_{\leq n}^{nr}\to\Re$
by
\[
\widetilde{\SURJ}_{n,r}(x)=p\left(\frac{1}{r}\sum_{j=1}^{r}\bigvee_{i=1}^{n}x_{i,j}\right).
\]
This function clearly approximates $\SURJ_{n,r}$ pointwise to $\epsilon/2.$
It follows that for any $d,$
\begin{align}
E(\SURJ_{n,r},d) & \leq\|\SURJ_{n,r}-\widetilde{\SURJ}_{n,r}\|_{\infty}+E(\widetilde{\SURJ}_{n,r},d)\nonumber \\
 & \leq\frac{\epsilon}{2}+E(\widetilde{\SURJ}_{n,r},d).\label{eq:surj-general-intermediate-error}
\end{align}

We have
\begin{align}
\orcomplexity(\widetilde{\SURJ}_{n,r}) & \leq\max\left\{ 1,\orcomplexity\!\left(\frac{1}{r}\sum_{j=1}^{r}\bigvee_{i=1}^{n}x_{i,j}\right)\right\} ^{\deg p}\norm p\nonumber \\
 & \leq2^{\deg p}\;\norm p\nonumber \\
\rule{0mm}{5mm} & \leq16^{\deg p}\nonumber \\
 & \leq2^{O\left(\sqrt{r\log(1/\epsilon)}\right)},\label{eq:surj-tilde-orcomplexity}
\end{align}
where the first and second steps use Proposition~\ref{prop:orcomplexity}~\ref{item:orcomplexity-composition-with-univariate},~\ref{item:orcomplexity-disjunction};
the third step follows from (\ref{eq:surj-general-p-bounded}) and
Lemma~\ref{lem:bound-on-poly-coeffs-UNIVARIATE}; and the final step
is valid by~(\ref{eq:surj-general-p-deg}). 

To restate (\ref{eq:surj-tilde-orcomplexity}), we have shown that
$\widetilde{\SURJ}_{n,r}$ is a linear combination of conjunctions
with real coefficients whose absolute values sum to $\exp(O(\sqrt{r\log(1/\epsilon)}))$.
By Corollary~\ref{cor:CONJUNCTION}, each of these conjunctions can
be approximated by a polynomial of degree $d$ to within  $2^{-\Theta(d^{2}/n)}$
pointwise. We conclude that
\[
E(\widetilde{\SURJ}_{n,r},d)\leq2^{O\left(\sqrt{r\log(1/\epsilon)}\right)}\cdot2^{-\Theta(d^{2}/n)},
\]
which along with (\ref{eq:surj-general-intermediate-error}) gives
\[
E(\SURJ_{n,r},d)\leq\frac{\epsilon}{2}+2^{O\left(\sqrt{r\log(1/\epsilon)}\right)}\cdot2^{-\Theta(d^{2}/n)}.
\]
Now~(\ref{eq:surj-small-range}) follows by taking 
\[
d=\Theta\left(\sqrt{n}\left(r\log\frac{1}{\epsilon}\right)^{1/4}+\sqrt{n\log\frac{1}{\epsilon}}\right).
\]
Finally, the approximating polynomial in question is given explicitly
because every stage of our proof, including the appeal to Theorem~\ref{thm:AND}
and Corollary~\ref{cor:CONJUNCTION},  is constructive.
\end{proof}

\section*{Acknowledgments}

The author is thankful to Paul Beame, Aleksandrs Belovs, Mark Bun,
Robin Kothari, Justin Thaler, Emanuele Viola, and Ronald de Wolf for
valuable comments on an earlier version of this paper. The author
is further indebted to Mark, Robin, and Justin for stimulating discussions
and for sharing a preliminary version of their manuscript~\cite{BKT17poly-strikes-back},
which inspired the title of this paper.

\bibliographystyle{siamplain}
\bibliography{refs}

\end{document}